\RequirePackage{fix-cm}
\RequirePackage{amsmath}
\documentclass[smallextended]{svjour3}       
\smartqed  
\usepackage{newtxtext,newtxmath}
\usepackage[utf8]{inputenc}
\usepackage{dsfont}
\usepackage{hyperref}
\usepackage{cite}
\usepackage{soul}

\usepackage{graphicx}

\newcommand{\tr}{\mathrm{Tr}}

\newcommand{\psc}{\mathrm{\Omega}}
\newcommand{\parity}{\mathrm{\Pi}}

\newcommand{\indLnorm}{L^2}
\newcommand{\indslnorm}{\ell^2}
\newcommand{\indsupnorm}{\mathrm{sup}}
\newcommand{\indhsnorm}{\mathrm{HS}}
\newcommand{\fockx}[1]{\psi^{\mathrm{Fock}}_{#1}(x)}

\newcommand{\al}{\hat{a}_{\lambda}}
\newcommand{\bl}{ \hat{a}^\dagger_{\lambda}}
\newcommand{\ahat}{\hat{a}}
\newcommand{\ahatdagg}{\hat{a}^\dagger}

\newcommand{\F}{\mathcal{F}}
\newcommand{\FF}{F}

\def\makeheadbox{\relax}

\newcommand*{\affaddr}[1]{#1}
\newcommand*{\affmark}[1][*]{\textsuperscript{#1}}

\journalname{Communications in Mathematical Physics}
\begin{document}

\title{Phase Spaces, Parity Operators, 
	and the Born-Jordan Distribution}	



\author{Bálint Koczor\protect\affmark[1,2,3]\,\href{https://orcid.org/0000-0002-4319-6870}{\includegraphics[height=10pt]{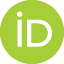}}       
\and
Frederik vom~Ende\protect\affmark[1,4]\,\href{https://orcid.org/0000-0002-2738-6893}{\includegraphics[height=10pt]{ORCID-iD_icon-64x64.png}}          \and
Maurice de Gosson\protect\affmark[5]\,\href{https://orcid.org/0000-0001-8721-1078}{\includegraphics[height=10pt]{ORCID-iD_icon-64x64.png}}
 \and
Steffen J.~Glaser\protect\affmark[1,4]\,\href{https://orcid.org/0000-0003-4099-3177}{\includegraphics[height=10pt]{ORCID-iD_icon-64x64.png}}
\and
Robert Zeier\protect\affmark[6,7]\,\href{https://orcid.org/0000-0002-2929-612X}{\includegraphics[height=10pt]{ORCID-iD_icon-64x64.png}}}

\authorrunning{Bálint Koczor et al.} 

\institute{
\affaddr{\affmark[1] Technische Universität München, Department Chemie,
              Lichtenbergstrasse 4, 85747 Garching, Germany, \email{frederik.vom-ende@tum.de, glaser@tum.de}}\\
\affaddr{\affmark[2] Munich Centre for Quantum Science and Technology (MCQST), Schellingstra{\ss}e 4, 80799 M{\"u}nchen, Germany}\\
\affaddr{\affmark[3] University of Oxford, Department of Materials, Parks Road,
			Oxford OX1 3PH, United Kingdom, \email{balint.koczor@materials.ox.ac.uk}}\\
\affaddr{\affmark[4] Munich Centre for Quantum Science and Technology (MCQST) \& 
Munich Quantum Valley (MQV), Schellingstra{\ss}e 4, 80799 M{\"u}nchen, Germany}\\
\affaddr{\affmark[5] Faculty of Mathematics (NuHAG), University of Vienna,
              Oskar-Morgenstern-Platz 1, 1090 Wien, Austria, \email{maurice.de.gosson@univie.ac.at}}\\
\affaddr{\affmark[6] Adlzreiterstrasse 23, 80337 München, Germany}\\
\affaddr{\affmark[7] Forschungszentrum Jülich GmbH, Peter Grünberg Institute,
           Quantum Control (PGI-8), 54245 Jülich, Germany, \email{r.zeier@fz-juelich.de}}}

\date{September 27, 2022}

\maketitle

\begin{abstract}
Phase spaces as given by the Wigner distribution function
provide a natural description of infinite-dimensional
quantum systems. They are an important tool in quantum optics and have been widely 
applied in the context of time-frequency analysis and pseudo-differential operators.
Phase-space distribution functions are usually specified via
integral transformations or convolutions
which can be averted and subsumed by (displaced) parity operators
proposed in this work. Building on earlier work 
for Wigner distribution functions [A.~Grossmann, Comm.\ Math.\ Phys.\ 48(3), 191 (1976)],
parity operators give rise to a general class of distribution functions
in the form of quantum-mechanical expectation values.
This enables us to precisely characterize the mathematical existence of
general phase-space distribution functions.
We then relate these distribution functions to the so-called Cohen class [L.~Cohen, J.\ Math.\ Phys. 7(5), 781 (1966)]
and recover various quantization schemes and distribution functions 
from the literature.
The parity-operator approach is also applied to the Born-Jordan
distribution which originates from the Born-Jordan quantization
[M.~Born, P.~Jordan, Z.\ Phys.\ 34(1), 858 (1925)].
The corresponding parity operator is written as a weighted average of both displacements
and squeezing operators and we determine its generalized spectral decomposition.
This leads to an efficient computation of the Born-Jordan parity operator 
in the number-state basis and example quantum states reveal unique features of the Born-Jordan
distribution.\vfill
\end{abstract}

\section{Introduction \label{intro}}

There are at least three logically independent descriptions of quantum mechanics:
the Hilbert-space formalism \cite{cohen1991quantum}, 
the path-integral method \cite{feynman2005}, and the 
phase-space approach such as given by the Wigner function
\cite{carruthers1983,hillery1997,kim1991,lee1995,gadella1995,zachos2005,schroeck2013,SchleichBook,Curtright-review}.
The phase-space formulation of quantum mechanics was initiated by Wigner in his
ground-breaking work \cite{wigner1932} from 1932, in which the Wigner function
of a spinless  non-relativistic quantum particle was introduced as
a quasi-probability distribution. 
The Wigner function can be used to express quantum-mechanical expectation values as classical
phase-space averages. More than a decade later, Groenewold
\cite{Gro46} and Moyal \cite{Moy49} formulated quantum mechanics
as a statistical theory on a classical phase by mapping a quantum state
to its Wigner function and they interpreted this correspondence 
as the inverse of the Weyl quantization \cite{Wey27,Weyl31,Weyl50}.

Coherent states have become a natural way to extend phase spaces to more general
physical systems \cite{1bayen1978,2bayen1978,berezin74,berezin75,brif98,perelomov2012,gazeau,ali2000,BERGERON2013}.
In this regard, a new focus on phase-space representations for
coupled, finite-dimensional quantum systems (as spin systems) 
\cite{rundle2021overview,thesis,DROPS,koczor2016,koczor2017,koczor2018,KdG10,tilma2016,RTD17,LZG18,koczor2021,koczor2020}
and their tomographic reconstructions \cite{rundle2017,leiner17,Leiner18,koczor2017}
has emerged recently.
A spherical phase-space representation of a single, finite-dimensional quantum system
has been used to naturally recover the infinite-dimensional phase space in the
large-spin limit \cite{koczor2017,koczor2018}.
These spherical phase spaces have been defined in terms of quantum-mechanical
expectation values of rotated parity operators \cite{tilma2016,RTD17,LZG18,koczor2017,koczor2018,thesis}
(as discussed below)
in analogy with displaced reflection operators in flat phase spaces.
But in the current
work, we exclusively focus on the (usual) infinite-dimensional case 
which has Heisenberg-Weyl symmetries
\cite{brif98,perelomov2012,gazeau,LI94}.
This case has been playing a
crucial role in characterizing the quantum theory of light \cite{glauber2006nobel}
via coherent states and displacement operators \cite{Cahill68,cahill1969,agarwal1970,Agarwal68}
and has also been
widely used in the context of time-frequency analysis and pseudo-differential operators
\cite{cohen1966generalized,Cohen95,boggiatto2010time,
boggiatto2010weighted,boggiatto2013hudson,bornjordan,thewignertransform,groechenig2001foundations}.
Many particular phase spaces have been unified under the concept
of the so-called Cohen class \cite{cohen1966generalized,Cohen95,thewignertransform}
(see Definition~\ref{defofcohenclass} below),
i.e.~all functions which are related to the
Wigner function via a convolution with a distribution 
(which is also known as the Cohen kernel).

Phase-space distribution functions
are mostly described by one of the following three forms:
(a) convolved derivatives of the Wigner function \cite{bornjordan,thewignertransform},
(b) integral transformations of a pure state (i.e.~a rapidly decaying, complex-valued function)
\cite{wigner1932,cohen1966generalized,Cohen95,boggiatto2010time,
	boggiatto2010weighted,boggiatto2013hudson,bornjordan,thewignertransform},
or (c) as integral transformations of quantum-mechanical expectation values \cite{Cahill68,cahill1969,agarwal1970,Agarwal68}.
Also, Wigner functions (and the corresponding Weyl quantization)
are usually described by integral transformations.
But the seminal work of Grossmann \cite{Grossmann1976,thewignertransform}
(refer also to \cite{Royer77}) allowed for a direct interpretation of the Wigner function 
as a quantum-mechanical expectation value of a displaced parity operator
$\parity$ (which reflects coordinates $\parity\, \psi(x) = \psi(-x)$ of a quantum state $\psi$).
In particular, Grossmann \cite{Grossmann1976}
showed that the Weyl quantization of the delta distribution determines the parity operator $\parity$.
This approach has been widely adopted
\cite{dahl1982group,dahl1988morse,LIParity,bishop1994,gadella1995,Royer1989,Royer96,chountasis1999}.

However, parity operators similar to the one
by Grossmann and Royer \cite{Grossmann1976,Royer77,thewignertransform}
have still been lacking for general phase-space distribution 
functions. (Note that such a form appeared implicitly for $s$-parametrized 
distribution functions in \cite{moya1993,cahill1969}.)
In the current work, we generalize 
the previously discussed parity operator $\parity$ \cite{Grossmann1976,Royer77,thewignertransform}
for the Wigner function by introducing a family of parity operators $\parity_\theta$ 
(refer to Definition~\ref{def_parity_operator})
which is parametrized 
by a function or distribution $\theta$.
This enables us to
specify general phase-space distribution functions
in the form of quantum-mechanical expectation values
(refer to Definition~\ref{def_phase_space_function}) as
\begin{equation*}
\FF_\rho(\psc,\theta)
:= (\pi \hbar)^{-1} \, \tr\,[ \, \rho \, \mathcal{D}(\psc)  \parity_\theta  \mathcal{D}^\dagger(\psc) ].
\end{equation*}
We will refer to the above operator $\parity_\theta $ as a parity operator
	following the lead of Grossmann and Royer \cite{Grossmann1976,Royer77}
	and
	given its resemblance
	and close analogy to the reflection operator $\parity$ discussed in prior work 
	\cite{bishop1994,tilma2016,RTD17,LZG18,koczor2017,koczor2018,thesis}.
Here, $\mathcal{D}(\psc)$ denotes the displacement operator and $\Omega$ describes 
suitable phase-space coordinates (see Sec.~\ref{translations}). 
(Recall that $\hbar = h / (2\pi)$ is defined as the Planck constant $h$ divided by $2\pi$.)
The quantum-mechanical
expectation values  in the preceding equation
give rise to a rich family of phase-space distribution functions $\FF_\rho(\psc,\theta)$
which represent arbitrary (mixed) quantum states as given by their
density operator $\rho$. In particular,
this family of phase-space representations contains all elements from the 
(above mentioned) Cohen class 
and naturally includes the pivotal Husimi Q
and Born-Jordan distribution functions.

We would like to emphasize that our approach to phase-space representations
averts the use of integral transformations, Fourier transforms, or convolutions
as these are subsumed in the parity operator $\parity_\theta$ which is independent
of the phase-space coordinate $\Omega$.
Although our definition also relies on an integral transformation
given by a Fourier transform, it
is only applied once and is completely absorbed into the definition of a parity
operator thereby avoiding redundant applications of Fourier transforms.
This leads to significant advantages
as compared to earlier approaches:
\begin{itemize}
	\item conceptual advantages (see also \cite{moya1993,KdG10,tilma2016,koczor2017,RTD17}):
	\begin{itemize}
	\item The phase-space distribution function is given as a quantum-mechanical 
	expectation value.
	And this form nicely fits with the experimental reconstruction of quantum states
	\cite{Banaszek99,heiss2000discrete,Bertet02,deleglise2008,rundle2017,koczor2017,Lutterbach97}.
	\item All the complexity from integral transformations (etc.) is condensed into
	the parity operator $\parity_\theta$.
	\item The dependence on the distribution $\theta$ and the particular phase space
	is separated from the displacement $\mathcal{D}(\psc)$. 
	\end{itemize}
	\item computational advantages:
	\begin{itemize}
	\item The repeated and expensive computation of integral transformations (etc.) in earlier approaches is avoided as
	$\parity_\theta$ has to be determined only once. Also,
	the effect of the displacement $\mathcal{D}(\psc)$ is relatively easy to calculate.
	\end{itemize}
\end{itemize}
In this regard, the current work can also be seen
as a continuation of \cite{koczor2017} where the parity-operator approach has been emphasized,
but mostly for finite-dimensional quantum systems.
Moreover, we connect  
results from quantum optics \cite{cahill1969,Cahill68,glauber2006nobel,leonhardt97},
quantum-harmonic analysis \cite{werner1984,thewignertransform,bornjordan,daubechies1980distributions,
daubechies1980coherent,daubechies1980I,daubechies1983II,keyl2016,cohen2012weyl}, and
group-theoretical approaches \cite{brif98,perelomov2012,gazeau,LI94}.
It is also our aim to narrow the gap between different communities
where phase-space methods have been successfully applied.

On the other hand, a major contribution of our work is the analysis of
existence properties of generalized phase-space distributions and their parity operators.
While the Wigner function has been known to exist for the general class of tempered distributions
(a class of generalized functions that includes the pivotal $L^2$ space), we further illuminate
which classes of Cohen kernels yield well-defined generalized phase-space distribution functions.
Such existence questions are 
fully absorbed into the parity operators and precise conditions are used to guarantee their mathematical
existence.

Similarly as the parity operator $\parity$ (which
is the Weyl quantization of the delta distribution),
we show that its generalizations $\parity_\theta$ are
Weyl quantizations of the corresponding Cohen kernel $\theta$
(refer to Sec.~\ref{reltoquant} for the precise definition 
of the Weyl quantization used in this work).
We discuss
how these general results reduce to well-known special cases
and discuss properties of phase-space distributions in relation
to their parity operators $\parity_\theta$.
In particular, we consider the class of $s$-parametrized
distribution functions \cite{glauber2006nobel,cahill1969,Cahill68,moya1993},
which 
include the Wigner, Glauber P, and
Husimi Q functions, as well as the $\tau$-parametrized family, 
which has been proposed in the context of time-frequency analysis and pseudo-differential operators
\cite{bornjordan,boggiatto2010time,boggiatto2010weighted,boggiatto2013hudson}.
We derive spectral decompositions of parity operators for all of these phase-space families,
including the Born-Jordan distribution.
Relations of the form $\parity_{\theta}=A_\theta \circ \parity$
motivate the name ``parity operator''
as they are in fact compositions of the usual parity operator $\parity$ followed by some operator
$A_\theta$ that usually corresponds to a geometric or physical operation (which commutes with $\parity$).
In particular, $A_\theta$ is a squeezing operator for the $\tau$-parametrized family and corresponds to photon
loss for the $s$-parametrized family (assuming $s<0$).
This structure of the parity operators $\parity_{\theta}$ connects phase spaces to
elementary geometric and physical 
operations (such as  reflection, squeezing operators, photon loss)
and these concepts are central to applications:
the squeezing operator models a non-linear optical process which
generates non-classical states of light in quantum optics
\cite{mandel1995,glauber2007,leonhardt97}.
These squeezed states of light have been widely used
in precision interferometry \cite{schnabel2017,Grangier87,Xiao87,McKenzie2002}
or for enhancing the performance of imaging \cite{lugiato2002,treps2003},
and the gravitational-wave detector GEO600 has been operating
with squeezed light since 2010 \cite{abadie2011,grote2013}.

The Born-Jordan distribution and its parity operator
constitute a most peculiar instance among the phase-space approaches.
This distribution function has convenient properties, e.g., 
it satisfies the marginal conditions and therefore allows for
a probabilistic interpretation \cite{bornjordan}.
The Born-Jordan distribution is however difficult to compute. 
But most importantly, the Born-Jordan
distribution and its corresponding quantization scheme have a
fundamental importance in quantum mechanics.
In particular, there have been several attempts in the literature to find the
``right'' quantization rule for observables
using either algebraic or analytical techniques. In a recent paper \cite{FPBJ},
one of us has analyzed the Heisenberg and Schr\"{o}dinger pictures of quantum
mechanics, and it is shown that the equivalence of both theories requires that one
must use the Born--Jordan quantization rule (as proposed by Born and Jordan \cite{born25})
\begin{equation*}
\text{(BJ)} \quad
x^{m}p^{\ell} \mapsto \frac{1}{m{+}1}
\sum_{k=0}^{m}
\hat{x}^{k} \hat{p}^{\ell} \hat{x}^{m-k},
\end{equation*}
instead of the Weyl rule
\begin{equation*}
\text{(Weyl) } \quad
x^{m} p^{\ell} \mapsto \frac{1}{2^{m}}
\sum_{k=0}^{m}
\binom{m}{k} \hat{x}^{k} \hat{p}^{\ell} \hat{x}^{m-k}
\end{equation*}
for monomial observables. The Born--Jordan and Weyl rules yield the same
result only if $m<2$ or $\ell<2$; for instance in both cases the quantization
of the product $xp$ is $\frac{1}{2}(\hat{x}\hat{p} + \hat{p}\hat{x})$.
It is however easy to find physical examples which give
different results. Consider for instance the square of 
the $z$ component of 
the angular momentum: it is given by
\begin{equation*}
\ell_{z}^{2}=x^{2}p_{y}^{2}+y^{2}p_{x}^{2}-2xp_{x}yp_{y}\label{lx2}
\end{equation*}
and its Weyl quantization is easily seen to be
\begin{equation} \label{lx2w}
\textup{Op}_{\textup{Weyl}}(\ell_{z}^{2})
=\hat{x}_{x}^{2} \hat{p}_{y}^{2}
+ \hat{x}_{y}^{2} \hat{p}_{x}^{2}
- \tfrac{1}{2}(\hat{x}_{x} \hat{p}_{x} {+} \hat{p}_{x} \hat{x}_{x} )
(\hat{x}_{y} \hat{p}_{y} {+} \hat{p}_{y} \hat{x}_{y})
\end{equation}
while its Born--Jordan quantization is the different expression
\begin{equation} \label{lx2bj}
\textup{Op}_{\textup{BJ}}(\ell_{z}^{2})
=\hat{x}_{x}^{2} \hat{p}_{y}^{2}
+ \hat{x}_{y}^{2} \hat{p}_{x}^{2}
- \tfrac{1}{2}(\hat{x}_{x} \hat{p}_{x} {+} \hat{p}_{x} \hat{x}_{x} )
( \hat{x}_{y} \hat{p}_{y} {+} \hat{p}_{y} \hat{x}_{y} ) - \tfrac{1}{6}\hbar^{2}.
\end{equation}
(Recall that the operators $\hat{x}_\eta$ and $\hat{p}_\kappa$ satisfy the canonical commutation relations
$[\hat{x}_\eta,\hat{p}_\kappa] = i\hbar \delta_{\eta \kappa}$ using the spatial coordinates 
$\eta,\kappa\in\{x,y,z\}$ and the Kronecker delta $\delta_{\eta \kappa}$.)
One of us has shown in \cite{dilemma} that the use of (\ref{lx2bj}) instead of
(\ref{lx2w}) solves the so-called ``angular momentum
dilemma'' \cite{dahl1,dahl2}. To a general observable $a(x,p)$, the
Weyl rule associates the operator
\begin{equation*}
\textup{Op}_{\textup{Weyl}}(a)
=(2\pi\hbar)^{-1}\int \F_{\sigma}a(x,p)
\mathcal{D}(x,p) \, \mathrm{d}x \, \mathrm{d}p
\end{equation*}
where $\F_{\sigma}a$ is the symplectic Fourier transform of $a$ and
$\mathcal{D}(x,p)$ the displacement operator (see Sec.~\ref{translations}); in the
Born--Jordan case this expression is replaced with
\begin{equation*}
\textup{Op}_{\textup{BJ}}(a)
=(2\pi\hbar)^{-1}\int \F_{\sigma}a(x,p) K_{\textup{BJ}}(x,p)
\mathcal{D}(x,p) \, \mathrm{d}x \, \mathrm{d}p
\end{equation*}
where the filter function $K_{\textup{BJ}}(x,p)$ 
is given by
\[
K_{\textup{BJ}}(x,p)=\mathrm{sinc}( \tfrac{ px }{2\hbar} ) = \frac{\sin(\tfrac{px}{2\hbar})}{px/(2 \hbar)}.
\]

We obtain significant, new results for the case of
Born-Jordan distributions and therefore substantially advance on previous
characterizations.
In particular, we derive its parity operator 
$\parity_{\textup{BJ}}$ in the form of a weighted average
of geometric transformations
\begin{equation} \label{eq:intro_bornjordan}
\parity_{\textup{BJ}} = 
\tfrac{1}{4 \pi \hbar} \int \mathrm{sinc}( \tfrac{ px }{2\hbar} ) \, \mathcal{D}(x,p) \, \mathrm{d}x \, \mathrm{d}p
=
[\,
\tfrac{1}{4}	 \int_{-\infty}^{\infty}
\,
\mathrm{sech}(\tfrac{\xi}{2})
S(\xi) 
\,\mathrm{d} \xi
\,] 
\; \parity,
\end{equation}
where $\mathcal{D}(x,p)$ is the displacement operator
and $S(\xi)$ is the squeezing operator (see Eq.~\eqref{sq_eq} below)
with a real squeezing parameter $\xi$. 
We have used the \emph{sinus cardinalis}
$\mathrm{sinc}(x):= \mathrm{sin}(x)/x$
and the \emph{hyperbolic secant}
$ \mathrm{sech}(x):=1/\mathrm{cosh}(x)$ functions.
The parity operator $\parity_{\textup{BJ}}$ in Eq.~\eqref{eq:intro_bornjordan}
decomposes into a product $\parity_{\theta}=A_\theta \circ \parity$ containing 
the usual reflection operator $\parity$. This is another example of the above-discussed 
motivation for our terminology of parity operators.
We prove in Proposition \ref{BJparityBound} that $\parity_{\textup{BJ}}$
is a bounded operator on the Hilbert space of square-integrable functions
and therefore gives rise to well-defined
phase-space distribution functions of arbitrary quantum states.
We derive a generalized spectral decomposition of this parity operator
based on a continuous family of generalized eigenvectors that satisfy
the following generalized eigenvalue equation for every real $E$
(see Theorem \ref{BJdecomposition}):
\begin{equation*}
\parity_{\textup{BJ}} \, | \psi^E_{\pm} \rangle	=  \pm   \pi/2  \, \mathrm{sech}(\pi E)  \,  | \psi^E_{\pm} \rangle.
\end{equation*}
Facilitating a more efficient computation of the Born-Jordan distribution,
we finally derive explicit matrix representations in the so-called
Fock or number-state basis, which constitutes a natural 
representation for bosonic quantum systems such as in quantum optics
\cite{mandel1995,glauber2007,leonhardt97}.
Curiously, the parity operator $\parity_{\textup{BJ}}$
of the Born-Jordan distribution is not diagonal
in the Fock basis as compared to the diagonal parity operators
of $s$-parametrized phase spaces (cf.\  \cite{koczor2017}) that
enable the experimental reconstruction of distribution
functions from photon-count statistics \cite{deleglise2008,Lutterbach97,Bertet02,Banaszek99}
in quantum optics.
We calculate the matrix elements $[\parity_{\textup{BJ}}]_{mn}$ in the Fock or number-state basis 
and provide a convenient formula for a direct recursion, for which we conjecture that the matrix elements  
are completely determined by eight rational initial values. 
This recursion scheme has significant computational advantages
for calculating Born-Jordan distribution functions as compared to previous approaches
and allows for an efficient implementation. In particular, large matrix representations of
the parity operator $\parity_{\textup{BJ}}$ can be well approximated using rank-9
matrices.
We finally illustrate our results for simple
quantum states by calculating their Born-Jordan distributions
and comparing them to other phase-space representations.
Let us summarize the main results of the current work:
\begin{itemize}
	\item quantum-mechanical expectation
	values of the parity operators $\parity_\theta$ 
	from Def.~\ref{def_parity_operator} define
	distribution functions (see Def.~\ref{def_phase_space_function}) and form
	the Cohen class
	(Theorem \ref{convolutionproposition});
	\item existence properties of parity operators and generalized phase-space functions
	are clarified in Sec.~\ref{parityopquantizationsec} and we refer in particular to 
	the crucial Lemma~\ref{lemma_bounded_op_sufficient};
	\item the parity operators $\parity_\theta$ are Weyl quantizations
	of the corresponding Cohen convolution kernels $\theta$ (Sec.~\ref{reltoquant});
	\item parity operators for important distribution functions are summarized in Sec.~\ref{explicitparitysection}
	along with their operator norms (Theorem \ref{tauparitycoordinate}) and 
	generalized spectral decompositions in Sec~\ref{BJdecomp};
	\item the Born-Jordan parity operator is a weighted
	average of displacements (Theorem \ref{BJdistributionTheorem}) or
	equivalently a weighted average of squeezing operators
	(Theorem \ref{BJparitySqueezing}), and it is bounded
	(Proposition \ref{BJparityBound});
	\item the Born-Jordan parity operator admits a generalized spectral decomposition
	(Theorem \ref{BJdecomposition});
	\item its matrix representation is calculated in the number-state basis
	in Theorem \ref{BJParityMatrix};
	and an efficient, recursion-based computation scheme is proposed in
	Conjecture~\ref{BJParityMatrixRecursion}.
\end{itemize}

Our work has significant implications: 
General (infinite-dimensional) phase-space functions can now be 
conveniently and effectively described as natural expectation values.
We provide a much more comprehensive understanding 
of Born-Jordan phase spaces and means for effectively
computing the corresponding phase-space functions. Working
in a rigorous mathematical framework, we also facilitate future discussions 
of phase spaces by connecting different communities in physics and mathematics.

We start by  recalling precise  definitions of distribution functions and quantum states
for infinite-dimensional
Hilbert spaces in Sec.~\ref{distributions_qs}.
In Sec.~\ref{phasespace}, we discuss phase-space translations
of quantum states using coherent states, state one known formulation of
translated parity operators, and relate a general class of phase spaces
to Wigner distribution functions and their properties.
We note that an experienced reader can skip
most of the introductory Sections~\ref{distributions_qs}
and \ref{phasespace} and jump directly to our results.
These preparations will however guide our study
of phase-space representations of quantum states as expectation
values of displaced parity operators in Section~\ref{parityopquantizationsec}.
We present and discuss our results for the case of the Born-Jordan distribution
and its parity operator in Section~\ref{bornjordansection}.
Formulas for the matrix elements of the Born-Jordan parity operator
are derived in Section~\ref{BJmatrix}. 
Explicit examples for simple quantum systems are discussed and visualized in
Section~\ref{sec_ex}, before we conclude. A larger part of the proofs have been relegated
to Appendices.

\section{Distributions and Quantum States \label{distributions_qs}}

All of our discussion and results in this work will strongly rely on precise notions
of distributions and related descriptions of quantum states 
in infinite-dimensional Hilbert spaces. Although most (or all) of this material is quite
standard and well-known \cite{ReedSimon1,kanwal2012generalized,thewignertransform,hall2013quantum},
we find it prudent to shortly summarize this background material
in order to fix our notation and keep our presentation self-contained.
This will also help to clarify differences and connections 
between divergent concepts and notations used in the literature.
We hope this will also contribute to narrowing the gap between 
different physics communities that are interested in this topic.

\subsection{Schwartz Space and Fourier Transforms \label{functionfourier}}

We will now summarize function spaces that are central
to this work, refer also to \cite[Ch.~1.1.3]{thewignertransform}.
The set of all smooth, complex-valued functions on $\mathbb R^n$ that decrease faster
(together with all of their partial derivatives)
than the reciprocal of any polynomial is called the Schwartz space and is usually denoted by
$\mathcal{S}(\mathbb{R}^n)$, refer to \cite[Ch.~V.3]{ReedSimon1}
or \cite[Ch.~6]{kanwal2012generalized}.
More precisely, a function $\psi:\mathbb R^n\to\mathbb C$ is called fast decreasing if the absolute values
$| x^{\beta} \partial_x^{\alpha} \psi(x) |$ are bounded 
for each multi-index of natural numbers $\alpha:=(\alpha_1, \dots, \alpha_n)$ and $\beta:=(\beta_1, \dots, \beta_n)$,
where by definition $x^{\beta}:=x_1^{\beta_1} \cdots x_n^{\beta_n}$ and 
$\partial_x^{\alpha}:=\partial_{x_1}^{\alpha_1} \cdots \partial_{x_n}^{\alpha_n}$,
refer to \cite[Ch.~1.1.3]{thewignertransform}. This gives rise to a family of seminorms 
$\|\psi\|_{\alpha,\beta}:=\sup_{x\in\mathbb R^n}|x^\beta\partial_x^\alpha\psi(x)|$ which turn $\mathcal S(\mathbb R^n)$ 
into a topological space which is even a Fréchet space \cite[Thm.~V.9]{ReedSimon1}.

The topological dual space $\mathcal S'(\mathbb R^n)$ of $\mathcal S(\mathbb R^n)$ is often referred to as the space of
tempered distributions, and we will
denote the distributional pairing for $\phi \in \mathcal S'(\mathbb R^n)$
and $\psi \in \mathcal S(\mathbb R^n)$ as $\langle \phi , \psi\rangle:=\phi(\psi)\in\mathbb C$.
In Sec.~\ref{distributions_qs}, we will consistently use the symbol $\phi$ to denote distributions
and $\psi,\psi'$ to denote Schwartz or square-integrable functions.
Also note that $\mathcal{S}(\mathbb{R}^n)$ is dense in $ L^2 (\mathbb{R}^n)$
and tempered distributions naturally include the usual function spaces
$\mathcal{S}(\mathbb{R}^n)\subset L^2 (\mathbb{R}^n) \subset \mathcal S'(\mathbb R^n)$
via distributional pairings in the form of an integral
$ \langle \phi , \psi\rangle = \int_{\mathbb R^n}\phi^*(x)\psi(x)\,\mathrm{d}x$, 
where $\phi^*(x)$ is the complex conjugate of $\phi(x) \in L^2(\mathbb{R}^n)$ or $\phi(x) \in \mathcal S(\mathbb{R}^n)$.
This inclusion is usually referred to as a rigged Hilbert space
\cite{vilenkin1964generalized,chruscinski2003} or the Gelfand triple. Recall that the Lebesgue spaces 
$L^q({\mathbb R^n})$ with $0<q<\infty$ are subspaces of equivalence classes of measurable functions $f:\mathbb R^n\to\mathbb C$ 
that differ only on a set of measure zero
such that 
the $q$-th power of their absolute value is  Lebesgue integrable, i.e.~$\int_{\mathbb R^n} |f(x)|^q \,\mathrm{d}x < \infty$ \cite{ReedSimon1}.

Remarkably, every tempered distribution is the derivative of some polynomially bounded continuous function, 
that is, given $\phi\in\mathcal S'(\mathbb R^n)$ there exists $g:\mathbb R^n\to\mathbb C$ continuous such that 
$|g(x)|\leq C(1{+}x^2)^m$ for some $C,m\geq 0$ and all $x\in\mathbb R^n$, as well as a multi-index $\alpha$ such that
$
\langle\phi,\psi\rangle =(-1)^{|\alpha|}\int_{\mathbb R^n}g^*(x) (\partial_x^\alpha\psi)(x)\,\mathrm{d}x
$
for all $\psi\in\mathcal S(\mathbb R^n)$---for short one can write $\phi=\partial_x^\alpha g$ \cite[Thm.~V.10]{ReedSimon1}.

In particular one can construct tempered distributions by considering smooth functions $\phi$
that (together with all of their partial derivatives)
grow slower than certain polynomials.
More precisely, a smooth map $\phi:\mathbb R^n\to\mathbb C$ is said to be slowly increasing or of slow growth if there exist
for every $\alpha=(\alpha_1, \dots, \alpha_n)$
constants $C$, $m$, and $A$ such that 
$|\partial_x^\alpha
\phi(x)|\leq C\|x\|^m$ 
for all $\|x\|>A$, 
where $\| x \|$ is the Euclidean norm in $\mathbb R^n$,
refer to \cite[Ch.~6.2]{kanwal2012generalized}.
A classical example of
such functions are polynomials.
In particular, every slowly increasing function $\phi(x)$ generates a tempered distribution
$\langle \phi , \psi \rangle  = \int_{\mathbb R^n} \phi^*(x) \psi(x)\,\mathrm{d}x$ for all
$\psi\in\mathcal S(\mathbb R^n)$,
and, therefore, such functions are usually denoted as $\phi(x) \in \mathcal{S}'(\mathbb{R}^n)$
(refer to \cite[Ch.~6.2]{kanwal2012generalized}).
\begin{example}
This motivates the delta distribution
$\langle  \delta_b , \psi \rangle := \psi(b)$ which is in its integral representation commonly
written as  $ \int_{\mathbb R^n} \delta(x{-}b) \psi(x) \,\mathrm{d}x = \psi(b) $. We emphasize that
the notation $\delta(x)$  is however only formal, cf.~\cite[Eq.~(V.3)]{ReedSimon1}. 
Moreover this tempered distribution is generated by the second derivative of the polynomially bounded 
continuous function $g(x):=x{-}b$ for $x\geq b$ and zero otherwise, 
i.e.~$\langle\delta_b,\psi\rangle=\int_\mathbb Rg(x)\psi''(x)\,\mathrm{d}x$ for all $\psi\in\mathcal S(\mathbb R)$  
\cite[Ch.~V, Ex.~8]{ReedSimon1}. But this generating function is not unique as, for example, 
one also has $\delta_b=({d^2}/{\mathrm{d}x^2})|x{-}b|/2$.
\end{example}

For the rest of our work, we will restrict the general space
of $\mathbb R^n$ with $n\geq 1$ to the case of $\mathbb R$ which is most relevant
for the applications we highlight. This simplifies our notation, even though many statements
could be generalized.

Recall that for all $a\in\mathcal S(\mathbb R^2)$ the symplectic Fourier transform $[\F_{\sigma}a](x,p)$
(see App.~B in \cite{thewignertransform})
is related to the usual Fourier
transform\footnote{A multitude of sign and normalization conventions
are commonly used throughout various fields
as characterized by the two parameters $q$ and $r$ in the generic
expression for the one-dimensional Fourier transform  $\F[a(\cdot)](x) =  \sqrt{|r| (2\pi)^{q-1}}
\int e^{i r x x'} a(x') \, \mathrm{d}x'$. In this work, $q=0$ (because then 
$\F[\F[a]](x)=a(-x)$ for all $a\in\mathcal S(\mathbb R)$) and $r=-\hbar^{-1}$ .}
\begin{equation*}
[\F a](x,p):=
(2 \pi \hbar)^{-1}\int e^{- \tfrac{i}{\hbar} (
	x'x + p'p
	)} a(x',p') \, \mathrm{d}x'\, \mathrm{d}p',
\end{equation*}
up to a coordinate transformation $[\F_{\sigma}a](x,p) = [\F a](p,-x)$ where
\begin{equation} \label{sympfourier}
[\F_{\sigma} a](x,p):=
(2 \pi \hbar)^{-1}\int e^{-\tfrac{i}{\hbar} (x'p - xp')} a(x',p') \, \mathrm{d}x'\, \mathrm{d}p'.
\end{equation}
Note that the square $[\F_{\sigma} \F_{\sigma} a](x,p) =  a(x,p)$ is equal to the identity, and
that the Fourier transform of every function in $\mathcal{S}(\mathbb{R}^n)$ is also
in $\mathcal{S}(\mathbb{R}^n)$, cf.~\cite[Ch.~IX.1]{ReedSimon1}.
The fact that $\F_\sigma$ is hermitian, 
i.e.~$\langle \F_\sigma\phi,\psi\rangle_{L^2}=\langle\phi,\F_\sigma\psi\rangle_{L^2}$ for all $\phi,\psi\in \mathcal S(\mathbb R^2)$ 
(see Sec.~\ref{quantumstates})
motivates us to define the symplectic Fourier
transform of tempered distributions via the distributional pairing
$\langle \F_{\sigma} \phi , \psi\rangle := \langle  \phi, \F_{\sigma} \psi \rangle = \phi(\F_{\sigma} \psi)$
for $\phi \in S'(\mathbb R^2)$ and $\psi\in \mathcal S(\mathbb R^2)$.
Thus this is the extension of $\F_\sigma$ with respect to the distributional pairing in our sense, cf.~also Appendix \ref{app_operator_extension}.
In particular the symplectic Fourier transform generalizes to phase-space
distribution functions $a(x,p)$ without further adjustment and all the properties
of $\F_\sigma$ on $\mathcal S(\mathbb R^2)$ transfer to $\mathcal S'(\mathbb R^2)$.

Let us come back to our previous example: the delta distribution can be identified formally via the brackets
$ \langle  \delta_0 , \F_{\sigma} \psi \rangle = [\F_{\sigma} \psi](0) =(2\pi\hbar)^{-1}  \langle  1 , \psi \rangle $
as the Fourier transform $\delta(x) = (2 \pi \hbar)^{-1} \F_{\sigma}[1] $ 
of the constant function, refer to \cite[Ch.~6.4]{kanwal2012generalized}.

\subsection{Quantum States and Expectation Values \label{quantumstates}}

Let us denote the abstract state vector of a quantum system by $| \psi \rangle$
which is an element of an abstract, infinite-dimensional, separable complex
Hilbert space (here and henceforth denoted by) $\mathcal H$.
The Hilbert space  $\mathcal H$
is known as the state space and it is
equipped with a
scalar product $\langle \,\cdot\,|\,\cdot\, \rangle$ \cite{hall2013quantum}. Considering projectors
$\mathcal{P}_{\psi}:=| \psi  \rangle \langle \psi |$ defined via
the open scalar products $\mathcal{P}_{\psi} = \langle \psi | \,\cdot\, \rangle  \, | \psi  \rangle$, an orthonormal
basis of $\mathcal H$ is given by $\{ | \phi_n \rangle, n\in \mathbb{N} \} $ if
$ \langle \phi_n | \phi_m \rangle = \delta_{nm}$ for all $m,n\in\mathbb N$ and 
$\sum_{n = 0}^\infty \mathcal{P}_{\phi_n} = \mathds{1}$ in the strong operator topology. 
For a broader introduction to this topic we refer to \cite{hall2013quantum}.

Depending on the given quantum system, explicit representations of the state space can be obtained
by specifying its Hilbert space \cite{gieres2000}. In the case of bosonic systems, the 
Fock (or number-state) representation is widely used. A quantum state $| \psi \rangle$ is
an element of the Hilbert space $\ell^2$ of square-summable sequences
of complex numbers \cite{hall2013quantum}. It is
characterized by its expansion $
| \psi \rangle = \sum_{n=0}^{\infty} \psi_n | n \rangle$
 into the orthonormal Fock basis
 $\{ | n \rangle, n=0,1, \dots \}$ of number states
using the expansion coefficients
$\psi_n = \langle n | \psi \rangle \in \mathbb{C}$,
refer to, e.g., \cite{Cahill68} and \cite[Ch.~11]{hall2013quantum}.
The scalar product $\langle  \psi | \psi' \rangle$ then corresponds to the
usual scalar product of vectors, i.e.~to the absolutely convergent sum
$\sum_{n=0}^{\infty} (\psi_n)^* \psi'_n =:  \langle \psi | \psi' \rangle_{\indslnorm}$.
The corresponding norm of vectors is then given by 
$\|\psi \|_{\indslnorm} = 
\| ( | \psi \rangle ) \|_{\indslnorm} = [\langle \psi | \psi \rangle_{\indslnorm}]^{1/2}$.

For a quantum state $| \psi \rangle$, the coordinate representation 
$\psi(x)\in \mathcal{S}(\mathbb{R})$ and
its Fourier transform (or momentum representation)
$\psi(p)\in \mathcal{S}(\mathbb{R})$
are given by complex, square-integrable,
and
smooth functions
that are also fast decreasing.
The quantum state 
$| \psi \rangle = \int_{\mathbb{R}} \psi(x) | x \rangle \, \mathrm{d}x$ of
$\psi(x) = \langle x | \psi \rangle$
is then defined via coordinate 
eigenstates\footnote{For the position operator $\hat {x}:\mathcal S(\mathbb R)\to \mathcal S(\mathbb R), 
\psi(x)\mapsto x\psi(x)$ one can consider the dual
$\hat x': \mathcal S'(\mathbb R)\to \mathcal S'(\mathbb R), \phi\mapsto\phi\circ\hat x$.
This map satisfies the generalized eigenvalue equation
$ \hat{x}' |x_0\rangle = x_0 |x_0\rangle$ for all $x_0\in\mathbb R$ where its
generalized eigenvector $|x_0\rangle \in \mathcal S'(\mathbb R)$ is the delta distribution,
which allows for the resolution of the position operator
$\hat{x} = \int_{\mathbb{R}} x | x \rangle \langle x |    \, \mathrm{d}x$. 
For more details, we refer to \cite{vilenkin1964generalized} or \cite[p.1906]{gieres2000}.}
$ | x \rangle $.
The coordinate representation of a coordinate eigenstate is given by the distribution
$\delta(x'{-}x) \in \mathcal{S}'(\mathbb{R})$, refer to \cite{hall2013quantum,gieres2000}.
The scalar product $\langle  \psi | \psi' \rangle$
is then fixed by
the usual $L^2$ scalar product, i.e.~by the convergent integral
$\int_{\mathbb{R}}\psi^*(x)  \psi'(x) \, \mathrm{d}x =: \langle \psi | \psi' \rangle_{\indLnorm}$.
This integral induces the norm of square-integrable functions via
$\|\psi(x)\|_{\indLnorm} = [\langle \psi | \psi \rangle_{\indLnorm}]^{1/2}$.

The above two examples are particular representations of the state space,
which are convenient for particular physical systems, however these representations
are equivalent via
\begin{equation}
\mathcal{H} \simeq \ell^2 \simeq  L^2(\mathbb{R},\mathrm{d}x) \simeq  L^2(\mathbb{R},\mathrm{d}p),
\end{equation}
refer to Theorem 2 in \cite{gieres2000}.
In particular, any coordinate representation $\psi(x)\in \mathcal{S}(\mathbb{R})$
of a quantum state
$| \psi \rangle$ can be expanded in the number-state basis
into $\psi(x) = \sum_{n=1}^{\infty}\psi_n\, \fockx{n}$
 via $\psi_n   = \int_{\mathbb{R}} [\fockx{n}]^*(x)  \psi(x) \, \mathrm{d}x$
where $\fockx{n} \in \mathcal{S}(\mathbb{R})$ are eigenfunctions of the
quantum-harmonic oscillator. 
For any $\psi(x),\psi'(x)\in \mathcal{S}(\mathbb{R})$, the $L^2$ scalar product is
equivalent to the $\ell^2$ scalar product
\begin{equation}
\int_{\mathbb{R}}\psi^*(x)  \psi'(x) \, \mathrm{d}x 
= \sum_{n,m=1}^{\infty}\psi^*_n \psi'_m \int_{\mathbb{R}}  [\fockx{n}]^* \fockx{m} \, \mathrm{d}x 
= \sum_{n=1}^{\infty}\psi^*_n \psi'_n,
\end{equation}
and it is invariant with respect to the choice of orthonormal basis, i.e.~any two orthonormal bases are related
via a unitary transformation. 
The Plancherel formula
$\int_{\mathbb{R}}\psi^*(x)  \psi'(x) \, \mathrm{d}x
=\int_{\mathbb{R}}\psi^*(p)  \psi'(p) \, \mathrm{d}p$
yields the
equivalence  $L^2(\mathbb{R},\mathrm{d}x) \simeq  L^2(\mathbb{R},\mathrm{d}p)$.

In the following, we will consistently use the notation $\langle \,\cdot\,|\,\cdot\, \rangle$
for scalar products in Hilbert space, without specifying the type of representation.
This is motivated by the invariance of the scalar product under the choice of
representation. However, in order to avoid confusion with different types
of operator or Euclidean norms, we 
will use in the following the explicit norms $\|\psi(x)\|_{\indLnorm}$
and $\|\, |\psi\rangle \|_{\indslnorm}$, despite their equivalence.

We will now shortly define the trace of operators on infinite-dimensional Hilbert spaces,
refer to \cite[Ch.VI.6]{ReedSimon1} for a comprehensive introduction.
Recall that the trace of a positive semi-definite operator\footnote{
Here, $\mathcal B(\mathcal H)$ denotes the set of bounded linear operators on $\mathcal H$,
and one has $| \langle x | A  y\rangle | < \infty$ for every $A \in \mathcal B(\mathcal H)$ and $x,y \in \mathcal H$ .
An operator $A\in\mathcal B(\mathcal H)$ is said to be
positive semi-definite if $A$ is self-adjoint and $\langle x | Ax\rangle\geq 0$ for all $x\in\mathcal H$.} 
$A\in\mathcal B(\mathcal H)$ is defined via
$
\tr{(A)}=\sum_{n=1}^\infty \langle\psi_n|A\psi_n \rangle
$,
where the sum of non-negative numbers on the right-hand side is independent of the chosen orthonormal basis
$\lbrace |\psi_n\rangle\,,\,n\in\mathbb N\rbrace$ of $\mathcal H$,
but it does not necessarily converge.
Moreover recall that the set of trace-class operators is given by
\begin{equation*}
\mathcal B^1(\mathcal H):=\lbrace A\in\mathcal B(\mathcal H)\,,\tr{(\sqrt{A^\dagger A})}<\infty\rbrace\subseteq\mathcal K(\mathcal H),
\end{equation*}
where $\mathcal K(\mathcal H)$ denotes the set of compact operators on 
$\mathcal H$ and $A^{\dagger}$ is the adjoint of $A$
(which is in finite dimensions given by the complex conjugated and transposed matrix). 
The expression $\tr{(\sqrt{A^\dagger A})}=:\|A\|_1$ is called the trace norm on $\mathcal B^1(\mathcal H)$ 
which turns the trace class into a Banach space.
Every $A \in\mathcal B^1(\mathcal H)$ has a finite trace via the absolutely convergent sum
(of not necessarily positive numbers)
$
\tr{(A)}:=\sum_{n=1}^\infty\langle \psi_n|A\psi_n \rangle.
$
For $A \in\mathcal B^1(\mathcal H)$,
the mapping $A\mapsto\tr{(A)}$ is linear, continuous with respect to the trace norm, and independent 
of the chosen orthonormal basis of $\mathcal H$.
Trace-class operators $A\in\mathcal B^1(\mathcal H)$ have the important property that their products
with bounded operators $B\in\mathcal B(\mathcal H)$ are also in the trace class,
i.e.~$A B,BA\in\mathcal B^1(\mathcal H)$. Using this definition, one can calculate the trace independently
from the choice of the orthonormal basis or representation that is used for evaluating scalar products.

A density operator or state $\rho\in\mathcal B^1(\mathcal H)$ is 
defined to be positive semi-definite with $\tr{(\rho)}=1$. It therefore
admits a spectral decomposition \cite[Prop.~16.2]{MeiseVogt},
i.e.~there exists an orthonormal system
$\lbrace |\psi_n\rangle,n\in \mathbb{N} \rbrace$ in $\mathcal H$ such that 
\begin{equation} \label{densityop}
\rho= \sum_{n = 1}^\infty \, p_n | \psi_n \rangle \langle \psi_n |.
\end{equation}
The probabilities $\lbrace p_n, n\in \mathbb{N} \rbrace$ satisfy $p_1\geq p_2\geq\ldots \geq0$
 and $\sum_{n = 1}^\infty p_n=1$.
Expectation values of observables are computed via the trace
$
\langle O \rangle_\rho = \tr{(\rho \,O)} = \sum_{n=1}^\infty p_n   \langle \psi_n | O  \psi_n \rangle
$
where $O\in\mathcal B(\mathcal H)$ is self-adjoint. The following is a simple consequence of, e.g., \cite[Lemma 16.23]{MeiseVogt}.
\begin{lemma} \label{lemma1}
The expectation value 
of an observable $O$ in a mixed quantum state is upper bounded
by the operator norm 
$
| \tr{(\rho \,O)}  | \leq \|O\|_\indsupnorm$
for arbitrary density operators $\rho$, where
we have used the definition 
$\|O\|_\indsupnorm := \sup_{\|  (|\psi\rangle) \|_{\indslnorm}=1}
\| O | \psi \rangle \|_{\indslnorm}$ for the Hilbert space $\ell^2$.
Equivalently, the definition $\|O\|_\indsupnorm :=\sup_{\| \psi(x) \|_{\indLnorm} =1 }
\| O \psi(x)\|_{\indLnorm}$
for square-integrable functions $\psi(x)$ can be used.
\end{lemma}	

\section{Coherent States, Phase Spaces, and Parity Operators \label{phasespace}}

We continue to fix our notation by discussing 
an abstract definition of phase spaces that relies on displaced parity operators.
This usually appears concretely in terms of coherent states \cite{brif98,perelomov2012,gazeau,LI94},
for which we consider two equivalent but equally important parametrizations of the phase space 
using the coordinates $\alpha$ or $(x,p)$ (see below). This definition of phase spaces
can be also related to convolutions of Wigner functions which is usually known as the Cohen class
\cite{thewignertransform,cohen1966generalized,Cohen95}.
We also recall important postulates for Wigner functions as given by Stratonovich \cite{stratonovich,brif98}
and these will be later considered in the context of general phase spaces.

\subsection{Phase-Space Translations of Quantum States \label{translations}}
We will now recall a definition of the phase space for quantum-mechanical systems
via coherent states, refer to \cite{brif98,perelomov2012,gazeau,LI94}. 
We consider a quantum system which has a specific dynamical symmetry
group given by a Lie group $G$. The Lie group $G$
acts on the Hilbert space $\mathcal{H}$ using an irreducible 
unitary representation $\mathcal{D}$ of $G$.
By choosing a fixed reference state as an element $| 0 \rangle \in \mathcal{H}$ 
of the Hilbert space, one can define a set of coherent states
as $ | g \rangle := \mathcal{D}(g) | 0 \rangle$ where $g\in G$.
Considering the subgroup $H\subseteq G$ of
elements $h \in H$ that
act on the reference state only by
multiplication 
$\mathcal{D}(h)| 0 \rangle  := e^{i\phi}| 0 \rangle $
with a phase factor $e^{i\phi}$,
any element $g \in G$ can be decomposed into $g = \psc h$
with $\psc \in G /H$.
The phase space is then identified with the set of coherent
states $| \psc \rangle := \mathcal{D}(\psc) | 0 \rangle$.
In the following,
we will consider the Heisenberg-Weyl group $H_3$, for which the phase space
$\psc \in H_3 /U(1)$ is a plane. We introduce
the corresponding displacement operators that generate translations of the
plane. Displacement operators are also known 
as Heisenberg-Weyl operators \cite{thewignertransform} or, in the physics literature, simply as Weyl operators
\cite{Davies76,AlickiLendi07,Holevo12}.

In particular, for harmonic oscillator systems, the phase space $\psc \equiv \alpha \in \mathbb{C}$
is usually parametrized by the complex eigenvalues $\alpha$ of the annihilation
operator $\ahat$ and Glauber coherent states can be represented explicitly \cite{Cahill68}
in the so-called Fock  (or number-state) basis  as
\begin{equation}\label{D_alpha}
| \alpha \rangle = 
e^{-|\alpha|^2/2} \sum_{n=0}^{\infty} \tfrac{\alpha^n}{\sqrt{n!}}  | n \rangle  =
e^{\alpha \ahatdagg -\alpha^* \ahat} | 0 \rangle =: \mathcal{D}(\alpha) | 0 \rangle.
\end{equation}
Here, the second equality specifies the displacement operator $\mathcal{D}(\alpha)$ as a power
series of the usual bosonic annihilation $\ahat$ and creation $\ahatdagg$ operators,
which satisfy the commutation relation $[\ahat,\ahatdagg] = 1$, refer to Eq.~(2.11)
in \cite{Cahill68}.
In particular, the number state representation of displacements is given by \cite{Cahill68}
\begin{equation}
\label{displacementopmatirx}
[\mathcal{D}(\alpha)]_{m n} :=  \langle m | \mathcal{D}(\alpha)  | n \rangle
=
(\tfrac{n!}{m!})^{1/2} \alpha^{m-n} e^{-|\alpha|^2/2} L_n^{(m-n)}(|\alpha|^2),
\end{equation}
where $L_n^{(m-n)}(x)$ are generalized Laguerre polynomials.
This is the usual formulation for bosonic systems (e.g., in
quantum optics) \cite{leonhardt97}, where the optical phase space is the
complex plane and the phase-space integration measure is given by
$\mathrm{d}\psc =2\hbar \,\mathrm{d}^2\alpha  =2 \hbar  \, \mathrm{d}\Re(\alpha) \,\mathrm{d}\Im(\alpha)$
(where one often sets  $h=2\pi\hbar=1$, cf., \cite{Cahill68,cahill1969,brif98}).
The annihilation operator admits a simple decomposition 
\begin{equation*}
\ahat =
2 \hbar  \, \int_{-\infty}^{\infty} \int_{-\infty}^{\infty} 
 \alpha \,  | \alpha \rangle \langle \alpha |
  \,   \mathrm{d}\Re(\alpha) \,\mathrm{d}\Im(\alpha)
\end{equation*}
with respect to its eigenvectors,
see, e.g., \cite[Eqs.~(2.21)-(2.27)]{Cahill68}. 

Let us now consider the 
coordinate representation $\psi(x) \in \mathcal{S}(\mathbb{R})$ of a quantum state.
The phase space is parametrized by $\psc \equiv (x,p) \equiv z \in \mathbb{R}^2$ and the
integration measure is $\mathrm{d}\psc = \mathrm{d}z= \mathrm{d}x \, \mathrm{d}p$.
The displacement operator acts via (see also \cite{Wey27,Weyl31,Weyl50})
\begin{equation} \label{dispopdefxp}
\mathcal{D}(x_0,p_0) \psi(x)
:=
e^{\tfrac{i}{\hbar}(p_{0}x-\tfrac{1}{2}p_{0}x_{0})}
\psi(x{-}x_{0})
=
 e^{-\tfrac{i}{\hbar} (x_0 \hat{p} - p_0 \hat{x}) }  \psi(x),
\end{equation}
where $x,x_0,p_0 \in \mathbb{R}$.
The right hand side of Eq.~\eqref{dispopdefxp} 
specifies the displacement operator as a power
series of the usual operators $\hat{x}$ and $\hat{p}$,
which satisfy the commutation relation   
 $[\hat{x},\hat{p}] = i\hbar$, refer to \cite[Sec.~1.2.2., Def.~2]{thewignertransform}.
 
The most common representations of these two unbounded operators 
are $\hat{x} \psi(x)= x \psi(x)$ and
$\hat{p} \psi(x)= -i \hbar {\partial \psi(x)}/{\partial x}$.
Displacements of tempered distributions $\phi(x) \in \mathcal{S}'(\mathbb{R})$
are understood via the distribution pairings
$  (\mathcal{D}(\psc)\phi)( \psi):= \phi( \mathcal{D}(-\Omega) \psi  )$ where $-\Omega=(-x_0,-p_0)$.
This definition guarantees that\footnote{
This differs from other approaches where one considers 
the embedding $\iota:\mathcal S(\mathbb R)\to\mathcal S'(\mathbb R)$, $\phi\mapsto\int\phi(x)(\cdot)(x)\,\mathrm dx$ 
and the extension of $\mathcal D$ to tempered distributions is given by $\mathcal D(-x_0,p_0)$, 
cf.\ Example~\ref{ex_extend_displacement}(1) in Appendix \ref{app_operator_extension}.}
$$
\mathcal D(\Omega)[\langle \phi,\cdot\,\rangle](\psi)=\langle \mathcal D(\Omega)\phi,\psi\rangle
$$
as integrals from Section~\ref{functionfourier} (cf.~Example \ref{ex_extend_displacement}(2), Appendix \ref{app_operator_extension})
for all $\phi:\mathbb R\to\mathbb C$ such that $\langle \phi,\cdot\,\rangle\in\mathcal S'(\mathbb R)$,
and all $\psi\in\mathcal S(\mathbb R)$, $\Omega\in\mathbb R^2$.
In particular it does not matter whether $\mathcal D(\Omega)$ acts on a function $\phi:\mathbb R\to\mathbb C$ or on the induced functional
$\psi\mapsto\langle\phi,\psi\rangle$.

The two (above mentioned) physically motivated examples are
particular representations of the displacement operator for the Heisenberg-Weyl group
in different Hilbert spaces while relying on different parametrizations of the phase space.
Let us now highlight the equivalence of
these two representations. In particular, 
we obtain the formulas
$\al=(\lambda \hat{x} + i \lambda^{-1} \hat{p} )/\sqrt{2 \hbar}$
and
$\bl=(\lambda \hat{x} - i \lambda^{-1} \hat{p} )/\sqrt{2 \hbar}$
for any non-zero real $\lambda$, refer to Eqs.~(2.1-2.2) in \cite{Cahill68}.
In the context of quantum optics, the operators $\hat{x}$ and $\hat{p}$ are the so-called
optical quadratures \cite{leonhardt97}.
The operators $\al$ and $\bl$ are now defined on the Hilbert space $L^2(\mathbb R)$,
whereas $\ahat$ and $\ahatdagg$ act on elements of the Hilbert space $\ell^2$.
For any $\lambda \neq 0$ they reproduce the commutator
$[\al,\bl]=\operatorname{id}_{L^2}$, i.e.~$[\al,\bl] \, \psi(x) = \psi(x)$ for all $\psi(x)\in L^2(\mathbb R)$,
and they correspond to raising and lowering operators of the quantum 
harmonic-oscillator\footnote{For example, the choice $\lambda=\sqrt{m \omega}$ corresponds to the 
quantum-harmonic oscillator of mass $m$ and angular frequency $\omega$. And 
$\lambda = \sqrt{\epsilon \omega}$ is related to a normal mode of the electromagnetic field in a dielectric.}
eigenfunctions $\fockx{n}$, refer to \cite{hall2013quantum}.
Substituting now $\hat{x}=\sqrt{\hbar/2} \lambda^{-1} (\al {+} \bl)$
and $\hat{p}=-i \sqrt{\hbar/2} \lambda (\al {-} \bl)$
into the exponent on the right-hand side of \eqref{dispopdefxp} yields
\begin{equation*}
 -\tfrac{i}{\hbar} (x_0 \hat{p} - p_0 \hat{x}) 
 =  \bl (x_0 \lambda + i \lambda^{-1} p_0 )/\sqrt{2 \hbar} - \al (x_0 \lambda - i \lambda^{-1} p_0 )/\sqrt{2 \hbar}.
\end{equation*}
This then confirms the equivalence
\begin{equation} \label{tworepresentations}
\mathcal{D}(x_0,p_0) \psi(x)
=
e^{-\tfrac{i}{\hbar} (x_0 \hat{p} - p_0 \hat{x}) }  \psi(x)
=
e^{\bl \alpha - \al \alpha^* }  \psi(x)
= \mathcal{D}(\alpha) \psi(x),
\end{equation}
where the phase-space coordinate $\alpha$ is defined by $\alpha := (x_0 \lambda + i \lambda^{-1} p_0 )/\sqrt{2 \hbar}$.
Note that the corresponding phase-space element is then
$\mathrm{d}\psc = 2 \hbar \, \mathrm{d}\Re(\alpha) \,\mathrm{d}\Im(\alpha) = \mathrm{d} x \, \mathrm{d}p $
which is independent of the choice of $\lambda$.
Let us also recall two properties of the displacement operator \cite{Wey27,Weyl31,Weyl50} (see, e.g., \cite[p.~7]{thewignertransform}):
\begin{align}
\mathcal{D}(x_0,p_0) \mathcal{D}(x_1,p_1) & = e^{\tfrac{i}{\hbar}(p_0 x_1 - p_1 x_0)} 
\label{displacement_commutation}
\mathcal{D}(x_1,p_1) \mathcal{D}(x_0,p_0)
 \\
\label{eq:D_2}
\mathcal{D}(x_0{+}x_1,p_0{+}p_1) &= e^{-\tfrac{i}{2\hbar}(p_0 x_1 - p_1 x_0)} \mathcal{D}(x_0,p_0) \mathcal{D}(x_1,p_1).
\end{align}

In the following, we will use both of the phase-space coordinates $\alpha$ and $(x,p)$ 
interchangeably. The displacement operator is obtained in both parametrizations,
and they are equivalent via \eqref{tworepresentations}. Motivated
by the group definition, we will also use the parametrization $\psc$ for the phase space
via $\mathcal{D}(\psc)$, where $\psc$ corresponds to any representation of the
group, including the ones given by the coordinates $\alpha$ and $(x,p)$.

\subsection{Phase-Space Reflections and the Grossmann-Royer Operator \label{reflections}}
Recall that the parity operator $\parity$ reflects wave functions via
$\parity \psi(x) := \psi(-x)$ and $\parity \psi(p) := \psi(-p)$
for coordinate-momentum representations \cite{Grossmann1976,Royer77,thewignertransform,bishop1994,LIParity},
and $\parity | \psc \rangle := | {-} \psc \rangle $
for phase-space coordinates of coherent states \cite{cahill1969,Royer77,bishop1994,LIParity}. 
This parity operator is obtained as a phase-space average 
\begin{equation}
\label{paritydef}
\parity :=
(4 \pi \hbar)^{-1}\int \mathcal{D}(\psc) \, \mathrm{d}\psc
\end{equation}
of
the displacement operator from \eqref{dispopdefxp}.
One finds for all $\psi\in\mathcal S(\mathbb R)$, $x\in\mathbb R$ that
\begin{align*}
[\parity\psi](x)&=(4 \pi \hbar)^{-1}\int [\mathcal{D}(\psc)\psi](x) \, \mathrm{d}\psc\\
&=\tfrac12\cdot(2\pi\hbar)^{-1}\int e^{-\frac{i}\hbar(xp'-x'p)}[\mathcal{D}(\psc)\psi](x) \, \mathrm{d}\psc\;\Big|_{x'=p'=0}\\
&=\tfrac12 \big\{ \F_{\sigma} [\mathcal{D}\psi(x)](\psc')\big\}  \big|_{\psc'=0},
\end{align*}
or $\parity=\tfrac{1}{2} \{ [\F_{\sigma} \mathcal{D}](\psc')\}  |_{\psc'=0}$ for short.
Thus the parity operator equals evaluating the symplectic Fourier transform
of the displacement operator at
the phase-space point $\psc' = 0$. This
is related to the Grossmann-Royer operator
\begin{equation}\label{defgr}
\tfrac{1}{2} [\F_{\sigma} \mathcal{D}](-\psc)
=
\mathcal{D}(\psc)  \parity  \mathcal{D}^\dagger(\psc),
\end{equation}
which is the parity operator transformed by the displacement operator \cite{Grossmann1976,Royer77,thewignertransform,LIParity,bishop1994}.
Here, we use in both \eqref{paritydef} and \eqref{defgr} an abbreviated notation for formal integral
transformations of the displacement operator.

\begin{remark}\label{justification}
This abbreviation in Eq.~\eqref{defgr} is justified as the existence
of the corresponding integral
$
(\parity \phi)(\psi)=
(4 \pi \hbar)^{-1} \int\phi( \mathcal{D}^\dagger(\psc)\psi)\, \mathrm{d}\psc=(4 \pi \hbar)^{-1} \int\phi( \mathcal{D}(\psc)\psi)\, \mathrm{d}\psc
$
is guaranteed by, e.g., \cite[Sec.~1.3., Prop.~8]{thewignertransform}
for all $\phi\in\mathcal{S}'(\mathbb{R})$.
In the following, we will use this abbreviated notation
for formal integral transformations of the displacement operator, i.e.~by dropping $\phi$.
However, we might need
to restrict the domain
of more general parity operators to ensure the existence of the respective integrals.
\end{remark}

\subsection{Wigner Function and the Cohen class \label{wigfunct}}
The Wigner function $W_{\psi}(x,p)$ of a pure quantum state $| \psi \rangle$ was originally defined by
Wigner in 1932 \cite{wigner1932} and it is (in modern terms) the integral transformation 
of a pure state  $\psi \in L^2(\mathbb{R})$, i.e.
\begin{gather*}
W_{\psi}(x,p)
=
(2 \pi \hbar)^{-1}\int e^{- \tfrac{i}{\hbar} py} \psi^*(x{-}\tfrac{1}{2}y) \psi(x{+}\tfrac{1}{2}y)  \, \mathrm{d}y \\
=
(\pi \hbar)^{-1} \langle \psi, \, \mathcal{D}(x,p)  \parity  \mathcal{D}^\dagger(x,p) \psi \rangle
= (\pi \hbar)^{-1} \, \tr\,[ \, (|\psi\rangle\langle\psi|) \, \mathcal{D}(\psc)  \parity  \mathcal{D}^\dagger(\psc) ].
\end{gather*}
The second and third equalities specify the Wigner function using the Grossmann-Royer operator
\cite{Grossmann1976,thewignertransform}  from \eqref{defgr}, refer to \cite[Sec.~2.1.1., Def.~12]{thewignertransform}.
We use this latter form to extend the definition of the Wigner function to mixed quantum states as
in \cite{cahill1969,agarwal1970,Royer77,bishop1994}.

\begin{definition}\label{inifnitedimdefinition}
The Wigner function $W_\rho (\psc)\in L^2(\mathbb{R}^2)$ of an infinite-dimensional
density operator (or quantum state) 
$\rho=\sum_n p_n | \psi_n \rangle \langle \psi_n |\in\mathcal B^1(\mathcal H)$ is
proportional to the quantum-mechanical expectation value
\begin{equation} \label{wigdef}
W_\rho (\psc) :=
 (\pi \hbar)^{-1} \, \tr\,[ \, \rho \, \mathcal{D}(\psc)  \parity  \mathcal{D}^\dagger(\psc) ]
=  \sum_n p_n W_{\psi_n}(\psc)
\end{equation}
of the Grossmann-Royer operator from \eqref{defgr}, which
is the parity operator $\parity$ 
transformed by the displacement operator $\mathcal{D}(\psc)$, 
refer also to \cite{cahill1969,agarwal1970,Royer77,bishop1994,LIParity,thewignertransform}.
\end{definition}

The square-integrable cross-Wigner transform $ W_{\psi,\psi'}(\psc) \in L^2(\mathbb{R}^2)$
of two functions $\psi,\psi' \in L^2(\mathbb{R})$
used in time-frequency analysis
\cite{thewignertransform,bornjordan}
is obtained via the finite-rank operator $A = | \psi \rangle  \langle \psi'|$
in the form $W_{\psi,\psi'}(\psc):=W_A (\psc)$. Furthermore as 
$\lbrace\mathcal D(\psc):\psc\in \mathbb C\rbrace$ forms a subgroup 
of the unitary group on $\mathcal H$, the expression in 
Definition~\ref{inifnitedimdefinition} is closely related to the $C$-numerical range 
of bounded operators \cite{dirr_ve}, i.e.~$\lbrace W_\rho (\psc):\psc\in\mathbb C\rbrace$ 
by \eqref{wigdef} forms a subset of the $\rho$-numerical range of $(\pi \hbar)^{-1}\parity$.

The Wigner representation is in general  a bijective, linear mapping between
the set of density operators (or, more generally, the trace-class operators) and the phase-space
distribution functions $W_\rho$ that satisfy the so-called Stratonovich
postulates \cite{stratonovich,brif98}:
\begin{align*}
&\text{Postulate (i):}\;
\text{ $\rho \mapsto W_\rho	$ is one-to-one} \;\text{(linearity)}, \\
&\text{Postulate (ii):}\;
W_{\rho^\dagger}=W_\rho^* \;\text{(reality)},  \\
&\text{Postulate (iiia):}\;
\tr( \rho) = \int W_\rho  \, \mathrm{d} \psc \;\text{(normalization)},  \\
&\text{Postulate (iiib):}\; 
\tr( A^\dagger \rho ) = \int a^* \,  W_{\rho} \, \mathrm{d} \psc \;\text{(traciality)},  \\
&\text{Postulate (iv):}\; 
\text{$W_{\mathcal{D}(\psc')\rho}(\psc)=W_\rho(\psc {-} \psc')$} \;\text{(covariance)}.
\end{align*}
The not necessarily bounded\footnote{
	For unbounded operators $A$,
this postulate still makes sense if $\rho$ is 
has a finite 
representation in the number state basis, that is,
$\rho=\sum_{m,n=1}^N\langle m|\rho n\rangle|m\rangle\langle n|$ for some $N\in\mathbb N_0$.
Then this postulate gets replaced by the well-defined expression
$
\sum_{m,n=0}^N\langle m\rho|n\rangle\langle m|An \rangle^* = \int a^* \,  W_{\rho} \, \mathrm{d} \psc
$,
see also Appendix \ref{app_phase_space_extension}.} operator $A$ is the Weyl quantization
of the phase-space function (or distribution) $a(\psc) \in \mathcal{S}'(\mathbb{R}^2)$,
refer to Sec.~\ref{reltoquant}.
Based on these postulates, the Wigner function was defined for phase-spaces of
quantum systems with different dynamical symmetry groups via  coherent states
\cite{perelomov2012,gazeau,brif98,tilma2016,koczor2016,koczor2017}.

 Before finally presenting the definition of
the Cohen class for density operators following \cite[Sec.~8.1., Def.~93]{thewignertransform}
or \cite{Cohen95}, let us first recall the concept of convolutions.
Given Schwartz functions $a,\phi\in\mathcal S(\mathbb R^2)$ one defines their convolution via
	\begin{equation}\label{convolutiondef}
	\phi \ast a:= 2\pi \hbar\, \F_\sigma [(\F_\sigma \phi )\, (\F_\sigma a)]
	\end{equation} 
which is again in $\mathcal S(\mathbb R^2)$.
In principle this formula extends to general functions, although convergence may become an issue. 
These extensions are used in Theorem \ref{convolutionproposition} as well as Section \ref{reltoquant}.
	Now Eq.~\eqref{convolutiondef} as well as the fact that
$$(\phi\ast a)(\psc)=\int\phi(\psc')a(\psc{-}\psc')\,\mathrm d\psc'=\langle\phi^*,[\mathcal{T}(\Omega) a]^\vee\rangle$$
are for example shown in \cite[Thm.~IX.3]{ReedSimon1},
where $a^\vee(\Omega) := a(-\Omega)$ 
and $\mathcal{T}(\Omega)$ is the operator which translates a
function by $\Omega$ [i.e.~$\mathcal{T}(\Omega) a(\Omega') := a(\Omega'{-}\Omega)$].
With this in mind one arrives at an extension of the convolution to tempered
distributions \cite[Eq.~(4.37) ff.]{groechenig2001foundations}: Given $\theta\in\mathcal S'(\mathbb R^2)$, $a\in\mathcal S(\mathbb R^2)$ set
\begin{equation}\label{eq:convol_distr}
(\theta\ast a)(\psc):=\theta\big([\mathcal{T}(\Omega) a]^\vee\big)
\end{equation}
for all $\psc\in\mathbb R^2$.
This definition extends in a natural way to general linear functionals $\theta:D_\theta\to\mathbb C$
on some subspace $D_\theta\subseteq(\mathbb R^2\to\mathbb C)$, and general functions
$a:\mathbb R^2\to\mathbb C$ as long as $[\mathcal{T}(\Omega) a]^\vee\in D_\theta$ for all $\Omega\in\mathbb R^2$.

Defining the convolution via Eq.~\eqref{eq:convol_distr} is consistent with the distributional 
pairing in the sense that $\langle\phi^*|\ast a\equiv \phi\ast a$, if $\langle\phi^*|
(\psi):=\langle\phi^*,\psi\rangle$ on $\mathcal S(\mathbb R^2)$.
Moreover one readily verifies the identity
$\langle (\theta\ast a)^*,\psi\rangle=\theta(a^\vee\ast\psi)$ for all $\theta\in\mathcal S'(\mathbb R^2)$,
$a,\psi\in\mathcal S(\mathbb R^2)$. This shows that Eq.~\eqref{eq:convol_distr}
is equivalent to other extensions of convolutions commonly found in the literature, e.g., \cite[p.~324]{ReedSimon1}.
Be aware that $\theta\ast a$ is always a function of slow growth,
that is, $\langle(\theta\ast a)^*,\cdot\rangle\in\mathcal S'(\mathbb R^2)$
for all $\theta\in\mathcal S'(\mathbb R^2)$, $a\in\mathcal S(\mathbb R^2)$ \cite[Thm.~IX.4]{ReedSimon1}.

\begin{definition} \label{defofcohenclass}
	The Cohen class is the set of all linear mappings from
	density operators to phase-space distributions that
	are related to the Wigner function via a convolution.
	More precisely a linear map 
	$\FF:\mathcal B^1(L^2(\mathbb R))\to(\mathbb R^2\to\mathbb C)$, $\rho\mapsto \FF_\rho$
	maps to the phase-space distributions if $\langle \FF_\rho,\cdot\,\rangle\in\mathcal S'(\mathbb R^2)$
	for all $\rho\in\mathcal B^1(L^2(\mathbb R))$. Then $\FF$
	belongs to the Cohen class if there exists\footnote{
	More precisely $\theta$ has to be a linear functional on a subspace
$D_\theta$ of $\mathbb R^2\to\mathbb C$ such that
$[\mathcal{T}(\Omega) W_\rho]^\vee\in D_\theta$ for all $\rho\in\mathcal B^1(L^2(\mathbb R))$, $\Omega\in\mathbb R^2$.
However we will keep things informal by assuming henceforth that all convolutions we encounter are well-defined in the sense of Eq.~\eqref{eq:convol_distr}.
	} $\theta\in\mathcal S'(\mathbb R^2)$ (called ``Cohen kernel'') such that 
	$$
	\FF_\rho(\psc)=[\theta \ast W_\rho](\psc).
	$$
\end{definition}

This is a generalization of the definition commonly found in the literature \cite[Def.~93]{thewignertransform}:
there one restricts the domain of $\FF$ from the full trace class to only rank-one operators $\rho=|\phi\rangle\langle\psi|$
for some $\phi,\psi\in L^2(\mathbb R)$ or even $\in\mathcal S(\mathbb R)$.
As a simple example \cite[p.~90]{thewignertransform} the Wigner function is in the Cohen class:
To see this choose $\theta=\delta$ in the above definition:
$
[\delta\ast W_\rho](\psc)=\delta([\mathcal T(\psc)W_\rho]^\vee)=W_\rho(\psc).$

\begin{remark}\label{rem_cohen_fourier}
	Given some $\theta\in\mathcal S'(\mathbb R^2)$ associated to an element $\FF$ of the Cohen class, 
	one formally obtains
	$ \F_\sigma[\FF_\rho]=\F_\sigma[\theta \ast W_\rho]=K_\theta \F_\sigma[W_\rho] $
	if the symplectic Fourier transform of $\theta$ is generated by a function $K_\theta:\mathbb R^2\to\mathbb C$ 
	via the usual distributional pairing (we will call this ``admissible'' later, cf.~Section \ref{phase_space_theory_sec}).
	The reason we make this observation is that this object always exists: it is a product of 
	two classical functions where $\F_\sigma[W_\rho]$ is a bounded and square-integrable function,
	i.e.~$|\F_\sigma[W_\rho](\psc)| = |\tr[\mathcal D(\psc) \rho]| \leq \|\mathcal D(\psc)\|_\indsupnorm\|\rho\|_1 =1$
	due to unitarity of $\mathcal D(\psc)$, and $W_\rho\in
	L^2(\mathbb{R}^2)$ \cite[Proposition 68]{thewignertransform} so the same holds true for its Fourier transform.
	Thus---while the expression $\theta\ast W_\rho$ may be ill-defined for certain $\theta\in\mathcal S'(\mathbb R^2)$, 
	$\rho\in\mathcal B^1(L^2(\mathbb R))$---going to the Fourier domain yields a well-defined object which can 
	be studied rather easily.
\end{remark} 	
\section{Theory of Parity Operators and Their Relation to Quantization \label{parityopquantizationsec}}

\subsection{Phase-Space Distribution Functions via Parity Operators\label{phase_space_theory_sec}}
We propose a definition for phase-space distributions
and the Cohen class based on parity operators, the explicit form of which will
be calculated in Section~\ref{explicitparitysection}.  A similar form has already appeared 
in quantum optics for the so-called $s$-parametrized
distribution functions, see, e.g.,  \cite{cahill1969,moya1993}.
In particular, an explicit form of a parity operator that
requires no integral-transformation
appeared in (6.22) of \cite{Cahill68},
including its eigenvalue decomposition which was later
re-derived in the context of measurement probabilities in \cite{moya1993},
refer also to \cite{Royer77,LIParity}.
Apart from those results, mappings between density operators and their
phase-space distribution functions have been established only in terms of
integral transformations of expectation values,
as in \cite{cahill1969,Agarwal68,agarwal1970}.

For a convolution kernel $\theta \in  S'(\mathbb R^2)$, we introduce the corresponding \emph{filter kernel}
\begin{equation}\label{eq_general_K_theta}
K_\theta :=  2\pi\hbar\, \F_\sigma(\theta)
\end{equation}
where 
$\F_\sigma$ denotes the symplectic Fourier transform
(see Section~\ref{functionfourier}).
Henceforth we say $\theta\in\mathcal S'(\mathbb R^2)$ is \emph{admissible} if its filter kernel is generated 
by a function via the usual integral form of the distributional pairing  $\langle \phi , \psi\rangle=\phi(\psi)\in\mathbb C$
for $\phi \in \mathcal S'(\mathbb R)$
and $\psi \in \mathcal S(\mathbb R)$ 
(see Section~\ref{functionfourier}):
More precisely $\theta$ is admissible if
there exists a function $K_\theta$ from 
$\mathbb R^2$ to $\mathbb C$ such that $2\pi\hbar\, 
\F_\sigma(\theta)  (\psi) =\langle K_\theta^*, \psi \rangle$ for $\psi \in \mathcal S(\mathbb R)$. In this case we call $K_\theta$ 
the \textit{filter function} associated with $\theta$.

Most importantly if the convolution kernel is admissible and itself is generated by a function,
i.e.~if we consider $\langle\theta^*,\cdot\,\rangle\in  S'(\mathbb R^2)$ admissible, then Eq.~\eqref{eq_general_K_theta} simplifies to
\begin{equation}\label{eq_K_theta}
K_\theta(\Omega) = 2\pi\hbar\, [\F_\sigma \theta^\vee](\Omega)= 2\pi\hbar\, [\F_\sigma  \theta](-\Omega)
\end{equation}
for all $\Omega\in\mathbb R^2$.
As before $\theta^\vee(\Omega) = \theta(-\Omega)$.
The technical condition of $\theta$ being admissible is always satisfied in practice 
(cf.~Tables~\ref{quanttable} and \ref{filterfunctionstable}).
The advantage of only considering admissible kernels is that the definition of the (generalized) parity
operator makes for an obvious generalization of the parity operator from Section \ref{reflections}. 
For an even more general definition we refer to Remark~\ref{rem_general_def_parity} in Appendix \ref{app_operator_extension}.

\begin{definition} \label{def_parity_operator}
Given any admissible convolution kernel $\theta\in S'(\mathbb R^2)$ with associated filter function $K_\theta$
we define a parity operator $\parity_\theta$ on $\mathcal S(\mathbb R)$ via
\begin{equation} \label{genparitydef}
\parity_\theta := (4 \pi \hbar)^{-1}\int K_\theta(\psc) \mathcal{D}(\psc) \, \mathrm{d}\psc,
\end{equation}
that is, $[\parity_\theta\psi](x):=(4\pi\hbar)^{-1}\int K_\theta(\Omega)[\mathcal D(\Omega)\psi ](x)\,\mathrm d\psc$
for all $\psi\in\mathcal S(\mathbb R)$, $x\in\mathbb R$. 
This extends to a parity operator on the tempered distributions $\langle \Pi_\theta|:D_\theta\to\mathcal S'(\mathbb R)$ via
\begin{equation}\label{eq:parity_extension}
\langle \Pi_\theta|  := (4 \pi \hbar)^{-1}\int K_\theta^*(\psc) \mathcal{D}^\dagger(\psc) \, \mathrm{d}\psc
\end{equation}
(where the notation $\langle \Pi_\theta|$ is replaced below with $\Pi_\theta$) with domain
\begin{equation} \label{integral_domain}
	D_\theta:=	\{
	\phi \in \mathcal{S}'(\mathbb{R})
	\,\text{ s.t. }\,
	\int K_\theta^*(\psc)\phi[ \mathcal{D}(-\psc)(\cdot)]\, \mathrm{d}\psc
	\in\mathcal{S}'(\mathbb{R})
	\}.
\end{equation}
\end{definition}

The derivation of the extension \eqref{eq:parity_extension} of $\Pi_\theta$ to tempered distributions is detailed in
Appendix~\ref{app_operator_extension}. Displacements of tempered distributions $\phi \in \mathcal{S}'(\mathbb{R})$
are understood via the distributional pairing
$ \langle \mathcal{D}(\psc)\phi,  \psi  \rangle = \langle \phi, \mathcal{D}^\dagger(\psc) \psi  \rangle$
and \eqref{eq:parity_extension} gives rise to a well-defined linear operator $\langle\Pi_\theta|$ from
$D_\theta$ to $\mathcal S'(\mathbb R)$ acting on $\psi\in\mathcal S(\mathbb R)$ via
	 \begin{equation}\label{eq:action_gen_parity}
( \langle \Pi_\theta|\phi)(\psi)=(4\pi\hbar)^{-1}\int K_\theta^*(\psc) 
	 \langle  \phi,\,
	  e^{-\tfrac{i}{\hbar}(p_{0}x+\tfrac{1}{2}p_{0}x_{0})}
	 \psi(x{+}x_{0})
	 \rangle 
	 \, \mathrm{d}\psc.
	 \end{equation}

The definition of $\Pi_\theta$ 
is independent of
the object it acts on
(see Appendix~\ref{app_operator_extension}): 
$\langle \Pi_\theta|\langle\phi,\cdot\,\rangle=\langle\Pi_\theta\phi,\cdot\,\rangle$ for all 
$\phi\in\mathcal S(\mathbb R)$ where 
$\langle\phi,\cdot\,\rangle$ denotes
the functional
$\psi\mapsto\langle\phi,\psi\rangle\in\mathcal S'(\mathbb R)$.
All filter functions used in practice (refer to Tables~\ref{quanttable} and
\ref{filterfunctionstable}) obey
$K_\theta^*(x_0,p_0)=K_\theta(x_0,-p_0)$ for all $x_0,p_0\in\mathbb R$.
In this case,
$\langle\Pi_\theta|$ is not only compatible with
the inner product on $L^2(\mathbb R)$, but also with the embedding
$\mathcal S(\mathbb R)\hookrightarrow\mathcal S'(\mathbb R)$ usually employed in mathematical physics
(see Lemma \ref{lemma_filter_function_symmetry} in Appendix~\ref{app_operator_extension}).
This motivates us to
henceforth write $\Pi_\theta$ both in the case of \eqref{genparitydef} and 
instead of $\langle \Pi_\theta|$ in \eqref{eq:parity_extension}.

While our definition above is pleasantly intuitive, we have to explicitly consider the domain
of the parity operator.
For a general (admissible) kernel $\theta$,
one needs to restrict the domain $D_\theta \subseteq \mathcal{S}'(\mathbb{R})$
of
$\parity_\theta$ to tempered distributions for which
the integral in Eq.~\eqref{genparitydef} exists, as
done in Eq.~\eqref{integral_domain} and
already hinted at in Remark~\ref{justification}.

\begin{example}\label{example:divergence}
Domain considerations are illustrated using the standard ordering with
$K_\theta(\Omega_0)=\exp[{i p_0 x_0/(2 \hbar) }]$ (see Table~\ref{quanttable}). 
Given any $\phi,\psi\in\mathcal S(\mathbb R)$, we have
\begin{align}\nonumber
\langle\phi,\Pi_\theta\psi\rangle&=(\Pi_\theta\langle\psi,\cdot\rangle)(\phi)^*
=(4\pi\hbar)^{-1}\iiint\phi^*(x{+}x_0)e^{\tfrac{i}{\hbar}p_0(x+x_0)}\psi(x)\,\mathrm{d}x_0\,\mathrm{d}x\,\mathrm{d}p_0\\
\nonumber
&=(8\pi\hbar)^{-1/2}\Big(\int[\F\phi](p_0)\,\mathrm{d}p_0\Big)^*\Big(\int\psi(x)\,\mathrm{d}x\Big)\\
&=\sqrt{\tfrac{\pi\hbar}{2}}\F[\F\phi]^*(\tilde{p})\big|_{\tilde{p}=0}[\F\psi](\tilde{x})\big|_{\tilde{x}=0}
=\sqrt{\tfrac{\pi\hbar}2}\phi^*(0)[\F\psi](0). \label{eq:ex2}
\end{align}
This reproduces known properties
as in  Eq.~(5.39) of
\cite{cohen2012weyl} (cf.\ Remark \ref{remark_cohen}); however we emphasize that, although Eq.~\eqref{eq:ex2}
exists for all functions $\phi,\psi$ as long as $[\F\psi](0)$ exists, this expression is only equal to
$\langle\phi,\Pi_\theta\psi\rangle$ if in addition $\phi$ and $\F\phi$ are both in $L^1$ (else the Fourier
inversion formula used in the last step cannot be applied).
In other words a function $\phi:\mathbb R\to\mathbb R$ 
is in the domain  $D_\theta$ of $\Pi_\theta$ if and only if its Fourier transform exists and is in
$L^1(\mathbb R)$ if and only if \eqref{eq:ex2} (resp.~Eq.~(5.39) of
\cite{cohen2012weyl}) equals $\langle\phi,\Pi_\theta\psi\rangle$
for all suitable $\psi$. 
In particular,
$D_\theta$ contains all Schwartz functions showing that $\Pi_\theta$ is
densely defined.
However the functional $\langle\phi,\cdot\,\rangle\in\mathcal S'(\mathbb R)$ fails 
to be in $D_\theta$ for most functions $\phi:\mathbb R\to\mathbb C$ of slow growth including 
non-zero constant ones such as $\phi:= 1 \in\mathcal{S}'(\mathbb{R})$. 
In particular, $\Pi_\theta$ does not extend
to a well-defined operator on $L^2(\mathbb R)$ as not all square-integrable
functions
will be contained in  $D_\theta$.
\end{example}

Following this line of thought, we investigate the well-definedness and boundedness of 
$\Pi_\theta$ on the Hilbert space $L^2(\mathbb R)$.
As in Example \ref{example:divergence}, we observe 
that $\mathcal S(\mathbb R)\subseteq D_\theta$ for all filter functions $K_\theta$
which is particularly relevant for applications.
This follows by interpreting $\Pi_\theta$ as a Weyl quantization 
(cf.~Section \ref{reltoquant}) whereby $\theta\mapsto\Pi_\theta$ is specified as a map from $\mathcal S'(\mathbb R^2)$
to the linear maps between 
$\mathcal S(\mathbb R)$ and $\mathcal S'(\mathbb R)$ (cf.~Chapter 6.3 in \cite{bornjordan} or
Lemma 14.3.1 in \cite{groechenig2001foundations}).
Consequently, every parity operator has a well-defined matrix representation
in the number-state basis (which is a subset of $\mathcal S(\mathbb R)$,
cf.~Section \ref{quantumstates}).
The following stronger statement is shown in Appendix~\ref{proofs_4_1_a}:
\begin{lemma}\label{lemma_bounded_op_sufficient}
Given any convolution kernel $\theta\in\mathcal S'(\mathbb R^2)$ the following are equivalent:\smallskip\\
(i,a) $\Pi_\theta:L^2(\mathbb R)\to L^2(\mathbb R)$ is a well-defined linear operator, that is, the mapping 
$x\mapsto \frac12\theta(\F_\sigma[ \mathcal D\psi(x) ])$ (cf.~Remark \ref{rem_general_def_parity},
Appendix \ref{app_operator_extension}) is in $L^2(\mathbb R)$ for all $\psi\in L^2(\mathbb R)$.\\
(i,b) $[\theta \ast W_{\psi}](0,0)$ exists for all $\psi\in L^2(\mathbb R)$, i.e.~$[\theta \ast W_{\psi}](0,0)<\infty$.\smallskip

Also the following statements are equivalent:\smallskip\\
(ii,a) $\sup_{\psi,\phi\in L^2(\mathbb R),\|\psi\|=\|\phi\|=1}| [\theta \ast W_{\phi,\psi}](0,0)|<\infty$.\\
(ii,b) $(\phi,\psi)\mapsto \theta \ast W_{\phi,\psi}$ is weakly continuous on $L^2(\mathbb R)$ in the sense that there exists 
$C>0$ such that  $| [\theta \ast W_{\phi,\psi}](0,0)|\leq C\|\phi\|\|\psi\|$ for all $\phi,\psi\in L^2(\mathbb R)$.\\
(ii,c) $\Pi_\theta\in\mathcal B(L^2(\mathbb R))$.\smallskip\\
Moreover if $\theta$ is admissible, then (i,a), (i,b) and (ii,a), (ii,b), (ii,c) are all equivalent.
\end{lemma}

Recalling from Section \ref{wigfunct}, $W_{\phi,\psi}$ is the usual cross-Wigner transform given by
\begin{align*}
W_{\phi,\psi}(x,p)&=(\pi\hbar)^{-1}\langle\psi,\mathcal D(x,p)\Pi\mathcal D^\dagger(x,p)\phi\rangle\\
&=(2\pi\hbar)^{-1}\int_{-\infty}^\infty e^{-\tfrac{i}{\hbar}py}\psi^*(x{-}\tfrac{y}2)
\phi(x{+}\tfrac{y}2)\,dy.
\end{align*}
Let us highlight that condition (ii,b) in Lemma \ref{lemma_bounded_op_sufficient} is a known sufficient condition from 
time-frequency analysis to ensure that a tempered distribution $\theta$ is an element of the Cohen class,
cf.~Theorem 4.5.1 in \cite{groechenig2001foundations}.
Now the almost magical result of Lemma \ref{lemma_bounded_op_sufficient} is that $\Pi_\theta$ being well defined on $L^2(\mathbb R)$
automatically implies boundedness as long as $\theta$ is admissible. This can also be attributed to the folklore that unbounded 
operators ``cannot be written down explicitly'': As the operator $\Pi_\theta$ for admissible kernels is defined via an explicit
integral, one gets the boundedness of $\Pi_\theta$ ``for free.''
Indeed the proof that all five statements from the above lemma
are equivalent breaks down if one considers not only admissible but arbitrary kernels.

We define a general class of phase-space distribution functions $\FF_\rho(\psc,\theta)$ via the (formal) 
expression $ (\pi \hbar)^{-1} \, \tr\,[ \, \rho \, \mathcal{D}(\psc)  \parity_\theta  \mathcal{D}^\dagger(\psc) ]$. For general $\theta$, 
however, this only makes sense if all displaced quantum states
$\mathcal D^\dagger(\psc)\rho\mathcal D(\psc)$ are supported on $D_\theta$.
We avoid these technicalities by restricting the 
definition to those filter functions which give rise to operators $\Pi_\theta$ that are bounded on $L^2(\mathbb R)$
and thereby allow for general $\rho$.
\begin{definition}\label{def_phase_space_function}
Given any 
$\theta\in\mathcal S'(\mathbb R^2)$ such that $\Pi_\theta\in\mathcal B(L^2(\mathbb R))$ we define a linear mapping $\FF_\rho(\cdot,\theta)$ on
the density operators $\rho\in\mathcal B^1(L^2(\mathbb R))$ 
in the form of the quantum-mechanical expectation value
\begin{equation}
\label{genpsdef}
\FF_\rho(\psc,\theta)
:= (\pi \hbar)^{-1} \, \tr\,[ \, \rho \, \mathcal{D}(\psc)  \parity_\theta  \mathcal{D}^\dagger(\psc) ].
\end{equation}
\end{definition}

While our definition considers the practically most important case of bounded parity operators, 
we give a detailed account in Appendix~\ref{app_phase_space_extension}
of the extension of $\FF_\rho(\psc,\theta)$ to arbitrary $\theta\in\mathcal S'(\mathbb R^2)$  whereby
the associated parity operators may be unbounded.
This is of importance for, e.g., the standard and antistandard orderings as shown in Example \ref{example:divergence}.
The prototypical case
where these extensions may \textit{not} apply due to $\theta\notin\mathcal S'(\mathbb R^2)$ is
the case of the Glauber P function which is well known to be singular except for classical thermal states.
However, most other convolution kernels appearing
in practice are induced by a tempered distribution and thus fit into the
framework of either Definition~\ref{def_phase_space_function} or its
extension in Appendix~\ref{app_phase_space_extension}.

Either way Definition \ref{def_phase_space_function} has many conceptual and computational advantages as we have detailed in 
the introduction. 
To further clarify the scope of said definition
we now---similarly to the proof of Lemma \ref{lemma_bounded_op_sufficient}---relate
the distribution functions $\FF_\rho(\psc,\theta)$ from Eq.~\eqref{genpsdef} to
the Cohen class (see Definition~\ref{defofcohenclass} and \cite[Ch.~8]{thewignertransform}) by considering the 
filter function associated with any admissible kernel.
\begin{theorem} \label{convolutionproposition}
Given any $\theta \in S'(\mathbb R^2)$ such that $\Pi_\theta\in\mathcal B(L^2(\mathbb R))$,
the corresponding phase-space distribution function 
	$\FF_\rho(\psc,\theta) \in \mathcal{S}'(\mathbb{R}^2)$
	as defined in Eq.~\eqref{genpsdef}
	is an element of the Cohen class.
	In particular,
 $\FF_\rho(\psc,\theta)$ is related to the Wigner function
	$W_\rho(\psc)$ via the convolution
	\begin{equation} \label{convolutioneq}
	\FF_\rho(\psc,\theta)
	=
	[\theta \ast W_\rho](\psc).
	\end{equation}
If the convolution kernel $\theta\in\mathcal S'(\mathbb R^2)$ 
is additionally admissible---meaning it is the reflected symplectic Fourier
	transform
$\theta=(2\pi\hbar)^{-1}\langle \F_\sigma K_\theta^*|$ of its filter function $K_\theta$---then
	in analogy to \eqref{defgr} one finds
	\begin{equation} \label{generalizedfourier}
	\mathcal{D}(\psc)  \parity_\theta  \mathcal{D}^\dagger(\psc)
	=\tfrac{1}{2} \F_{\sigma} [ K_\theta(\cdot) \mathcal{D}(\cdot) ](-\psc).
	\end{equation}
\end{theorem}

The proof of Theorem \ref{convolutionproposition} is given in Appendix~\ref{proofs_4_1_b}.
The construction of a particular class of phase-space distribution functions was detailed
in \cite{agarwal1970}, where the term ``filter function'' also appeared in the context of
mapping operators. However, these filter functions were restricted to non-zero, analytic
functions. Definition~\ref{def_phase_space_function}
extends these cases to the Cohen class
via Theorem \ref{convolutionproposition} which allows for more general phase spaces.
For example, the filter function of the Born-Jordan distribution
has zeros (see Theorem~\ref{BJdistributionTheorem} below), and is therefore not
covered by \cite{agarwal1970}. Most of the well-known distribution functions are elements of the Cohen class.
We calculate important special cases in Sec.~\ref{explicitparitysection}.
The Born-Jordan distribution and its parity operator are detailed in Sec.~\ref{bornjordansection}.

Our approach to define phase-space distribution functions using 
displaced parity operators also nicely fits
with the characteristic \cite{cahill1969,Cohen95,leonhardt97}
or ambiguity \cite[Sec.~7.1.2, Prop.~5]{thewignertransform} function  $\chi(\Omega) \in L^2(\mathbb{R}^2)$
of a quantum state that is defined as
the expectation value 
$\chi(\Omega):=\tr[\rho \mathcal{D}(\psc)]=[\F_{\sigma} W_\rho](\Omega)$
or, equivalently, as the symplectic Fourier transform of 
the Wigner function $W_\rho(\Omega)$.
By multiplying the characteristic function $\chi(\Omega)$
with a suitable filter function $K_\theta(\Omega)$
and applying the symplectic Fourier transform, one obtains
the Cohen class of phase-space distribution functions.

\begin{remark}\label{remark_cohen}
	Definitions~\ref{def_parity_operator} and \ref{def_phase_space_function} for the parity operator and
	the phase-space function
	can be compared to prior work where special cases or similar parity operators have implicitly appeared
	and where similar restrictions on their existence must be observed.
	For example, the integral definition \cite{cohen1966generalized}
	of phase-space functions
\begin{align}
&\FF_{|\phi\rangle\langle\phi|}(x,p,\theta)\nonumber\\ \label{eq:cohen}
&=(4\pi^2\hbar^2)^{-1}\iiint \phi^*(x'{-}\tfrac{y}2)
\phi(x'{+}\tfrac{y}2)K_\theta(-y,p')e^{-\tfrac{i}{\hbar}(xp'+ yp  -x'p' )}\,\mathrm{d}x'\,\mathrm{d}y\,\mathrm{d}p'\\ 
&=(\pi\hbar)^{-1}\langle\phi,\mathcal D(x,p)\Pi_\theta\mathcal D^\dagger(x,p)\phi\rangle.  \label{eq:cohen_2}
\end{align}
        as given\footnote{The filter function in \cite{cohen2012weyl} agrees with our $K_\theta(-y,p')$
	up to substituting $-y$ with $y$ and switching arguments, which is usually immaterial
	as $K_\theta(-y,p')=K_\theta(p',y)$ for all filter functions seen in practice.}
	in Eq.~(5.2) of \cite{cohen2012weyl}
	translates into the definition \eqref{eq:cohen_2}
	with the parity operator. Both Eqs.~\eqref{eq:cohen} and \eqref{eq:cohen_2} need 
	to respect domain restrictions as discussed in Example~\ref{example:divergence}
	and neither equation is well defined for tempered distributions in $S'(\mathbb R)$
	or square-integrable functions in $L^2(\mathbb R)$ that are not contained in the domain $D_\theta$.
\end{remark}

\subsection{Common Properties of Phase-Space Distribution Functions\label{property_phase_space}}

We detail now important properties
of $\FF_\rho(\psc,\theta)$ and their relation to properties of $K_\theta(\psc)$ and $\parity_{\theta}$.
These properties will guide our discussion
of parity operators and this allows us to compare the Born-Jordan distribution to other phase spaces.
Table~\ref{summaryofproperties} provides a summary of these properties and
the proofs have been deferred to Appendix~\ref{proofsofproperties}. Recall that
we are dealing exclusively with convolution kernels $\theta\in\mathcal S'(\mathbb R^2)$
which give rise to bounded operators $\Pi_\theta$ so the induced phase-space distribution $\FF_\rho(\psc,\theta)$ is well defined everywhere.

\begin{property} \label{bound} 
Boundedness of phase-space functions $\FF_\rho(\psc,\theta)$:
The phase-space distribution function $\FF_\rho(\psc,\theta)$
 is bounded in its absolute value, i.e.~$
\pi \hbar \, | \FF_\rho(\psc,\theta) | \leq \| \parity_\theta \|_\indsupnorm$ for all quantum states $\rho$,
refer to Lemma \ref{lemma1}.
In particular then $\FF_\rho(\Omega,\theta)\in\mathcal S'(\mathbb R^2)$.
Moreover one finds that square-integrable filter functions
give rise to bounded parity operators due to
$ \| \parity_\theta \|_\indsupnorm \leq 
\|K_\theta\|_{\indLnorm}/\sqrt{ 8 \pi \hbar}$.
The proof of Property~\ref{bound} in Appendix~\ref{proofsofproperties} implies the even stronger statement
that $\parity_{\theta}$ is a Hilbert-Schmidt operator if and only if
$K_\theta$ is square integrable.
\end{property}

\begin{property} \label{l2property}
Square integrability: The phase-space distribution function
$\FF_\rho(\psc,\theta)$ is square integrable [i.e.~$\FF_\rho(\psc,\theta) \in L^2(\mathbb{R}^2)$]
for all $\rho \in \mathcal{B}^1(\mathcal{H})$ if
the absolute value of the filter function is bounded
[i.e.~$K_\theta(\psc) \in L^\infty(\mathbb{R}^2)$].
In particular this implies $ \FF_\rho(\Omega,\theta)\in\mathcal S'(\mathbb R^2)$.
\end{property}

\begin{property} \label{translationalcovariance}
Postulate (iv): The  phase-space distribution function $\FF_\rho(\psc,\theta)$
satisfies by definition the covariance property.
In particular,  a displaced density operator $\rho':=\mathcal{D}(\psc') \rho \mathcal{D}^\dagger(\psc')$
is mapped to the inversely displaced distribution function
$\FF_{\rho'}(\psc,\theta)  = \FF_\rho(\psc {-} \psc',\theta)$.
\end{property}

\begin{property} \label{rotationalcovariance}
	Rotational covariance:
	Let us denote a rotated density operator $\rho^{\phi} = U_\phi \rho U^\dagger_\phi$,
	where the phase-space rotation operator is given by $U_\phi := \exp{(-i \phi  \ahatdagg \ahat)}$
	in terms of creation and annihilation operators.
	The phase-space distribution function is covariant under
	phase-space rotations,\footnote{Note that
		any physically motivated distribution function must be covariant under $\pi/2$
		rotations in phase-space, which corresponds to the Fourier transform of pure states
		and connects coordinate representations $\psi(x)$ to momentum representations $\psi(p)$.} 
	i.e.~$\FF_{\rho^{\phi}}(\psc,\theta)=\FF_\rho(\psc^{-\phi},\theta)$,
	if the filter function $K_\theta(\psc)$ (or equivalently the parity operator $\parity_{\theta}$)
	is invariant under rotations. Here, $\psc^{-\phi}$
	is the inversely rotated phase-space coordinate, e.g., $\alpha^{-\phi} = \exp{(i \phi  )} \alpha$.
	As a consequence of this symmetry, the corresponding parity operators are diagonal in the number-state representation,
	i.e.~$\langle n | \parity_\theta | m \rangle \propto \delta_{nm}$.
\end{property}

\begin{table}
\centering
\begin{tabular}{@{\hspace{1mm}}l@{\hspace{6mm}}l@{\hspace{6mm}}l@{\hspace{1mm}}}
\hline\noalign{\smallskip}
Property of $\FF_\rho(\psc,\theta)$ & Description & Requirement   \\
\noalign{\smallskip}\hline\noalign{\smallskip}
Boundedness & $|\FF_\rho(\psc,\theta)|$ is bounded  & $\|\parity_\theta\|_\indsupnorm$
is bounded\\[2mm]
Square integrability & $\FF_\rho(\psc,\theta) \in L^2(\mathbb{R}^2)$
 & $| K_\theta(\psc) |$ is bounded \\[2mm]
Linearity &  $\rho \mapsto \FF_\rho(\psc,\theta)$  is linear & by definition\\[2mm]
Covariance & $\mathcal{D}(\psc') \rho \mathcal{D}^\dagger(\psc')  \mapsto \FF_\rho(\psc {-} \psc',\theta)$ & by definition \\[2mm]
Rotations &  covariance under rotations & $K_\theta$ is invariant under rotations \\[2mm]
Reality & $\rho^\dagger  \mapsto \FF^*_\rho(\psc,\theta) $  & Symmetry $K_\theta^*(-\psc)=K_\theta(\psc)$ \\[2mm]
Traciality & $ \tr\,[  \rho ] \mapsto \int \FF_\rho(\psc,\theta) \, \mathrm{d} \psc$  & $[K_\theta(\psc)]|_{\psc=0}=1$ \\[2mm]
Marginal condition & $|\psi(x)|^2$ and $|\psi(p)|^2$ are recovered   & $[K_\theta(x,p)]|_{p=0}$\\
& & $=[K_\theta(x,p)]|_{x=0}=1$
\\[0mm]
\noalign{\smallskip}\hline
\end{tabular}
\caption{\label{summaryofproperties} Properties of phase-space distribution functions from
Definition~\ref{def_phase_space_function}.}
\end{table}

\begin{property} \label{reality}
Postulate (ii): The  phase-space distribution function $\FF_\rho(\psc,\theta)$
is real if $\parity_\theta$ is self-adjoint.
This condition translates to the symmetry $K^*_\theta(-\psc)=K_\theta(\psc)$ of
the filter function.
\end{property}

\begin{property} \label{cohencl}
Postulate (iiia):
The trace of a trace-class operator $ \tr\,[  \rho ]$ is mapped to
the phase-space integral $\int \FF_\rho(\psc,\theta) \, \mathrm{d} \psc$
if the corresponding filter function satisfies $K_\theta(0) = 1$.
Note that this property also implies that the trace exists, i.e.~$\tr (\parity_{\theta})  = K_\theta(0)/2$, in some particular basis,
even though $\parity_{\theta}$ might not be of trace class.
\end{property}

\begin{property} \label{marginal}
Marginals: An even more restrictive
subclass of the Cohen class satisfies the marginal properties
$
\int \FF_\rho(x,p,\theta) \,\mathrm{d} x = |\psi(p)|^2 
$ and
$
\int \FF_\rho(x,p,\theta) \,\mathrm{d} p = |\psi(x)|^2
$
if and only if $[K_\theta(x,p)]|_{p=0}=1$ and $[K_\theta(x,p)]|_{x=0}=1$.
This follows, e.g., directly from Proposition 14 in Sec.~7.2.2 of \cite{bornjordan}.
\end{property}

\subsection{Relation to Quantization \label{reltoquant}}
The Weyl quantization of
a tempered distribution $\mathfrak{a}\in\mathcal S'(\mathbb R^2)$ is 
obtained from the Grossmann-Royer operator in Eq.~\eqref{defgr} (cf.~\cite[Sec.~6.3., Def.~7 and Prop.~9]{bornjordan}), i.e.~
\begin{equation}\label{eq:def_weyl_quant_general}
\textup{Op}_{\textup{Weyl}}(\mathfrak{a})=(\pi\hbar)^{-1}\mathfrak{a}(\mathcal{D}\Pi\mathcal{D}^\dagger).
\end{equation}
More precisely $\textup{Op}_{\textup{Weyl}}(\mathfrak{a}):\mathcal S(\mathbb R)\to\mathcal S'(\mathbb R)$ is the well-defined linear map
\begin{equation}\label{eq:weyl_quant_def}
[\textup{Op}_{\textup{Weyl}}(\mathfrak{a})\psi](x)=(\pi\hbar)^{-1}\mathfrak{a}[(\mathcal{D}\Pi\mathcal{D}^\dagger\psi)(x)]
\end{equation}
for all $\psi\in\mathcal S(\mathbb R)$, $x\in\mathbb R$, where the argument of $\mathfrak{a}$ is 
the Schwartz function $\psc\mapsto (\mathcal D(\psc)\Pi\mathcal D^\dagger(\psc)\psi)(x)$ on $\mathbb R^2$
\cite[Sec.~6.3., Prop.~13]{bornjordan}.
If $\mathfrak{a}$ is generated by a phase-space function $a:\mathbb R^2\to\mathbb C$, i.e.~$\mathfrak a\equiv\langle a^*,\cdot\rangle$, then
\begin{equation} \label{weylquant}
\textup{Op}_{\textup{Weyl}}(\mathfrak a)
= (\pi \hbar)^{-1} \int a(\psc) \, \mathcal{D}(\psc)  \parity  \mathcal{D}^\dagger(\psc)  \, \mathrm{d} \psc
= (2 \pi \hbar)^{-1} \int a_\sigma(\psc) \, \mathcal{D}(\psc)  \, \mathrm{d} \psc,
\end{equation}
where the symplectic Fourier transform $a_\sigma(\psc) = [\F_{\sigma} a(\cdot)](\psc)$
is used for the second equality.
Thus $\textup{Op}_{\textup{Weyl}}$ is similar to the generalized parity operator 
in the sense that it maps a function (or tempered distribution) to a linear operator which
acts on real-valued functions. This is not by chance as these two objects are very much related to each other: 
recall that quantizations associated with the Cohen class
$\textup{Op}_{\theta}(\mathfrak a)$ are essentially Weyl quantizations of
convolved phase-space functions
up to coordinate reflection:
$\textup{Op}_{\theta}(\mathfrak a) := \textup{Op}_{\textup{Weyl}}(\theta^\vee \ast a)$
where $\theta^\vee(\psi):=\theta(\psi^\vee)$ for all $\psi\in\mathcal S(\mathbb R^2)$. If $\theta$ is an admissible kernel
in the sense of Section \ref{phase_space_theory_sec}, then formally
\begin{align} \label{thetaquant1}
\textup{Op}_{\theta}(\mathfrak a) 
&= (\pi \hbar)^{-1} \int (\theta^\vee \ast a) (\psc) \, \mathcal{D}(\psc)  \parity  \mathcal{D}^\dagger(\psc)  \, \mathrm{d} \psc  \\ 
&=(2\pi \hbar)^{-1}  \int a_\sigma(\psc) K_\theta(\psc) \, \mathcal{D}(\psc)  \, \mathrm{d} \psc, \label{thetaquant}
\end{align}
cf.~\cite[Sec.~7.2.4, Prop.~17]{bornjordan}.
The symplectic Fourier transform $ [\F_{\sigma} (\theta^\vee \ast a)	](\psc) = a_\sigma(\psc) K_\theta(\psc)$
(as functionals on $\mathcal S(\mathbb R^2)$ so in particular $ \langle \theta^\vee\ast a,\cdot\rangle\in\mathcal S'(\mathbb R^2) $)
from Theorem \ref{convolutionproposition}
is used for the second equality, refer to \S 7.2.4 in \cite{bornjordan}.

\begin{proposition} \label{quantdefinition}
Let $\theta,\mathfrak{a}\in\mathcal S'(\mathbb R^2)$ be given such that $\theta$ is admissible and
the parity operator $\Pi_\theta$ from Definition~\ref{def_parity_operator} is in $\mathcal B(L^2(\mathbb R))$.
If $\mathfrak a$ is generated by a phase-space function $a:\mathbb R^2\to\mathbb C$
(i.e.~$\mathfrak a\equiv\langle a^*,\cdot\rangle$) and if $ \langle \theta^\vee\ast a,\cdot\rangle\in\mathcal S'(\mathbb R^2) $, then
\begin{equation*}
\textup{Op}_{\theta}(\mathfrak a)
= (\pi \hbar)^{-1} \int  a(\psc) \, \mathcal{D}(\psc)  \parity_{\theta}  \mathcal{D}^\dagger(\psc)  \, \mathrm{d} \psc
\end{equation*}
in analogy to \eqref{weylquant} as quadratic forms on $\mathcal S(\mathbb R)$.
\end{proposition}

\begin{proof}
	The Plancherel formula
	$\int a(\psc) b(\psc)\, \mathrm{d} \psc = \int a_\sigma(\psc) b_\sigma(-\psc)\, \mathrm{d} \psc $ implies that
	\begin{align*}
	\textup{Op}_{\theta}(\mathfrak a)
	&=(2\pi \hbar)^{-1}  \int a_\sigma(\psc) K_\theta(\psc) \, \mathcal{D}(\psc)  \, \mathrm{d} \psc \\
	&=(2\pi \hbar)^{-1}  \int  \F_\sigma [a_\sigma(\cdot)](\psc) \, 
	\F_\sigma [ K_\theta(\cdot) \, \mathcal{D}(\cdot)](-\psc)  \, \mathrm{d} \psc,
	\end{align*}
	where the equality 	$\mathcal{D}(\psc)  \parity_\theta  \mathcal{D}^\dagger(\psc)
	=\tfrac{1}{2} \F_{\sigma} [ K(\cdot) \mathcal{D}(\cdot) ](-\psc)$
	follows from \eqref{generalizedfourier} (Theorem \ref{convolutionproposition}) and
	$\F_\sigma [a_\sigma(\cdot)](\psc) = a(\psc)$ is applied.\qed
\end{proof}
This result motivates the following extension of Eq.~\eqref{eq:def_weyl_quant_general}:
\begin{definition}
The quantization $\textup{Op}_{\theta}(\mathfrak a)$ of any $\mathfrak a \in \mathcal{S}'(\mathbb{R}^2)$
is defined to be
$$
\textup{Op}_{\theta}(\mathfrak{a})=(\pi\hbar)^{-1}\mathfrak{a}(\mathcal{D}\Pi_\theta\mathcal{D}^\dagger)\
$$
in the sense of Eq.~\eqref{eq:weyl_quant_def}.
\end{definition}

One can consider single Fourier components
$e^{i (p_0 x - x_0 p)/\hbar}=:f_{\psc_0}(\psc)$, for which the Weyl quantization yields
the displacement operator $\textup{Op}_{\textup{Weyl}}( f_{\psc_0}) = \mathcal{D}(\psc_0)$
from Sec.~\ref{translations}, refer to Proposition 51 in \cite{thewignertransform} or
Proposition 11 in Sec.~6.3.2 of \cite{bornjordan}.
Let us now consider the $\theta$-type quantization of a single Fourier component,
which results in the displacement operator being multiplied by the corresponding filter function
via \eqref{thetaquant1}-\eqref{thetaquant}.
Substituting $a_\sigma(\psc)=\F_\sigma[f_{\psc_0}](\psc)$ into \eqref{thetaquant1},
one obtains
\begin{equation*}
\textup{Op}_{\theta}( f_{\psc_0}) =  K_{\theta}(\psc_0) \mathcal{D}(\psc_0)
\quad \text{and} \quad
\parity_{\theta} = (4 \pi \hbar)^{-1}\int  \textup{Op}_{\theta}( f_{\psc_0})  \, \mathrm{d}\psc_0.
\end{equation*}
The second equality
follows from \eqref{genparitydef} and it
 specifies the parity operator as a phase-space average of quantizations
of single Fourier components.

But it is even more instructive to consider the case of the delta distribution
$\delta^{(2)}
$, the Weyl quantization of which yields the
Grossmann-Royer parity operator $\textup{Op}_{\textup{Weyl}}( \delta^{(2)} ) =  \pi \hbar\, \parity$
(as obtained in \cite{Grossmann1976}).
Applying \eqref{thetaquant1}, the Cohen quantization
of the delta distribution yields  the parity
operator from \eqref{genparitydef}. In particular, the operator $\parity_{\theta}$ from
Definition~\ref{def_parity_operator}
is a $\theta$-type quantization
of the delta distribution  as
\begin{equation} \label{quantizationofdelta}
\parity_{\theta}
= 
( \pi \hbar)^{-1} \, \textup{Op}_{\theta}( \delta^{(2)})
=
(\pi \hbar)^{-1} \, \textup{Op}_{\textup{Weyl}}( \theta^\vee ),
\end{equation}
or equivalently, the Weyl quantization of the Cohen kernel, up to
coordinate reflection.
Since $\parity_{\theta}$ is the Weyl quantization of the tempered distribution
$\theta^\vee
\in \mathcal{S}'(\mathbb{R}^2)$, one can adapt results
contained in \cite{daubechies1980distributions} to precisely state conditions
on $\theta$, for which bounded operators $\parity_{\theta}$ 
are obtained via their Weyl quantizations, refer also to Property~\ref{bound}.
For example, square-integrable $\theta \in L^2(\mathbb{R}^2)$ result in Hilbert-Schmidt operators $\parity_{\theta}$,
absolutely integrable $\theta\in L^1(\mathbb{R}^2)$ result in compact operators $\parity_{\theta}$,
and Schwartz functions $\theta\in \mathcal{S}(\mathbb{R}^2)$ result in trace-class operators $\parity_{\theta}$,
refer to \cite{daubechies1980distributions}.

We consider now a class of explicit quantization schemes along the lines of
\cite{Cahill68,cahill1969,Agarwal68,bornjordan} and they are motivated by
different $(\tau,s)$-orderings
of non-commuting operators $\hat{x}$ and $\hat{p}$ or
$\ahat$ and $\ahatdagg$
($-1 \leq s \leq 1$ and $0 \leq \tau \leq 1$).
This class is obtained via the $(\tau,s)$-parametrized
filter function (where 
the relation $\alpha= (\lambda x + i p/\lambda)/\sqrt{2 \hbar}$ 
from Section~\ref{translations} is used)
\begin{equation}
\label{generalfilterfunction}
K_{\tau s}(\psc) := \exp{[ \tfrac{2\tau-1}{4} (\alpha^2 {-} (\alpha^*)^2) +  \tfrac{s}{2} |\alpha|^2]}
=\exp{[ \tfrac{i(2\tau-1) px}{2 \hbar}  +  \tfrac{s( \lambda^2 x^2 + \lambda^{-2} p^2)}{4 \hbar} ]},
\end{equation}
which admits the symmetries $K_{\tau s}(\psc) = K_{\tau s}(-\psc)$ and $ K_{\tau s}^*(x,p)=K_{\tau s}(x,-p) $.
The corresponding $(\tau,s)$-parametrized quantizations of a single Fourier component
are given by the operators $\textup{Op}_{\tau{}s}( f_{\psc_0}) :=  K_{\tau{}s}(\psc) \mathcal{D}(\psc)$,
which are central in ordered expansions
into non-commuting operators.
Also note that for $s \leq 0$, the resulting parity operators are bounded as is readily verified;
hence the corresponding distribution functions
are in the Cohen class with $K_{\tau s}(\psc) \in  \mathcal{S}'(\mathbb{R}^2)$
due to Theorem~\ref{convolutionproposition}.
Important, well-known special
cases are summarized in Table~\ref{quanttable},
refer also to \cite{cahill1969,agarwal1970,Agarwal68,bornjordan}.

\begin{table}
	\centering
	\begin{tabular}{lccc@{\hspace{1mm}}r}
		Ordering& $(\tau,s)$ & $K_{\tau{}s}(\psc_0)$ &  $\textup{Op}_{\tau{}s}( f_{\psc_0})$ \\[1mm]
		\hline \\[-3mm]
		Normal     &   $(\tfrac{1}{2},1)$    &  $e^{|\alpha_0|^2/2}$      &   $e^{\alpha_0 \ahatdagg } e^{-\alpha_0^* \ahat }$    \\[3mm]
		Antinormal &   $(\tfrac{1}{2},-1)$   &  $e^{-|\alpha_0|^2/2}$      & $e^{-\alpha_0^* \ahat }  e^{\alpha_0 \ahatdagg }$   \\[3mm]
		Weyl       &   $(\tfrac{1}{2},0)$    &   1     &  
		$e^{\alpha_0 \ahatdagg -\alpha_0^* \ahat} =e^{\tfrac{i}{\hbar} p_0 \hat{x} - \tfrac{i}{\hbar} x_0 \hat{p}} $ \\[3mm]
		Standard   &   $(1,0)$    &  $e^{\tfrac{i}{ 2 \hbar} p_0 x_0}$      
		&  $e^{\tfrac{i}{\hbar}  p_0 \hat{x}} e^{- \tfrac{i}{\hbar} x_0 \hat{p}} $      \\[3mm]
		Antistandard &   $(0,0)$    & $e^{-\tfrac{i}{2 \hbar} p_0 x_0 }$       
		&   $ e^{- \tfrac{i}{\hbar} x_0 \hat{p}} e^{\tfrac{i}{\hbar} p_0 \hat{x}}$  	\\[3mm]
		Born-Jordan &   $\int_{0}^{1}(\tau,0)\,\mathrm{d}\tau$    &  $\mathrm{sinc}(p_0 x_0/2)$      
		&  refer to Sec.~\ref{bornjordansection}    \\[2mm]
		\hline	
	\end{tabular}
	\caption{\label{quanttable} 
	Common operator orderings, their defining filter functions $K_{\tau s}(\psc_0)$, and the corresponding
	single Fourier component quantizations $\textup{Op}_{\tau{}s}( f_{\psc_0})$ as displacement operators with 
	$e^{i (p_0 x - x_0 p)/\hbar}=:f_{\psc_0}(\psc)$, refer to, e.g., \cite{cahill1969,Royer77,agarwal1970,Agarwal68,bornjordan}.
	Coordinates $\psc_0 \simeq \alpha_0 \simeq (x_0,p_0)$ with subindex $0$ are used for clarity.
	}
\end{table}

\subsection{Explicit Form of Parity Operators\label{explicitparitysection}}
Expectation values of displaced parity operators 
\begin{equation} \label{tausparity} 
\parity_{\tau s} = 
(4 \pi \hbar)^{-1}\int K_{\tau s}(\psc) \mathcal{D}(\psc) \, \mathrm{d}\psc
=
(4 \pi \hbar)^{-1}\int  \textup{Op}_{\tau s}( f_{\psc_0})  \, \mathrm{d}\psc
\end{equation}
are obtained via 
the kernel function in \eqref{generalfilterfunction}
and recover well-known phase-space distribution functions\footnote{
This family of phase-space representations
is related to the one considered in \cite{agarwal1970,Agarwal68}
by setting $\lambda=s/2$ and $\mu = -\nu={2\tau-1}/{4}$.}
for particular cases of $\tau$ or $s$, which are motivated by the 
ordering schemes $\textup{Op}_{\tau s}( f_{\psc_0})$ from Table~\ref{quanttable}.
Important special cases
of these distribution functions and their corresponding filter functions
and Cohen kernels are summarized in Table~\ref{filterfunctionstable}.

In particular, the parameters $\tau=1/2$ and $s=0$ identify the Wigner function
with $K_{1/2, 0}(\psc) \equiv 1$
and \eqref{tausparity} reduces to \eqref{paritydef}.
Note that the corresponding Cohen kernel $\theta$
from Theorem \ref{convolutionproposition} is the $2$-dimensional delta distribution
$\delta^{(2)}(\psc)$ and that convolving with 
$\delta^{(2)}(\psc)$
is the identity operation, i.e.~$\delta^{(2)} \ast W_\rho = W_\rho$ [see \eqref{convolutioneq}].

\begin{table}
	\centering
	\begin{tabular}{lccc}
		\hline\noalign{\smallskip}
		Name & $(\tau, s)$ & $K_{\tau s}(\psc)$  & $\theta_{\tau s}(\psc)$ \\
		\noalign{\smallskip}\hline\noalign{\smallskip}
		Wigner function &  $(1/2,0)$ & 1 & $\delta^{(2)}(\psc)$ \\[2mm]
		$s$-parametrized&  $(1/2,s)$ & $\exp{[   \tfrac{s}{2} |\alpha|^2]}$ &
		$- \tfrac{1}{\pi s} \exp{[   \tfrac{2}{s} |\alpha|^2]}$ \\[2mm]
		Husimi Q function&  $(1/2,-1)$ & $\exp{[  - \tfrac{1}{2} |\alpha|^2]}$ &
		$ \tfrac{1}{\pi} \exp{[- 2 |\alpha|^2]}$ \\[2mm]
		Glauber P function&  $(1/2,1)$ & $\exp{[  \tfrac{1}{2} |\alpha|^2]}$ &
		$ - \tfrac{1}{\pi} \exp{[ 2 |\alpha|^2]}$ \\[2mm]
		Shubin's $\tau$-distribution&  $(\tau,0)$ & $\exp{[ \tfrac{i}{\hbar} \tfrac{2\tau-1}{2} px ]}$ &
		$\tfrac{1}{\hbar  \pi|2 \tau -1|}\exp{[  \tfrac{2i}{\hbar (2\tau-1)} px ]}$ \\[2mm]
		Born-Jordan distribution&  $ \int_0^1 (\tau,0) \, \mathrm{d}\tau$ &
		$\mathrm{sinc}[px/(2\hbar)]$ & $\F_\sigma\{\mathrm{sinc}[px/(2\hbar)]\}/(2\pi \hbar)$
		\\[2mm]
		\noalign{\smallskip}\hline
	\end{tabular}
	\caption{\label{filterfunctionstable} Well-known phase-space distribution functions
		and their corresponding Cohen kernels recovered for particular values of $\tau$ or $s$
		via expectation values of displaced parity operators from \eqref{tausparity}.
	}
\end{table}

The filter function $K_{\tau s}$ from \eqref{generalfilterfunction} 
results for a fixed parameter of $\tau=1/2$
in the Gaussian $K_{s}(\psc) := K_{1/2, s}(\psc) = \exp{[   \tfrac{s}{2} |\alpha|^2]}$.
The corresponding parity operators are diagonal in the number-state representation
(refer to Property~\ref{rotationalcovariance}),	
and they can be specified for $-1 \leq s < 1$ in terms of number-state projectors 
\cite{cahill1969,moya1993,Royer77} as
\begin{equation} \label{sparametrizedparityform} 
\parity_s := \parity_{1/2,s} =
(4 \pi \hbar)^{-1} \int e^{s |\alpha|^2/2} \, \mathcal{D}(\psc) \, \mathrm{d}\psc = 
\sum_{n=0}^\infty (-1)^n \frac{(1{+}s)^n}{(1{-}s)^{n+1}} | n \rangle \langle n | ,
\end{equation}
where the second equality specifies $\parity_s$ in the form of a spectral decomposition.
This form has implicitly appeared in, e.g., \cite{cahill1969,moya1993,Royer77}.
We provide a more compact proof in Appendix~\ref{proofofsparametrized}.
Equation \eqref{sparametrizedparityform} readily implies
$\|\parity_s\|_\indsupnorm=(1{-}s)^{-1}$ for $s\leq 0$, and for $s<0$ one even finds that $\parity_s$ are trace-class operators due to
\begin{align*}
\|\parity_s\|_1=\sum_{n=0}^\infty\frac{(1{+}s)^n}{(1{-}s)^{n+1}}=\frac{1}{(1{-}s)-(1{+}s)}=(2|s|)^{-1}.
\end{align*}
Note that for $s>0$
the corresponding filter functions lie outside of our framework as then $K_s\not\in\mathcal S'(\mathbb R^2)$ 
due to its superexponential growth.
While one can still formally write down their distribution functions, one runs into convergence problems
resulting in singularities. However, their symplectic Fourier transform
always exists and it is related to the Wigner function via
$K_s(\psc) \F_\sigma [W_\rho](\psc)$ by
multiplying with the filter function $K_{s}(\psc)$
(cf.~Remark \ref{rem_cohen_fourier}).
This class of $s$-parametrized phase-space representations
has gained widespread applications in quantum optics and beyond 
\cite{mandel1995,glauber2007,zachos2005,schroeck2013,Curtright-review},
and they correspond to Gaussian convolved Wigner functions 
\begin{equation*}
\FF_\rho(\psc, s) 
= \FF_{|0\rangle}(\psc, s{+}1) \ast W_\rho(\psc),
\end{equation*}
for $s<0$
such as the Husimi Q function for $s=-1$. 
Note that the Cohen kernel $\theta_{s} $ via Theorem \ref{convolutionproposition} corresponds to 
the vacuum state $ \FF_{|0\rangle}(\psc, s{+}1) $ of a quantum harmonic oscillator \cite{cahill1969}.
Gaussian deconvolutions of the Wigner function are formally obtained for $s>0$,
which includes the Glauber P function for $s=1$ \cite{cahill1969}.
Due to the rotational symmetry of its filter function $K_{s}(\psc)$, the $s$-parametrized distribution
functions are covariant under phase-space rotations, refer to Property~\ref{rotationalcovariance}.

Another important special case is obtained for a fixed parameter of $s=0$,
which results in Shubin's $\tau$-distribution, refer to
\cite{bornjordan,boggiatto2010time,boggiatto2010weighted,boggiatto2013hudson}.
Its filter function from \eqref{generalfilterfunction} reduces to the chirp function
$K_{\tau{}0}(\psc) =: K_{\tau}(\psc)  =\exp{[ i (2\tau{-}1) px /(2 \hbar) ]}$ while
relying on the parametrization with $x$ and $p$.
The resulting distribution functions
$\FF_{\rho}(\psc,\tau)$
are in the Cohen class due to
Theorem \ref{convolutionproposition} and 
they are square integrable following 
Property~\ref{l2property} as
the absolute value of
$  K_{\tau}(\psc)$ is bounded.
We calculate the explicit action of the corresponding parity
operator $\parity_{\tau}$.
\begin{theorem} \label{tauparitycoordinate}
	The action of the $\tau$-parametrized parity operator $\parity_\tau :=\parity_{\tau{}0}$ 
	on some coordinate
	representation $\psi(x) \in L^2(\mathbb{R})$ is explicitly given for any $\tau\neq 1$ by
	\begin{equation} \label{tauparity}
	\parity_\tau \psi(x) =
	\frac{ 1}{2 |\tau{-}1|}   \psi(\tfrac{\tau x }{ \tau {-} 1}),
	\end{equation}
	which for the special case $\tau = 1/2$ reduces (as expected) to the usual parity operator $\parity$.
	It follows that $\parity_\tau$ is bounded for every $0 < \tau < 1$
	(or in general for every real $\tau$ that is not equal to $0$ or $1$)
	and its operator norm is given by $\| \parity_\tau \|_\indsupnorm = 1/\sqrt{4 (\tau {-} \tau ^2)}$. 
\end{theorem}
\begin{proof}
By \eqref{tausparity}, the parity operator $\parity_{\tau}$ acts on the coordinate representation $\psi(x)$ via
	\begin{equation*} 
	\parity_{\tau} \psi(x) =
	(4 \pi \hbar)^{-1}  \int  K_{\tau{}0}(\psc) \,   \mathcal{D}(\psc)  \psi(x) \, \mathrm{d}\psc.
	\end{equation*}
	This integral can be evaluated using the explicit form of $K_{\tau{}0}(\psc) $
	form \eqref{generalfilterfunction} and the action of $\mathcal{D}$ on coordinate representations $\psi(x)$
	from \eqref{dispopdefxp} yields
	\begin{align*}
	\int  K_{\tau{}0}(\psc) \,   \mathcal{D}(\psc)  \psi(x) \, \mathrm{d}\psc
	& =
	\int e^{\tfrac{i}{2 \hbar} (2\tau-1) p_0 x_0 } e^{\tfrac{i}{\hbar}(p_{0}x-\tfrac{1}{2}p_{0}x_{0})}
	\psi(x{-}x_{0}) \, \mathrm{d}x_0 \, \mathrm{d}p_0 
	\\
	& =
	\int\left[
	\int e^{\tfrac{i}{\hbar} p_{0}(x+(\tau-1)x_{0})}\, \mathrm{d}p_0 
	\right]  \psi(x{-}x_{0}) \, \mathrm{d}x_0 \\
	& =
	\frac{1}{|\tau{-}1|} \int\left[
	\int e^{\tfrac{i}{\hbar} p_{0}y}\, \mathrm{d}p_0 
	\right]  \psi(\tfrac{\tau x {-} y}{ \tau {-} 1}) \, \mathrm{d}y,
	\end{align*}
	where the change of variables $y=x+(\tau{-}1)x_{0}$ with $x_{0}=(y{-}x)/(\tau{-}1)$ and $\mathrm{d}y=
	|\tau{-}1|\,\mathrm{d}x_0 $
	was used.	 Therefore, the right-hand side is 
	\begin{equation*}
	\frac{ 2\pi\hbar}{|\tau{-}1|} \int\delta(y)\, \psi(\tfrac{\tau x {-} y}{ \tau {-} 1}) \, \mathrm{d}y
	=
	\frac{ 2\pi\hbar}{|\tau{-}1|} \psi(\tfrac{\tau x }{ \tau {-} 1}).
	\end{equation*}
	Now let $\tau\in(0,1)$.
	The operator norm
	is calculated via
	$$\|\parity_\tau \|_{\indsupnorm} = \mathrm{sup}_{\| \phi(x) \|_{\indLnorm}=1} \| \parity_\tau \phi(x) \|_{\indLnorm}.$$
	For an arbitrary square-integrable  $\phi(x)$ with $L^2$ norm $\| \phi(x) \|_{\indLnorm}=1$ one obtains 
	\begin{align*}
	\| \parity_\tau \phi \|^2_{\indLnorm}
	&= 
	\langle \parity_\tau \phi |  \parity_\tau \phi \rangle
	= 
	(2 |\tau{-}1|)^{-2} \int_{\mathbb{R}}   \phi^*(\tfrac{\tau x }{ \tau - 1}) 	   \phi(\tfrac{\tau x }{ \tau - 1})    \, \mathrm{d}x\\
	&=(2 |\tau{-}1|)^{-2} \tfrac{ | \tau - 1 |}{| \tau |} \|\phi(x) \|_{\indLnorm}
	=\tfrac{ 1}{ 4 |\tau-1| | \tau |} =
	 \tfrac{1 }{ 4 (\tau - \tau ^2)}	
	\end{align*}
	by applying a change of variables, which results in $\|\parity_\tau \|_{\indsupnorm} = [ 4 (\tau {-} \tau ^2)]^{-1/2}$.\qed
\end{proof}
This parity operator is bounded for every $0 < \tau < 1$,
and its expectation value gives rise to well-defined distribution functions (Property~\ref{bound}), 
which are also integrable as $K_{\tau}(0)=1$ (Property~\ref{cohencl}). 
Note that this family of distribution functions $\FF_\rho(\psc,\tau,0)$ for $\tau \neq 1/2$ 
does not satisfy 
Property~\ref{reality}, i.e.~self-adjoint operators $\rho$ are mapped to complex functions.
In particular, the symmetry $K^*_{\tau}(\psc) = K_{1-\tau}(\psc)$ implies that
\begin{equation*}
\FF_{\rho^\dagger}(\psc,\tau) = \FF^*_\rho(\psc, 1{-} \tau) .
\end{equation*}
In the following, we will rely on this $\tau$-parametrized family
to construct and analyze the parity operator of the Born-Jordan distribution.

\section{The Born-Jordan Distribution\label{bornjordansection}}

\subsection{Parity-Operator Description of the Born-Jordan Distribution}
The Born-Jordan distribution $\FF_\rho(\psc, \textup{BJ}) $
is an element of the Cohen class \cite{cohen1966generalized,bornjordan,boggiatto2010time}
and is obtained
by averaging over the $\tau$-distributions
$\FF_{\rho}(\psc,\tau) \in  L^2(\mathbb{R}^2)$:
\begin{equation}
\label{bjaveragedefinition}
\FF_\rho(\psc, \textup{BJ}) := \int_0^1 \FF_\rho(\psc,\tau) \, \mathrm{d}\tau.
\end{equation}
As in Definition~\ref{def_phase_space_function}, this distribution function is
also obtained via the expectation value of a parity operator.
\begin{theorem} \label{BJdistributionTheorem}
	The Born-Jordan distribution $\FF_\rho(\psc, \textup{BJ})$
	of a density operator $\rho\in\mathcal B^1(\mathcal H)$
	is an element of the Cohen class, and it
	is obtained as the expectation value
	\begin{equation}
	\label{BJdefinition}
	\FF_\rho(\psc,\textup{BJ})
	= (\pi \hbar)^{-1} \, \tr\,[ \, \rho \, \mathcal{D}(\psc)  \parity_{\textup{BJ}}  \mathcal{D}^\dagger(\psc) ]
	\end{equation}
	of the (displaced) parity operator $\parity_{\textup{BJ}}$ that is 
	defined by the relation
	\begin{equation} \label{BJparityDef}
	\parity_{\textup{BJ}} := (4 \pi \hbar)^{-1} \int K_{\textup{BJ}}(\psc) \mathcal{D}(\psc) \, \mathrm{d}\psc,
	\end{equation}
	where $K_{\textup{BJ}}(\psc) = \mathrm{sinc}(a)=\sin(a)/a$ is the cardinal sine function
	with the argument $a=(2\hbar)^{-1} \, px	= i [(\alpha^*)^2{-}\alpha^2]/4$.
	Here one applies the substitution $\alpha= (\lambda x + i p/\lambda)/\sqrt{2 \hbar}$
	from Sec.~\ref{translations} and the expression for $a$ is independent of $\lambda$.
\end{theorem}
\begin{proof}
	Combining Eqs.~\eqref{bjaveragedefinition} and \eqref{genpsdef},
	the Born-Jordan distribution is the expectation value
	\begin{equation*}
	\FF_\rho(\psc, \textup{BJ}) =
	(\pi \hbar)^{-1} \, \tr\,[ \, \rho \, \mathcal{D}(\psc)  
	(\int_0^1 \parity_{\tau{}0} \, \mathrm{d}\tau ) 
	\mathcal{D}^\dagger(\psc) ],
	\end{equation*}
	and the corresponding parity operator can be expanded as
	\begin{equation*} 
	\parity_{\textup{BJ}}=
	(4 \pi \hbar)^{-1} \int \left[ \int_0^1  K_{\tau{}0}(\psc) \, \mathrm{d}\tau \right]  \mathcal{D}(\psc) \, \mathrm{d}\psc.
	\end{equation*}
	Using the explicit form of $K_{\tau{}0}(\psc)$ from \eqref{generalfilterfunction},
	the evaluation of the integral 
	\begin{equation*}
	\int_0^1  K_{\tau{}0}(\psc) \, \mathrm{d}\tau
	=
	\int_0^1 \exp{[ \tfrac{i}{2 \hbar} (2\tau{-}1) px ]} \, \mathrm{d}\tau = \mathrm{sinc}[(2\hbar)^{-1} \, px]
	\end{equation*}
	over $\tau$ 
	concludes the proof. \qed
\end{proof}

This confirms that the Born-Jordan distribution
$\FF_\rho(\psc, \textup{BJ}) \in  L^2(\mathbb{R}^2)$ is square integrable
following Property~\ref{l2property}
as the absolute value of its filter function is bounded,
i.e., $| \mathrm{sinc}[(2\hbar)^{-1} \, px] \,| \leq 1$ for all $(x,p)\in\mathbb R^2$.
The filter function $K_{\textup{BJ}}$ satisfies
$K_{\textup{BJ}}(x,0) =K_{\textup{BJ}}(0,p)=1$, and the Born-Jordan distribution
therefore gives rise to the correct marginals as quantum-mechanical probabilities
(Property~\ref{marginal}).
In particular, integrating over the Born-Jordan distribution reproduces
the quantum-mechanical probability densities,
i.e.~$
\int \FF_\rho(x,p,\textup{BJ}) \,\mathrm{d} x = |\psi(p)|^2 
$
and
$
\int \FF_\rho(x,p,\textup{BJ}) \,\mathrm{d} p = |\psi(x)|^2
$.

Most importantly 
the operator $\parity_{\textup{BJ}}$ is bounded, meaning Born-Jordan
distributions are well defined and bounded for all quantum states,
refer to Property~\ref{bound}. Also, the largest (generalized) 
eigenvalue of $\parity_{\textup{BJ}}$ is exactly $\pi/2$ as shown in
Theorem \ref{BJdecomposition} below.
\begin{proposition} \label{BJparityBound}
        The Born-Jordan 
	parity operator $\parity_{\textup{BJ}}$ is bounded as 
	an upper bound of its operator norm is given by 
	$\|\parity_{\textup{BJ}}\|_\indsupnorm \leq \pi/2$.
\end{proposition}
\begin{proof}
Using the $\parity_\tau$-representation of $\parity_{\textup{BJ}}$ we compute
\begin{align*}
\|\parity_{\textup{BJ}}\psi(x)\|_{L^2}= \big\|\int_0^1 \parity_\tau\psi(x) \, 
\mathrm{d}\tau\big\|_{L^2}&\leq \int_0^1 \|\parity_\tau\psi(x)\|_{L^2} \, \mathrm{d}\tau\\
&\leq\int_0^1 \frac{1}{\sqrt{4(\tau{-}\tau^2)}} \, \mathrm{d}\tau\,\|\psi(x)\|_{L^2}=\frac{\pi}{2}\|\psi(x)\|_{L^2}
\end{align*}
for arbitrary $\psi(x)\in L^2(\mathbb R)$ where in the second-to-last step we used Theorem \ref{tauparitycoordinate}.\qed
\end{proof}
It is well-known that the Born-Jordan distribution is related to the
Wigner function via a convolution with the Cohen kernel $\theta_{\textup{BJ}}$,
refer to \cite{bornjordan,thewignertransform}.
However, calculating this kernel, or the corresponding parity operator directly
might prove difficult.
In the following, we establish a more convenient representation of the
Born-Jordan parity operator which is an ``average'' of $\parity_\tau$
from Theorem \ref{tauparitycoordinate} via the formal integral transformation
\begin{equation}\label{parity_bj_integral_rep}
\parity_{\textup{BJ}} = \int_0^1 	\parity_\tau \, \mathrm{d}\tau,
\end{equation}
which---as in Sect.~\ref{reflections}---is interpreted 
as $\parity_{\textup{BJ}} \psi(x) = \int_0^1 	\parity_\tau  \psi(x) \, \mathrm{d}\tau$
for all $\psi(x)\in L^2(\mathbb R)$. 
Recall  that the parity operator $\parity_\tau$ is well defined and bounded for every $0 < \tau < 1$.
\begin{remark}
Note that evaluating $\parity_{\textup{BJ}}\psi(x)$ at $x=0$ for some $\psi(x)\in L^2(\mathbb R)$ with $\psi(0)\neq 0$ 
leads to a divergent integral in \eqref{parity_bj_integral_rep}. This comes from the singularity at $\tau=1$ in \eqref{tauparity}. 
However, we will later see that this is harmless as it only happens on a set of measure zero 
(so one can define $\parity_{\textup{BJ}}\psi(x)|_{x=0}$ to be $0$ or $\psi(0)$ or arbitrary) and, 
more importantly, that $\psi(x)\in L^2(\mathbb R)$ implies $\parity_{\textup{BJ}}\psi(x)\in L^2(\mathbb R)$ 
(Proposition \ref{BJparityBound}).
\end{remark}
Let us explicitly specify the squeezing operator  
\begin{equation}\label{sq_eq}
S(\xi) := \exp[ \tfrac{i \xi}{2} ( \hat{p} \hat{x} {+} \hat{x}\hat{p}) ]= \exp[  \tfrac{i \xi}{2} ( \ahat^2 {-} (\ahatdagg)^2) ]
\end{equation}
while following \cite{chruscinski2004} and Chapter 2.3 in \cite{leonhardt97}. It acts on a coordinate representation 
$\psi(x)$ via $S(\xi) \psi(x) = e^{\xi/2} \psi (e^{\xi} x)$.

\begin{theorem} \label{BJparitySqueezing}
	Let us consider the squeezing
	operator $S(\xi)$ which depends on the real squeezing parameter $\xi \in \mathds{R}$.	
	The Born-Jordan parity operator 
	\begin{equation}\label{eq:pi_bj_composition}
	\parity_{\textup{BJ}} = 
	[\,
	\tfrac{1}{4}	 \int_{-\infty}^{\infty}
	\,
	\mathrm{sech}(\xi/2)
	S(\xi) 
	\,\mathrm{d} \xi
	\,] 
	 \; \parity
	\end{equation}
	is a composition of the reflection operator $\parity$ followed by a squeezing operator
	(and the two operations commute),
	and this expression is integrated with respect to a well-behaved weight function $\mathrm{sech}(\xi/2)
	 =2/(e^{\xi/2}{+}e^{-\xi/2})$.
	Note that the function $\mathrm{sech}(\xi/2)  \in \mathcal{S}(\mathbb{R}) $
	is fast decreasing and invariant under the Fourier
	transform (e.g., as Hermite polynomials).
\end{theorem}
\begin{proof}
	The explicit action of $\parity_{\textup{BJ}}$ on a coordinate representation
	$\psi(x) \in L^2(\mathbb{R})$ is given by (see Theorem \ref{tauparitycoordinate})
	\begin{equation*}
	\parity_{\textup{BJ}} \psi(x) = \int_0^1 	\parity_\tau \psi(x) \, \mathrm{d}\tau
	= \int_0^1 \frac{ 1}{2 |\tau{-}1|}   \psi(\tfrac{\tau x }{ \tau - 1}) \, \mathrm{d}\tau.
	\end{equation*}
	Applying a change of variables $e^{\xi} = \tau/(1{-}\tau)$
	with $\xi \in \mathbb{R}$ yields the substitutions
	$\tau = 1/(1 {+} e^{-\xi})$, ${1}/(2 |\tau{-}1|)  =  (1 {+} e^\xi )/2$.
	and $\mathrm{d} \tau = e^\xi/(1{+}e^\xi)^2\,\mathrm{d} \xi $.
	One obtains
	\begin{equation*}
	\parity_{\textup{BJ}} \psi(x) = 
	\int_{-\infty}^{\infty}
	\frac{e^{\xi}}{2(1{+}e^{\xi})}
	\psi ( - e^{\xi} x)
	\,\mathrm{d} \xi.
	\end{equation*}
	Let us recognize that $\psi ( - e^{\xi} x)= e^{- \xi/2} S(\xi)\,  \parity \, \psi(x)$ is the composition of a coordinate
	reflection and a squeezing of the pure state $\psi(x)$; also the two operations commute.
	This results in the explicit action
	\begin{equation*}
	\parity_{\textup{BJ}} \psi(x) = 
	\int_{-\infty}^{\infty}
	\frac{e^{\xi/2} }{2(1{+}e^{\xi})}
	  S(\xi)\,   
	\parity \, \psi(x) \,\mathrm{d} \xi ,
	\end{equation*}
	where $ e^{\xi/2} /[ 2(1{+}e^{\xi}) ] = [ 2(e^{-\xi/2} {+}e^{\xi/2}) ]^{-1}$
	concludes the proof.\qed
\end{proof}

The expression for the parity operator in Theorem \ref{BJparitySqueezing}
is very instructive when compared to Theorem \ref{BJdistributionTheorem},
and this confirms that the parity operator $\parity_{\textup{BJ}}$
decomposes into the
usual parity operator $\parity$ followed by a geometric transformation, refer also to
Section~\ref{sec_geom_parity}. In the case of the 
Born-Jordan parity operator this geometric transformation is
an average of squeezing operators.

\begin{remark}
	The Born-Jordan distribution is covariant under squeezing, which means that the squeezed
	density operator $\rho' = S(\xi) \rho S^\dagger(\xi)$ is mapped to
	the inversely squeezed phase-space representation $\FF_{\rho'}(x,p, \textup{BJ}) = \FF_{\rho}(e^{-\xi} x, e^{\xi} p, \textup{BJ}) $.
\end{remark}
The form of $\parity_{\textup{BJ}}$ given in \eqref{eq:pi_bj_composition} allows for
an alternative proof of Proposition \ref{BJparityBound}:
\begin{remark}
Recalling
	$
			\|\parity_{\textup{BJ}}\|_\indsupnorm = \sup_{\|\phi(x)\|_{\indLnorm}=1}
		\| \parity_{\textup{BJ}} \phi(x) \|_{\indLnorm},	
	$
	the norm of the function $\parity_{\textup{BJ}} \phi(x)$ for
	any $\phi(x)$ with $\|\phi(x)\|_{\indLnorm}=1$ can be expressed as
	\begin{align*}
		\| \parity_{\textup{BJ}} \phi(x) \|_{\indLnorm}^2 & = 
		\langle \phi | \parity_{\textup{BJ}}^\dagger \parity^{\phantom{\dagger}}_{\textup{BJ}} | \phi \rangle \\
		& =\tfrac{1}{16}	 \iint
		\mathrm{sech}(\xi/2) \mathrm{sech}(\xi'/2)
		\langle \phi | S(\xi'{-}\xi)   | \phi \rangle
		\,\mathrm{d} \xi \, \mathrm{d} \xi'\\
		& \leq
		\tfrac{1}{16}	 \iint
		\mathrm{sech}(\xi/2) \mathrm{sech}(\xi'/2)
		| \langle \phi | S(\xi'{-}\xi)   | \phi \rangle|
		\,\mathrm{d} \xi \, \mathrm{d} \xi',
	\end{align*}	
	and it was used that $\parity^\dagger\parity=\parity^2 =1$ and $S^\dagger(\xi) S(\xi') =  S(\xi'{-}\xi)  $.
	Since $S(\xi'{-}\xi)   $ is unitary, one obtains that
	$|\langle \phi | S(\xi'{-}\xi) \phi \rangle| \leq 1$. Finally, 
	\begin{equation*}
	\quad\| \parity_{\textup{BJ}} \phi \|_{\indLnorm}^2
	\leq
	\tfrac{1}{16}	 \iint
	\mathrm{sech}(\xi/2) \mathrm{sech}(\xi'/2)
	\,\mathrm{d} \xi \, \mathrm{d} \xi'
	=
	\pi^2/4.
	\end{equation*}
\end{remark}
\subsection{Spectral Decomposition of the Born-Jordan Parity Operator \label{BJdecomp}}

We will now adapt results for generalized spectral decompositions,
refer to \cite{chruscinski2003,chruscinski2004,mauringeneral,vilenkin1964generalized}.
This will allow us to solve the generalized eigenvalue equation for parity operators
and to determine their spectral decompositions.

Recall the distributional pairing for smooth, well-behaved functions
$\psi(x) \in \mathcal{S}(\mathbb R)$  in Section~\ref{functionfourier}
with respect to 
tempered distributions $a\in\mathcal S'(\mathbb R)$ (such as
functions of slow growth $a(x)$).
We will use this distributional pairing to construct 
$L^2$ scalar products
of the form $ \langle a , \psi \rangle  = \int_{\mathbb R}a^*(x)\psi(x)\,\mathrm{d}x$, which 
corresponds to a rigged Hilbert space
\cite{vilenkin1964generalized,chruscinski2003} or the Gelfand triple 
$\mathcal{S}(\mathbb{R})\subset L^2 (\mathbb{R}) \subset \mathcal S'(\mathbb R)$.
This rigged Hilbert space allows us to specify the generalized spectral decomposition
of the Born-Jordan parity operator with generalized eigenvectors in $\mathcal S'(\mathbb R)$
as functions of slow growth.

It was shown in the previous section that the Born-Jordan
parity operator $\parity_{\textup{BJ}}$ is a composition
of a coordinate reflection and a squeezing operator.
We now recapitulate results on the spectral decomposition of the squeezing operator from
\cite{chruscinski2003,chruscinski2004,bollini1993shannon}, up to minor modifications.
Recall that the squeezing operator forms a unitary, strongly continuous one-parameter
group  $S(\xi) = e^{-i \xi {H}}$ with $\xi\in\mathbb R$ that is generated by the (unbounded) self-adjoint Hamiltonian 
\begin{equation*}
{H}:= - \tfrac{1}{2} (\hat{x} \hat{p} {+} \hat{p} \hat{x}) = - \tfrac{1}{2} [ \ahat^2 {-} (\ahatdagg)^2].
\end{equation*}
This Hamiltonian admits a purely continuous spectrum $E \in \mathbb{R}$,
and satisfies generalized eigenvalue equations
\begin{equation*}
\langle H \psi | \psi^E_{\pm} \rangle =  \langle \psi | H  \psi^E_{\pm} \rangle = E \langle \psi | \psi^E_{\pm} \rangle
\end{equation*}
for every $\psi \in \mathcal{S}(\mathbb{R})$, 
where the last equation is equivalent to $ H | \psi^E_{\pm} \rangle = E  | \psi^E_{\pm}\rangle$.
The Gelfand-Maurin spectral theorem \cite{mauringeneral,vilenkin1964generalized,chruscinski2003}
results in a spectral resolution of 
\begin{equation*}
H = \int_{-\infty}^{\infty} E \, 
|  \psi^E_{\pm}  \rangle \langle \psi^E_{\pm}  | 
\, \mathrm{d} E.
\end{equation*}
Here the generalized eigenvectors are specified in terms of
their coordinate representations as slowly increasing functions,
i.e.~$\psi_\pm^E(x) := \langle x | \psi^E_\pm \rangle \in {S}'(\mathbb{R})$ with
\begin{equation} \label{eigenvectors}
\psi_+^E(x)  =  
\tfrac{1}{2\sqrt{\pi}} |x|^{-(iE + \tfrac{1}{2})} \; \text{ and } \;
\psi_-^E (x) = 
\tfrac{1}{2\sqrt{\pi}} \mathrm{sgn}(x) |x|^{-(iE + \tfrac{1}{2})},
\end{equation}
refer to \cite{chruscinski2003,chruscinski2004,mauringeneral,vilenkin1964generalized} 
and Appendix~\ref{appendixspectral} for more details.
Note that $\psi_\pm^E(x)$ are generalized eigenfunctions: they are 
not square integrable, but the integral 
$\int_{\mathbb R} [\psi_\pm^E(x)]^* \phi(x)\,\mathrm{d}x$ exists as a distributional pairing for every
$ \phi \in \mathcal S(\mathbb R)$.
Also note that these generalized eigenvectors can be decomposed into the 
number-state basis with finite expansion coefficients that decrease to zero
for large $n$, refer to Appendix~\ref{appendixspectral}.
The spectral decomposition of the squeezing operator is then given by
\begin{equation*}
S(\xi)	=
  \int_{-\infty}^{\infty}  e^{-i E \xi}
[ \, | \psi^E_{+} \rangle \langle \psi^E_{+} |
+| \psi^E_{-} \rangle \langle \psi^E_{-} | \, ]
\, \mathrm{d} E,
\end{equation*}
refer to Eq.~6.12 in \cite{chruscinski2003} and Eq.~2.14 in \cite{chruscinski2004}.
 Note that these eigenvectors
are also invariant under the Fourier transform (e.g., as Hermite polynomials).

It immediately follows that the squeezing operator satisfies the generalized
eigenvalue equation
\begin{equation} \label{squeezeingeigenvalue}
S(\xi)  | \psi^E_{\pm} \rangle  =  e^{-i \xi E}  \, | \psi^E_{\pm} \rangle,
\end{equation}
which can be easily verified using the explicit action
$S(\xi) \psi_\pm^E(x) = e^{\xi/2} \psi_\pm^E(e^{\xi} x) = e^{-iE \xi} \psi_\pm^E(x)$.
One can now specify the
Born-Jordan parity operator using its spectral decomposition.
\begin{theorem} \label{BJdecomposition}
	Generalized eigenvectors $| \psi^E_{\pm} \rangle $ of the squeezing operator
	from \eqref{eigenvectors} are also generalized eigenvectors of the
	Born-Jordan parity operator which satisfy
	\begin{equation*}
	\parity_{\textup{BJ}} \, | \psi^E_{\pm} \rangle	=  \pm  \tfrac{\pi}2 \, \mathrm{sech}(\pi E)  \,  | \psi^E_{\pm} \rangle
	\end{equation*}
	for all $E\in\mathbb R$. The parity operator $\parity_{\textup{BJ}}$ therefore admits the spectral decomposition
	\begin{equation*}
	\parity_{\textup{BJ}}	=  \tfrac{\pi}2 \, \int_{-\infty}^{\infty} \mathrm{sech}(\pi E) \, [ \, 
	| \psi^E_{+} \rangle \langle \psi^E_{+} | 
	-  | \psi^E_{-} \rangle \langle \psi^E_{-} |
	\, ]
	\, \mathrm{d}E,
	\end{equation*}
	where $ \langle \psi^E_{\pm} | \parity = \pm \langle \psi^E_{\pm} |$ has been used.
\end{theorem}
\begin{proof}
	The generalized eigenvalues can be computed via
	\begin{equation*}
	\parity_{\textup{BJ}}	 \, | \psi^E_{\pm} \rangle	=
		\tfrac{1}{4}	 \int_{-\infty}^{\infty}	\,
	\mathrm{sech}(\xi/2) 	S(\xi)  \,\mathrm{d} \xi
	\; \parity \, | \psi^E_{\pm} \rangle	,
	\end{equation*}
	where $\parity \, | \psi^E_{\pm} \rangle = \pm  | \psi^E_{\pm} \rangle$. 
	Using \eqref{squeezeingeigenvalue}, one obtains
	\begin{equation*}
	\quad\parity_{\textup{BJ}}	 \, | \psi^E_{\pm} \rangle	=
	\pm \tfrac{ 1}{4}	 \int_{-\infty}^{\infty}	\,
	\mathrm{sech}(\xi/2) 	e^{-iE \xi}  \,\mathrm{d} \xi
	\;  | \psi^E_{\pm} \rangle	=  \pm   \tfrac{\pi}2  \, \mathrm{sech}(\pi E)  \, | \psi^E_{\pm} \rangle .\tag*{\qed} 
	\end{equation*}
\end{proof}

\begin{remark}
Recall that $\parity_{\textup{BJ}}$ is a bounded (by Proposition \ref{BJparityBound}) and self-adjoint 
operator. Consequently, the usual spectral theorem in multiplication operator form \cite[Thm.~7.20]{hall2013quantum} 
yields a $\sigma$-finite measure space $(X,\mu)$, a bounded, measurable, real-valued function $h$ 
on $X$, and unitary $U:L^2(\mathbb R)\to L^2(X,\mu)$ such that
\begin{equation*}
[U\parity_{\textup{BJ}}U^{-1}(\psi)](\lambda)=h(\lambda)\psi(\lambda)
\end{equation*}
for all $\psi\in L^2(X,\mu)$ and $\lambda\in X$. While this undoubtedly 
is a nice representation, the spectral decomposition in Theorem \ref{BJdecomposition}
is more readily determined with the help of
the Gelfand-Maurin spectral theorem \cite{mauringeneral,vilenkin1964generalized,chruscinski2003}.
In particular, said theorem lets us directly work
with the generalized eigenfunctions in Eq.~\eqref{eigenvectors}, 
even though they are not square integrable.
\end{remark}

\subsection{Geometric Interpretation of Parity Operators\label{sec_geom_parity}}
While above we have comprehensively explored analytic properties of the Born-Jordan and other practically
	important parity operators, here we relate these mathematical objects to geometric transformations.
	Even the rather complex Born-Jordan parity operator admits a surprisingly
	simple decomposition into two elementary geometric transformations.
	Equation~\eqref{eq:pi_bj_composition} decomposes the  Born-Jordan parity operator
	into an ordinary reflection of the wave function's coordinate followed by a weighted
	average of squeezing operations as
	$$\parity_{\textup{BJ}}
	=[\,
	\tfrac{1}{4}	 \int_{-\infty}^{\infty}
	\,
	\mathrm{sech}(\xi/2)
	S(\xi) 
	\,\mathrm{d} \xi
	\,] 
	\, \parity.
	$$
As such, the action on any wave function $\psi(x) \in L^2(\mathbb{R})$  can be summarized as
the reflected, squeezed function $S(\xi)\parity \, \psi(x) =  e^{\xi/2} \psi (- e^{\xi} x)$
averaged over all parameters $\xi \in (-\infty, \infty)$ with respect to the rapidly
decaying weight function $\mathrm{sech}(\xi/2)$.

It is not only the Born-Jordan parity operator that admits a simple geometric 
interpretation, but it rather seems to hint at a universal property, at least in the classes of
practically important phase-space representations.
In particular, we now state that both $\parity_{\tau}$
and the pivotal parity operator $\parity_s$,
which contains the most popular variants of Wigner, Husimi and Glauber P
phase-space functions as special cases,
can be decomposed into elementary geometric transformations.

\begin{remark} \label{tausqueeze} 
	Applying the substitution $e^{\xi} := \tau/(1{-}\tau)$,
	the parity operator $\parity_\tau$ 
	from  \eqref{tauparity} can be decomposed 
	for $0 <\tau < 1$ into
	$$\parity_\tau = \cosh(\xi/2)   S(\xi)\,  \parity$$ which consists of
	a coordinate reflection and a squeezing.
\end{remark}

Consequently, the parity operator $\parity_\tau$ admits a spectral decomposition
\begin{equation*}
	\parity_\tau  =  \cosh(\xi/2)  \, \int e^{-i E \xi}
	[ \, | \psi^E_{+} \rangle \langle \psi^E_{+} |
	-| \psi^E_{-} \rangle \langle \psi^E_{-} | \, ]
	\, \mathrm{d} E,
\end{equation*}
where  $ \langle \psi^E_{\pm} | \parity = \pm \langle \psi^E_{\pm} |$ has been used.

\begin{remark}
The parity operator with $\kappa_s:=\ln[(1{+}s)/(1{-}s)]$ and $-1 <s < 1$
is the composition
$$\parity_s=(1{-}s)^{-1} e^{\kappa_s \ahatdagg \ahat}  \, \parity$$ from \eqref{sparametrizedparityform} 
of the usual coordinate reflection $\parity$ followed by a positive semi-definite operator.
In particular, $\parity_{-1}= \tfrac{1}{2} \, | 0 \rangle \langle 0 | \, \parity$.
Note that the positive semi-definite operator $e^{\kappa_s \ahatdagg \ahat}$ describes
the effective phenomenon of photon loss for $s < 0$, refer to \cite{Leonhardt93}.
Of course domain restrictions might need to be considered for $s>0$ as discussed earlier.
\end{remark}

\section{Explicit Matrix Representation of the Born-Jordan Parity Operator \label{BJmatrix}}

Recall that the $s$-parametrized parity operators $\parity_s$ are diagonal in  the 
Fock basis and their diagonal entries
can be computed using the simple expression 
in \eqref{sparametrizedparityform}. This enables the experimental reconstruction of distribution
functions from photon-count statistics \cite{deleglise2008,Lutterbach97,Bertet02,Banaszek99}
in quantum optics.
\begin{remark}
	The Born-Jordan parity operator $\parity_{\textup{BJ}}$
	is not diagonal in the number-state basis, as its filter
	function $K_{\textup{BJ}}(\psc)$ is not invariant under
	arbitrary phase-space rotations, refer to
	Property~\ref{rotationalcovariance}. The filter
	function $K_{\textup{BJ}}(\psc)$ is, however, invariant
	under $\pi/2$ rotations in phase space, and therefore
	only every fourth off-diagonal is non-zero.
\end{remark}
We now discuss the number-state representation of the
parity operator $\parity_{\textup{BJ}}$,
which provides a convenient way to calculate (or, more precisely, approximate) Born-Jordan distributions.
\begin{theorem} \label{BJParityMatrix}
	The matrix elements
	$[ \parity_{\textup{BJ}} ]_{mn} := \langle m | \parity_{\textup{BJ}}\, n \rangle$
	 of the Born-Jordan parity operator in the Fock basis
	can be calculated in the form of a finite sum
	\begin{equation} \label{BJparitymatrix}
	[ \parity_{\textup{BJ}} ]_{mn}
	=
	\sum_{k=0}^n \sum_{\substack{\ell=0 \\ \textrm{$\ell$ even}} }^{m-n}  d_{mn}^{k\ell} \;
	\mathrm{\Phi}_{(m-n-\ell)/2,\ell/2}^k \, ,
	\end{equation}
	for $m \in \{n, n+4, n+8, \dots \}$ and $[ \parity_{\textup{BJ}} ]_{mn} = [ \parity_{\textup{BJ}} ]_{nm}$
	with the coefficients
	\begin{align}
	d_{mn}^{k\ell} \label{dcoefficient}
	& :=  (-1)^{(\ell+m-n)/2} \sqrt{\frac{n!}{m!}} 2^{(2 k+m-n)/2} \binom{m}{n{-}k} \binom{m{-}n}{\ell} /k!,
	\\ \label{derivatives}
	\mathrm{\Phi}_{ab}^k &:=
	[  \partial^k_{\mu} [\partial^{a}_{\lambda}\partial^{b}_{\mu} f(\lambda,\mu)]|_{\lambda=\mu} ]|_{\mu=1}.
	\end{align}
	Here, $\mathrm{\Phi}_{ab}^k$ denotes the $a$th and $b$th partial
	derivatives of the function $f(\lambda,\mu) =  \mathrm{arcsinh}[1/\sqrt{\lambda \mu}]$
	with respect to its variables $\lambda$ and $\mu$, respectively, evaluated at $\lambda=\mu$,
	then differentiated again $k$ times and finally its variable is set to $\mu=1$.
\end{theorem}

Refer to Appendix~\ref{BJcomputation} for a proof.
The derivatives in \eqref{derivatives} can be calculated in the form
of a finite sum
\begin{equation}\label{xi_expansion}
\mathrm{\Phi}_{ab}^k
=
\sqrt{2} (-1)^{a+b} \,  2^{-2a -2b -k}
\sum_{j=0}^{a+b+k-1} \xi_j^{a{}b{}k},
\end{equation}
where $a+b+k\geq 1$ and $\xi_j^{a{}b{}k}$ are recursively defined integers, 
refer to \eqref{eq:88b} in Appendix~\ref{derivativescalculation}.
Substituting $\ell$ for $2\ell$ in \eqref{BJparitymatrix}, the matrix elements $[ \parity_{\textup{BJ}}]_{mn}$
then depend only on these integers $\xi_j^{a{}b{}k}$ via the finite sum
\begin{equation*}
[ \parity_{\textup{BJ}}]_{mn}
=
\gamma_{mn} \sum_{k=0}^n \sum_{\ell=0}^{(m-n)/2}  [\sum_j \xi_j^{a{}\ell{}k} ] \,
(-1)^{\ell +m -n} \binom{m}{n{-}k} \binom{m{-}n}{2 \ell}/k! ,
\end{equation*}
where $\gamma_{mn}:= 2^{[-(m-n)+1]/2} \sqrt{{n!}/{m!}}  $ for  $m \in \{n, n{+}4, n{+}8, \dots \}$
and  $a = (m{-}n{-} 2\ell)/2$. 

\begin{figure}
	\centering
	\includegraphics{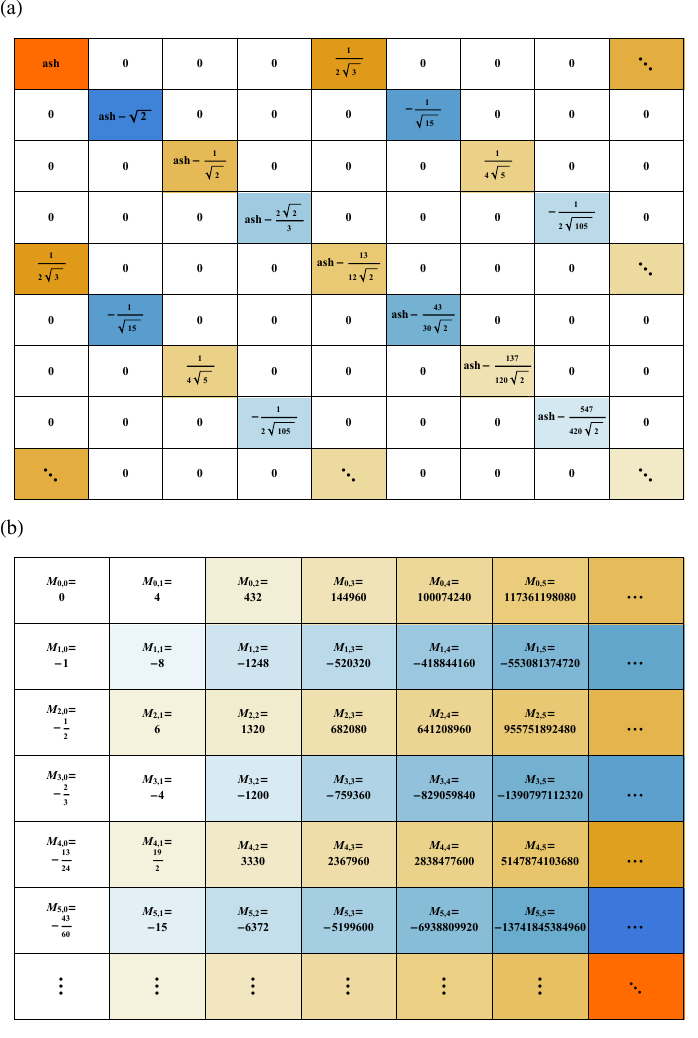}
	\caption{ (a) The first $8 \times 8$ matrix elements of the Born-Jordan parity operator
		from Eq.~\eqref{BJparitymatrix} where $\mathrm{ash}$ denotes $\mathrm{arcsinh}(1)$.
		(b) The first $6 \times 6$ elements of the recursive sequence that defines the
		Born-Jordan parity operator via Eq.~\eqref{BJoprtorationalnumbersrecursion}.
		Orange and blue colors represent positive and negative values, respectively,
		while the color intensity reflects the absolute value of the corresponding numbers.
		\label{BJmatrixPlot}}
\end{figure}

Figure~\ref{BJmatrixPlot}(a) shows the first $8 \times 8$ entries
of $[ \parity_{\textup{BJ}}]$. One observes the following structure:
only every fourth off-diagonal is non-zero, the matrix is real and symmetric,
and the entries along every diagonal and off-diagonal decrease in their absolute value.
In particular, the diagonal elements of $\parity_{\textup{BJ}}$ admit the following special property.
\begin{proposition}\label{prop_diag_ele}
	For every $n\in\lbrace 0,1,2,\ldots\rbrace$, the diagonal entries of $\parity_{\textup{BJ}}$ in the Fock basis are 
	\begin{equation}\label{eq:M_nn_finite_sum_1}
	[ \parity_{\textup{BJ}}]_{nn}=\operatorname{arcsinh}(1)-\sqrt{2}
	\sum_{k=0}^{n-1} \frac{(-1)^k}{k{+}1} \sum_{m=0}^{\lfloor \frac{k}{2}\rfloor} \binom{2m}{m}\Big(\frac{-1}{4}\Big)^{m}.
	\end{equation}
	In particular, $[ \parity_{\textup{BJ}}]_{nn}\to 0$ as $n\to\infty$. For a proof we refer to Appendix \ref{proof_prop_diag_ele}.
\end{proposition}
Also note that the sum of these decreasing diagonal entries results in a
trace $\tr[\parity_{\textup{BJ}}] = 1/2$ (Property~\ref{cohencl}) in the number-state basis.
However, this trace does not necessarily exist in an arbitrary basis,
as $\parity_{\textup{BJ}}$ is not a trace-class operator.

\begin{remark}
Let us emphasize that the boundedness of $\parity_{\textup{BJ}}$ (Proposition \ref{BJparityBound})
guarantees that using a (large enough) finite block of $\parity_{\textup{BJ}}$ for computations yields a good approximation.
The reason for this is that such a block does not differ too much (in trace norm) from the full operator,
which readily transfers to the phase-space distribution function. Details can be found
at the end of Appendix \ref{app_phase_space_extension}.
\end{remark}

In the following, we specify a more convenient form for the
calculation of these matrix elements, i.e.~a direct recursion
without summation, which is based on the following conjecture 
(see Appendix~\ref{app:directrecursive}). 

\begin{conjecture} \label{BJParityMatrixRecursion}
	The non-zero matrix elements 
	\begin{equation} \label{BJoprtorationalnumbersrecursion}
	[ \parity_{\textup{BJ}} ]_{k+4\ell,k} =  [ \parity_{\textup{BJ}} ]_{k, k+4\ell} =
	\Gamma_{k\ell}\, [  M_{k\ell} + \delta_{\ell0} \, \mathrm{arcsh}(1)/\sqrt{2} ],
	\end{equation}	
	of the Born-Jordan parity operator
	are determined by a set of rational numbers $M_{k\ell}$ where
	$\Gamma_{k\ell} =  2^{-2\ell + {1}/{2}} \sqrt{{k!}/{(k{+}4\ell)!}}$
	and $k,\ell \in \{0,1,2, \dots\}$. The indexing is specified relative to
	the diagonal (where $\ell=0$) and $\delta_{\ell{}m}$ is the Kronecker delta.
	The rational numbers $ M_{k\ell} $ can be calculated
	recursively using only 8 numbers as initial conditions, refer to
	Appendix~\ref{app:directrecursive} 	for details. This form does not require a
	summation. 
\end{conjecture} 

Figure~\ref{BJmatrixPlot}(b) shows the first $6 \times 6$ elements
of the recursive sequence of rational numbers $M_{k\ell}$. The first
column of $M_{k0}$ corresponds to the diagonal of the matrix
$[\parity_{\textup{BJ}}]_{mn}$ from Figure~\ref{BJmatrixPlot}(a).
For example, for $k=5$ one obtains $M_{5,0} = -{43}/{60}$,
which corresponds to $[\parity_{\textup{BJ}}]_{5,5}=\Gamma_{5,0} 
[ M_{5,0} + \delta_{0,0} \, \mathrm{arcsh}(1)/\sqrt{2} ]$
and  $\Gamma_{5,0}= \sqrt{2}$, and therefore
$[\parity_{\textup{BJ}}]_{5,5}=-{ 43 \sqrt{2}}/{60} +\mathrm{arcsh}(1) $
as detailed in Figure~\ref{BJmatrixPlot}(a).

The direct recursion in Conjecture~\ref{BJParityMatrixRecursion} enables us to
conveniently and efficiently calculate the matrix elements
$[\parity_{\textup{BJ}}]_{mn}$
and we have verified the correctness of this approach for up to
$6400$ matrix elements,
i.e.~by calculating a matrix representation of size $80 \times 80$.
This facilitates an efficient calculation and plotting of  Born-Jordan distributions 
for harmonic oscillator systems,
such as in quantum optics \cite{mandel1995,glauber2007,leonhardt97}.
Note that a recursively calculated $80 \times 80$ matrix representation,
which we have verified with exact calculations,
is sufficient for most physical applications,
i.e.~Figures~\ref{purestates} and \ref{OffDiagonal} (below) were calculated using $30 \times 30$
matrix representations. However, a matrix representation of size $2000 \times 2000$
can be easily calculated on a current notebook computer using the recursive method.
Numerical evidence shows that the matrix
representation of $\parity_{\textup{BJ}}$ can be well-approximated by a low-rank
matrix, i.e, diagonalizing the matrix $\parity_{\textup{BJ}}$ reveals only 
very few significant eigenvalues. In particular, the sum of squares
of the first $9$ eigenvalues corresponds to approximately $99.97\%$ of the sum of squares of all the eigenvalues
of a $2000 \times 2000$ matrix representation.

\begin{figure}[t]
	\centering
	\includegraphics{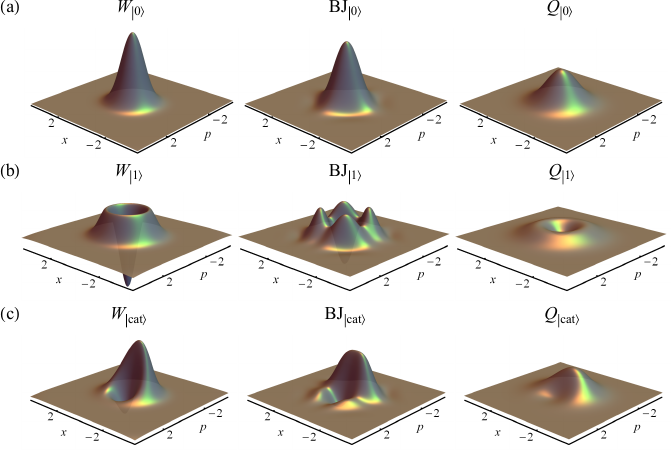}
	\caption{ \label{purestates}
	(a) and (b) Wigner, Born-Jordan, and Husimi Q phase-space plots of the number states $|0\rangle$
	and $|1\rangle$ (which are eigenstates of the quantum harmonic oscillator).
	(c) Corresponding phase-space plots of the Schrödinger cat state $|\textup{cat}\rangle = (|0\rangle {+} |1\rangle)/\sqrt{2}$
	(which is a superposition of the previous states).
	}
\end{figure}

\section{Example Quantum States\label{sec_ex}}

Matrix representations of  parity operators are used to conveniently
calculate phase-space representations for bosonic quantum states
via their associated 
Laguerre polynomial decompositions.
The $s$-parametrized distribution functions of
Fock states $|n\rangle$ are sums
\begin{equation*}
\FF_{|n\rangle} (\psc,s) = \sum_{\mu=0} |[\mathcal{D}(\psc)]_{n \mu}|^2 \, [\parity_s]_{\mu \mu}
\end{equation*}
of the associated Laguerre polynomials from \eqref{displacementopmatirx}, which are weighted by their parity operator
elements. The corresponding phase-space functions are radially symmetric as $|[\mathcal{D}(\psc)]_{n \mu}|^2$
depends only on the radial distance $x^2+p^2$.
The Wigner functions in Fig.~\ref{purestates}(a)-(b) are radially symmetric and
show strong oscillations, which are sometimes regarded as a quantum-mechanical feature \cite{leonhardt97}.

\begin{figure}
	\centering
	\includegraphics{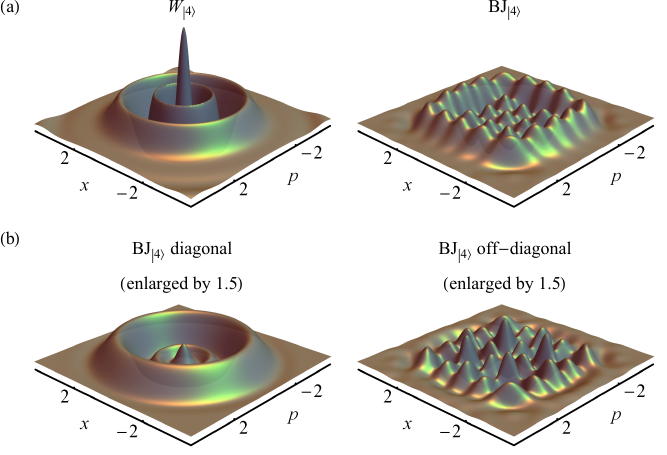}
	\caption{ \label{OffDiagonal}
				(a) Wigner and Born-Jordan phase-space plots of the number state $|4\rangle$.
				(b) The Born-Jordan distribution is decomposed into functions that correspond
				to diagonal and off-diagonal entries of the parity operator matrix. 
			}
\end{figure}

In contrast, the Born-Jordan parity operator is not diagonal
in the number state representation and
it can be written in terms of projectors as 
$
\parity_{\textup{BJ}}= \sum_{\mu=0}[ \parity_{\textup{BJ}} ]_{\mu \mu}  |\mu\rangle \langle \mu |+
\sum_{\mu=0} \sum_{\nu=1} [ \parity_{\textup{BJ}} ]_{\mu, 4\nu} \,  
(  |\mu {+} 4 \nu \rangle \langle \mu |  +  |\mu \rangle \langle \mu {+} 4 \nu | ),
$
and the Born-Jordan distribution of number states $|n\rangle$ is
given by
\begin{align}
\FF_{|n\rangle} (\psc,\textup{BJ}) 
&= \sum_{\mu=0}   |[\mathcal{D}]_{n \mu}|^2 \, [\parity_{\textup{BJ}} ]_{\mu \mu} \label{FBJI}\\
&+
\sum_{\mu=0} \sum_{\nu=1} 2 \, \Re
( [\mathcal{D}]_{n \mu} ([\mathcal{D}]_{n, \mu + 4\nu})^* ) [ \parity_{\textup{BJ}} ]_{\mu,\mu+4\nu}.
\label{FBJII}
\end{align}
The Born-Jordan distribution of coherent states, i.e.~the displaced vacuum states,
closely matches the Wigner functions, see Fig.~\ref{purestates}(a).
The first part in Eq.~\eqref{FBJI} contains the diagonal elements of the parity operator which correspond
to the radially symmetric part of $\FF_{|n\rangle} (\psc,\textup{BJ})$, see Fig.~\ref{OffDiagonal}(b) (left).
The second part in Eq.~\eqref{FBJII} results in a radially non-symmetric function,
see Fig.~\ref{OffDiagonal}(b) (right).
The radially symmetric parts are quite similar to the Wigner function
and have $n+1$ wave fronts enclosed by the Bohr-Sommerfeld band \cite{leonhardt97,dowling1991},
i.e.~the ring with radius $\sqrt{2n{+}1}$.
The radially non-symmetric functions have $n+1$ local maxima
along the outer squares,
i.e.~along phase-space cuts at the Bohr-Sommerfeld distance $x,p \propto \sqrt{2n{+}1}$.
The sum of these two contributions is the Born-Jordan distribution
and it is not radially symmetric for number states, see Fig.~\ref{OffDiagonal}(a).

\section{Conclusion \label{conclusion}}

We have introduced parity operators $\parity_\theta$
which give rise to a rich family of phase-space distribution
functions of quantum states. These phase-space functions have
been previously defined in terms of convolutions, integral transformations,
or Fourier transformations. 
Our approach using parity operators is both conceptually and
computationally advantageous and now allows for a direct
calculation of phase-space functions as quantum-mechanical expectation values.
This approach
therefore averts the necessity for the repeated and expensive computation of Fourier transformations.
We motivate the name ``parity operator'' by the fact that parity operators 
$\parity_\theta = A_\theta \circ \parity$ are composed of
the usual parity operator and some specific
geometric or physical
transformation. We detailed the explicit form of parity operators for various phase spaces
and, in particular, for the Born-Jordan distribution.
We have also obtained a generalized spectral decomposition
of the Born-Jordan parity operator, proved its boundedness,
and explicitly calculated its matrix representation in the 
number-state basis.
We conjecture that these matrix elements are determined by
a proposed recursive scheme which allows for an efficient 
computation of Born-Jordan distribution functions.
Moreover, large matrix representations of the Born-Jordan parity
operator can be well approximated using rank-9 matrices.
All this will be useful to connect
our results with applications 
in (e.g.) quantum optics, where techniques such as
squeezing operators and the number-state representation are
widely used.

\begin{acknowledgements}
F.~vom Ende thanks GH and M.~Alekseyev for providing the idea for the 
current proof of Proposition \ref{prop_diag_ele}
in a discussion on MathOverflow \cite{MO295019}.
We thank Michael Keyl and Gunther Dirr for their valuable comments.
B.~Koczor acknowledges financial support from the scholarship program of the 
Bavarian Academic Center for Central, Eastern and Southeastern Europe (BAYHOST),
funding from the EU H2020-FETFLAG-03-2018 under grant agreement No 820495 (AQTION)
and from the Glasstone Research Fellowship of the University of Oxford.
This research is funded by the Bavarian excellence network \textsc{enb}
via the International PhD Programme of Excellence
\textit{Exploring Quantum Matter} (\textsc{exqm}), the \textit{Munich Quantum Valley} of the Bavarian
State Government with funds from Hightech Agenda \textit{Bayern Plus}, and
the Deutsche Forschungsgemeinschaft (DFG, German Research Foundation) 
under Germany’s Excellence Strategy -- EXC-2111 -- 390814868.
M.~de Gosson has been financed by the Austrian Research Foundation FWF (Grant number P27773).
R.~Zeier acknowledges funding from the EU H2020-FETFLAG-03-2018 under grant agreement No 817482 (PASQuanS)
and from the European High-Performance Computing Joint Undertaking (JU) under grant agreement No 101018180. 
The JU receives support from the European Union’s Horizon 2020 research 
and innovation programme and Germany, France, Italy, Ireland, Austria, Spain.
\end{acknowledgements}

\vfill

\appendix

\section{Extension of Operators from Schwartz Functions to Distributions}\label{app_operator_extension}
Assume we have a linear operator $T:\mathcal S(\mathbb R)\to\mathcal S
(\mathbb R)$ or, more generally, $T:\mathcal S(\mathbb R)\to(\mathbb R\to
\mathbb C)$ and we want to extend its action to tempered distributions. Usually 
this is done by introducing some operator $\hat T:\mathcal S(\mathbb R)\to(
\mathbb R\to\mathbb C)$ (which can---but does not have to---be the same as $T$) 
such that
\begin{equation}\label{eq:extension_def}
(T\phi)(\psi):=\phi(\hat T\psi)
\end{equation}
for all $\phi\in D_T:=\{\phi\in\mathcal S'(\mathbb R):\phi\circ\hat T\in
\mathcal S'(\mathbb R)\}$ and all $\psi\in\mathcal S(\mathbb R)$. Usually $\hat T$ is 
chosen such that \eqref{eq:extension_def} is consistent with the action of $T$ on 
those distributions which are 
generated by Schwartz functions.
More precisely one fixes some injective, usually linear or antilinear map
$e:\mathcal S(\mathbb R)\to \mathcal S'(\mathbb R)$ and requires that $\hat T$ is 
chosen such that $[T,e]\equiv 0$, i.e.
$$
e(\phi)(\hat T\psi)=\big(T(e(\phi))\big)(\psi)=\big((T\circ e)(\phi)\big)(\psi)=\big((e
\circ T)(\phi)\big)(\psi)=e(T\phi)(\psi)
$$
for all $\phi,\psi\in\mathcal S(\mathbb R)$ with $e(\phi)\in D_T$. Thus by identifying 
functions $\phi\in\mathcal S(\mathbb R)$ with the tempered distribution $e(\phi)
\in\mathcal S'(\mathbb R)$ an extension is defined such that there is only a ``formal'' 
difference between the action of $T$ on $\phi$ compared to the action of $T$ on 
$e(\phi)$.
Let us illustrate this by means of a simple example:
\begin{example}\label{ex_extend_displacement}
Consider the displacement operator from Eq.~\eqref {dispopdefxp} which acts 
on any function $\psi:\mathbb R\to\mathbb C$---so in particular on any Schwartz 
function---via
$[\mathcal{D}(x_0,p_0) \psi](x)=\exp[{{i}p_0(x{-}x_{0}/2)/\hbar}]
\psi(x{-}x_{0})$ for all $x,x_0,p_0\in\mathbb R$. Depending on how one defines the 
distributional pairing of two classical functions there are two ways to extend
$\mathcal D$ to $\mathcal S'(\mathbb R)$:

\noindent (1) The usual way of embedding $\mathcal S(\mathbb R)
\hookrightarrow\mathcal S'(\mathbb R)$ is done via the linear map $
\iota(\psi):=\int\psi(x)(\cdot)(x)\,\mathrm{d}x$. Then for all $\phi,\psi\in
\mathcal S(\mathbb R)$ and all $\psc\in\mathbb R^2$ one finds
\cite[Eq.~(1.11)]{thewignertransform}
\begin{align*}
\iota(\phi)\big(\mathcal D(\Omega)\psi\big)=\iota\big(\mathcal D^{\wedge x}(\psc)
\phi\big)\psi\qquad\text{ or, equivalently }\qquad
\iota\big(\mathcal D(\psc)\phi\big)\psi=
\iota(\phi)\big(\mathcal D^{\wedge x}(\Omega)\psi\big)
\end{align*}
where $\mathcal D^{\wedge x}(x_0,p_0):=\mathcal D(-x_0,p_0)$ for all 
$x_0,p_0\in\mathbb R$. This suggests setting $\hat T:=\mathcal D^{\wedge x}$ in 
\eqref{eq:extension_def} because this way 
one has $[\mathcal D(\psc),\iota]\equiv 0$, that is, $\mathcal D(\psc)(\iota(\phi))
\equiv\iota(\mathcal D(\psc)\phi)$ for all $\psc\in\mathbb R^2$ and all $\phi:
\mathbb R\to\mathbb C$ such that $\int\phi(x)(\cdot)(x)\,\mathrm{d}x\in
\mathcal S'(\mathbb R)$.
In other words this extension of the displacement operator is consistent with its 
action on $\mathcal S(\mathbb R)$ by means of the embedding $\iota$.

\noindent  (2) One may also consider the canonical (antilinear, bijective) map from 
$L^2(\mathbb R)$ to its dual space $(L^2(\mathbb R))^*$ from the Riesz 
representation theorem which acts via $\langle\,\cdot\,|\;(\psi):=\langle\psi|\,\cdot\,
\rangle:=\int\psi(x)^*(\cdot)(x)\,\mathrm{d}x$. One readily verifies that $\langle\phi |
\mathcal D(\psc)\psi\rangle =\langle\mathcal D(-\psc)\phi | \psi\rangle $ for 
all $\phi,\psi\in L^2(\mathbb R)$,
i.e.~$D(\psc)^\dagger=\mathcal D(-\psc)$ for all
$\psc\in\mathbb R^2$. Setting $\hat T:=\mathcal D(\psc)^\dagger$ in 
\eqref{eq:extension_def} thus yields an extension of $\mathcal D$ which is consistent with 
respect to $\langle\,\cdot\,|\;$: one finds $[\mathcal D(\psc)\langle\,\cdot\,|\;]
\equiv 0$, that is, for all $\phi,\psi\in\mathcal S(\mathbb R)$
$$
\mathcal D(\psc)(\langle\phi,\cdot\,\rangle)(\psi)\overset{\eqref{eq:extension_def}}
=\langle\phi,\mathcal D(-\psc)\psi\rangle=\langle\mathcal D(\psc)\phi,\psi\rangle=\langle 
\mathcal D(\psc)\phi,\cdot\,\rangle(\psi).
$$
\end{example}

Having introduced the concept of operator extensions we may apply it to generalized 
parity operators.
But first let us generalize Definition \ref{def_parity_operator} to arbitrary tempered
distributions $\theta$, even though this is beyond what we need in the main sections of this article.
\begin{remark}\label{rem_general_def_parity}
\emph{Formally} \eqref{genparitydef} can be rewritten as $\Pi_\theta=(4\pi\hbar)^{-1}\langle K_\theta^*,\mathcal D\rangle$.
Because admissible kernels by definition satisfy $\langle K_\theta^*, \psi \rangle= 2\pi\hbar [\F_\sigma(\theta)]\psi $ 
for $\psi \in \mathcal S(\mathbb R)$ this
leads to an extension of Definition~\ref{def_parity_operator} to arbitrary $\theta\in\mathcal S'(\mathbb R^2)$ via the linear operator
$\Pi_\theta:\mathcal S
(\mathbb R)\to(\mathbb R\to\mathbb C)$, $\Pi_\theta:=\frac12[\F_\sigma(\theta)
\circ\mathcal D]$, i.e.
$$
(\Pi_\theta\psi)(x)=\tfrac12 \F_\sigma(\theta)(
\mathcal D\psi(x))=\tfrac12\theta\big(\F_\sigma(\mathcal D\psi(x))\big)
$$
for all $\psi\in\mathcal S(\mathbb R)$, $x\in\mathbb R$. Here $\mathcal D\psi(x)
\in\mathcal S(\mathbb R^2)$ defined via $\Omega\mapsto(\mathcal D(\Omega)\psi)
(x)$ is the displacement of $\psi$ at $x$ as a function of $\Omega$.
One readily verifies that for admissible kernels this definition reproduces
\eqref{eq:action_gen_parity} as well as the definition of the parity operator in \eqref{paritydef}.
\end{remark}

Similar to Example \ref{ex_extend_displacement}(1), let us extend $\Pi_\theta$ with respect 
to the embedding $\iota:\mathcal S(\mathbb R)\hookrightarrow\mathcal S'
(\mathbb R)$, that is, we have to find an operator $\tilde\Pi_\theta:
\mathcal S(\mathbb R)\to(\mathbb R\to\mathbb C)$ such that $\iota(\phi)
(\tilde\Pi_\theta\psi)=\iota(\Pi_\theta\phi)(\psi)$ for all $\psi,\phi\in
\mathcal S(\mathbb R)$. We claim that
\begin{equation}\label{eq:Pi_theta_ext_usual_pairing}
\tilde\Pi_\theta:=\tfrac12[\F_\sigma(\theta^{\wedge p})\circ\mathcal D]
\end{equation}
does the job where $\theta^{\wedge p}(a):=\theta(a^{\wedge p})$ and $a^{\wedge p}
(x_0,p_0):=a(x_0,-p_0)$ for all $a\in\mathcal S(\mathbb R^2)$, $x_0,p_0\in
\mathbb R$.
Before we verify this we will first show $
\iota(\phi)\big( \F_\sigma[\mathcal D\psi](x_0,-p_0) \big)=\iota\big( 
\F_\sigma[\mathcal D\phi](x_0,p_0) \big)(\psi)
$ for all $\phi,\psi\in\mathcal S(\mathbb R)$, $x_0,p_0\in\mathbb R$ as an intermediate result. One 
computes
\begin{align*}
& \iota(\phi)\big( \F_\sigma[\mathcal D\psi](x_0,-p_0) \big)
=(2\pi\hbar)^{-1}\int\phi(x) 
\int e^{-\frac{i}{\hbar}(x'(-p_0)-x_0p')} \big(\mathcal D(x',p')\psi\big)
(x)\,\mathrm{d}x'\,\mathrm{d}p'\,\mathrm{d}x\\
&
=(2\pi\hbar)^{-1}\int\phi(x) \int e^{-\frac{i}{\hbar}(x'(-p_0)-x_0p')} e^{\frac{i}\hbar p'(x-
\frac12 x')}\psi(x{-}x') \,\mathrm{d}x'\,\mathrm{d}p'\,\mathrm{d}x\\
&
=(2\pi\hbar)^{-1}\int\phi(x{+}x') \int e^{-\frac{i}{\hbar}((-x')p_0-x_0p')} e^{\frac{i}\hbar 
p'(x+\frac12 x')}\psi(x) \,\mathrm{d}x'\,\mathrm{d}p'\,\mathrm{d}x\\
&
=\int\Big((2\pi\hbar)^{-1}\int e^{-\frac{i}{\hbar}(\tilde xp_0-x_0p')} e^{\frac{i}\hbar p'(x-
\frac12 \tilde x)}\phi(x{-}\tilde x)  \,\mathrm{d}\tilde x\,\mathrm{d}p'\Big)\psi(x)
\,\mathrm{d}x\\
&
=\int  \F_\sigma[\mathcal D\phi(x)](x_0,p_0) \psi(x)\,\mathrm{d}x  = \iota\big( 
\F_\sigma[\mathcal D\phi](x_0,p_0) \big)(\psi).
\end{align*}
Together with the linearity of the integral as well as the linearity and continuity of $\theta$ 
this implies
\begin{align*}
& \iota(\phi)( \tfrac12[\F_\sigma(\theta^{\wedge p})\circ\mathcal D]\psi) = 
\tfrac12\int\phi(x) \theta\big(\big(\F_\sigma (\mathcal D\psi(x))\big)^{\wedge p}\big)\,
\mathrm dx  \\
&=\tfrac12\theta\Big( \int\phi(x) \big(\F_\sigma (\mathcal D\psi(x))\big)^{\wedge p}\,
\mathrm dx \Big)  
=\tfrac12\theta\Big(\int  \F_\sigma(\mathcal D\phi(x)) \psi(x)\,\mathrm{d}x \Big)  \\
&= \tfrac12\int\theta\big(\F_\sigma(\mathcal D\phi(x)) \big)\psi(x)\,\mathrm{d}x  = 
\iota\big(\tfrac12[\F_\sigma(\theta)\circ\mathcal D]\phi\big)(\psi)
\end{align*}
for all $\phi,\psi\in\mathcal S(\mathbb R)$.
Thus by setting $T=\Pi_\theta$ and $\hat T=\tilde\Pi_\theta$ in 
\eqref{eq:extension_def} with $\tilde\Pi_\theta$ from 
\eqref{eq:Pi_theta_ext_usual_pairing}, that is,
$$
(\Pi_\theta\phi)(\psi):=\tfrac12\phi\big[\theta\big(\big(\F_\sigma[\mathcal D\psi]
\big)^{\wedge p}\big)\big]
$$
for all $\psi\in\mathcal S(\mathbb R)$, $\phi\in D_\theta:=\{\phi\in\mathcal S'
(\mathbb R):\phi\big[\theta\big(\big(\F_\sigma(\mathcal D(\cdot))\big)^{\wedge p}\big)
\big]\in\mathcal S'(\mathbb R)\}$
we get an extension of $\Pi_\theta$ which is compatible with the integral pairing $
\iota$ in the sense that $[\Pi_\theta,\iota]\equiv 0$.

Now as in Example \ref{ex_extend_displacement}(2), let us extend $\Pi_\theta$ with respect 
to $\langle\,\cdot\,|:L^2(\mathbb R)\to(L^2(\mathbb R))^*$.
We claim that
\begin{equation}\label{eq:ext_parity_L2}
\tilde\Pi_\theta:=\tfrac12[\F_\sigma(\;{}^*\circ\theta\circ{}^*\;)\circ\mathcal D]
\end{equation}
satisfies $\langle\phi,\tilde\Pi_\theta\psi\rangle=\langle\Pi_\theta\phi,\psi\rangle$ for 
all $\psi,\phi\in\mathcal S(\mathbb R)$, where ${}^*$ is the usual complex conjugate.
Similar as before, one first shows
$$
\int\big(  \F_\sigma[\mathcal D\psi(x)](x_0,p_0)\big)^*\phi(x)
\,\mathrm dx=\int\psi(x)^*\F_\sigma[\mathcal D\phi(x)](x_0,p_0)\,\mathrm dx
$$
via direct calculation in order to conclude that
\begin{align*}
\langle \phi,\tilde\Pi_\theta\psi\rangle&=\tfrac12\int\phi(x)^*\big(\theta\big[ (\F_\sigma
(\mathcal D\psi(x)))^* \big]\big)^*\,\mathrm dx
=\tfrac12\Big( \int\phi(x)\theta\big[ (\F_\sigma(\mathcal D\psi(x)))^* \big] 
\,\mathrm dx\Big)^*\\
&=\tfrac12\Big(\theta\Big[ \int \big(\F_\sigma(\mathcal D\psi(x))\big)^*\phi(x)
\,\mathrm dx\Big]\Big)^*
=\tfrac12\Big(\theta\Big[\int\psi(x)^*\F_\sigma(\mathcal D\phi(x))\,\mathrm dx\Big]
\Big)^*\\
&=\tfrac12\int\big(\theta\big[\F_\sigma(\mathcal D\phi(x))\big]\big)^*\psi(x)\,
\mathrm dx=\langle\Pi_\theta\phi,\psi\rangle
\end{align*}
for all $\phi,\psi\in\mathcal S(\mathbb R)$. Hence
$\tilde\Pi_\theta=\tfrac12[\F_\sigma(\;{}^*\circ\theta\circ{}^*\;)\circ\mathcal D]$ is 
indeed the extension of $\Pi_\theta$ with respect to $\langle\,\cdot\,|$ which we were 
looking for.

In general these two extensions will be different so one has to be careful about which 
framework one uses. However from the explicit form of $\tilde\Pi_\theta$ one knows 
that for any $\theta\in\mathcal S'(\mathbb R^2)$ these extensions coincide if and 
only if $\theta^{\wedge p}\equiv {}^*\circ\theta\circ{}^*\,$. This translates to filter 
functions as follows:
\begin{lemma}\label{lemma_filter_function_symmetry}
Consider any admissible $\theta\in\mathcal S'(\mathbb R^2)$ with associated filter function 
$K_\theta:\mathbb R^2\to\mathbb C$. The extension of $\Pi_\theta$ with 
respect to $\iota$ coincides with the extension of $\Pi_\theta$ with respect to
$\langle\,\cdot\,|\;$ if and only if $K_\theta^*\equiv K_\theta^{\wedge p}$. In this case 
\eqref{eq:extension_def} becomes
\begin{align*}
(\Pi_\theta\phi)(\psi)&=(4\pi\hbar)^{-1}\int K_\theta^*(\Omega)\phi\big( \mathcal D(-\Omega)\psi \big)\,\mathrm d\psc\\
&=(4\pi\hbar)^{-1}\int K_\theta(x_0,p_0)\phi\big( \mathcal D(-x_0,p_0)\psi \big)\,\mathrm dx_0\mathrm dp_0
\end{align*}
for all $\phi\in \mathcal S'(\mathbb R)$ such that
$\int K_\theta^*(\Omega)\phi\big[ \mathcal D(-\Omega)(\cdot) \big]\,\mathrm d\psc\in\mathcal S'(\mathbb R)$,
and all $\psi\in\mathcal S(\mathbb R)$.
\end{lemma}
\begin{proof}
Because $\theta$ is admissible (i.e.~$\theta   =(2\pi\hbar )^{-1}\langle K_\theta^*, \F_\sigma(\cdot)\,\rangle$
for some $K_\theta:\mathbb R^2\to\mathbb C$) we compute
\begin{align*}
& \F_\sigma({}^*\circ\theta\circ{}^*)(a) 
=\big( \theta\big[\F_\sigma(a)^*\big] \big)^*=(2\pi\hbar)^{-1}\Big(\int 
K_\theta(\psc)\F_\sigma\big[\F_\sigma(a)^*\big](\psc)    \,\mathrm d\psc\Big)^*\\
&=(2\pi\hbar)^{-1}\Big(\int K_\theta(\psc)\F_\sigma\big[\F_\sigma(a)\big]^* (-\psc)   \,\mathrm d\psc\Big)^*
=(2\pi\hbar)^{-1}\int K_\theta^*(\psc)a (-\psc)   \,\mathrm d\psc
\end{align*}
 for all $a\in\mathcal S(\mathbb R^2)$. Here we used $\F_\sigma^2=\operatorname{id}$
 as well as the readily verified identity $\F_\sigma[a^*](\Omega)=\F_\sigma[a]^*(-\Omega)$.
If we denote the extension of $\Pi_\theta$ with respect to $\langle\,\cdot\,|$ by $\Pi_{\theta,L^2}$ this implies
 \begin{align*}
 (\Pi_{\theta,L^2}\phi)(\psi)=\tfrac12\phi\big(\F_\sigma(\;{}^*\circ\theta\circ{}^*\;)[\mathcal D\psi]\big)=(4\pi\hbar)^{-1}\int
 K_\theta^*(\psc)\phi\big(\mathcal D(-\psc)\psi\big)   \,\mathrm d\psc.
 \end{align*}
On the other hand the symplectic nature of $\F_\sigma$ yields
$\F_\sigma(a^{\wedge p})=(\F_\sigma(a))^{\wedge x}$ for all $a\in\mathcal S(\mathbb R^2)$
(where $a^{\wedge x}(x_0,p_0):=a(-x_0,p_0)$). Similarly let $\Pi_{\theta,\iota}$ denote the
extension of $\Pi_\theta$ with respect to $\iota$; this lets us compute
\begin{align*}
& (\Pi_{\theta,\iota}\phi)(\psi)=(4\pi\hbar)^{-1}\int K_\theta(\psc)\phi\Big(\F_\sigma\big[  \big(\F_\sigma[\mathcal D\psi]
\big)^{\wedge p}\big](\psc)\Big)   \,\mathrm d\psc\\
&=(4\pi\hbar)^{-1}\int K_\theta(\psc)\phi\Big(\F_\sigma\big[  \big(\F_\sigma[\mathcal D\psi]
\big)\big]^{\wedge x}(\psc)\Big)   \,\mathrm d\psc\\
& =(4\pi\hbar)^{-1}\int K_\theta(x_0,p_0)\phi\big(\mathcal D(-x_0,p_0)\psi\big) \,\mathrm dx_0\mathrm dp_0\\
&=(4\pi\hbar)^{-1}\int K_\theta(x_0,-p_0)\phi\big(\mathcal D(-x_0,-p_0)\psi\big) \,\mathrm dx_0\mathrm dp_0.
\end{align*}
Thus $\Pi_{\theta,\iota}\equiv\Pi_{\theta,L^2}$ is equivalent to $K_\theta^*\equiv K_\theta^{\wedge p}$ as claimed.\qed
\end{proof}
We emphasize that all filter functions used in practice satisfy $K_\theta^*\equiv K_\theta^{\wedge p}$
(cf.~Tables \ref{quanttable} and \ref{filterfunctionstable}), so it does not matter for applications 
whether one extends $\Pi_\theta$ with respect to $\iota$ or $\langle\,\cdot\,|$.
\section{Proofs of Lemma \ref{lemma_bounded_op_sufficient} and Theorem \ref{convolutionproposition}  \label{proofs_4_1}}

Before we dive into the proofs of the results in question we first need a lemma which relates convolutions
of the cross-Wigner transform with matrix elements of the generalized Grossmann-Royer operator.
\begin{lemma}\label{lemma_wign_gros_roy}
Given any $\theta\in\mathcal S'(\mathbb R^2)$ one finds
\begin{equation}\label{eq:wigner_grossmann_royer_connection}
\langle\phi,\mathcal D(\psc)\Pi_\theta\mathcal D^\dagger(\psc)\psi\rangle=\pi\hbar[\theta\ast W_{\phi,\psi}](\psc)
\end{equation}
for all $\phi,\psi\in\mathcal S(\mathbb R)$ and all $\psc\in\mathbb R^2$. If $\Pi_\theta\in\mathcal B(L^2(\mathbb R))$
then Eq.~\eqref{eq:wigner_grossmann_royer_connection} even holds for all $\phi,\psi\in L^2(\mathbb R)$.
\end{lemma}
\begin{proof}
Sums in the argument of the displacement operator decompose as (see Eq.~\eqref{eq:D_2} and \cite[Eq.~(1.10)]{thewignertransform}):
$$
\mathcal D(\psc{-}\psc')=e^{-\frac{i}{2\hbar}( xp'-x'p )}\mathcal D(\psc)\mathcal D(-\psc').
$$
This lets us connect the r.h.s.~of \eqref{eq:wigner_grossmann_royer_connection} with the Grossmann-Royer operator \eqref{defgr}:
\begin{align*}
& \langle \phi,[\mathcal T(\psc)\mathcal D]^\dagger(\psc')\Pi [\mathcal T(\psc)\mathcal D](\psc')\psi\rangle
=\langle\phi,\mathcal D(\psc{-}\psc')\Pi\mathcal D^\dagger(\psc{-}\psc')\psi\rangle\\
&=\langle\phi,\mathcal D(\psc)\mathcal D(-\psc')\Pi\mathcal D^\dagger(-\psc')\mathcal D^\dagger(\psc)\psi\rangle
=\langle\mathcal D(-\psc)\phi ,\tfrac12 \F_\sigma[\mathcal D(\cdot)\mathcal D(-\psc)\psi](\psc')\rangle
\end{align*}
Together with linearity and continuity of $\theta$ this implies
\begin{align*}
& \pi\hbar[\theta\ast W_{\phi,\psi}](\psc)=
\theta\big(  \langle\phi, [\mathcal T(\psc)\mathcal D]^\dagger(\cdot)\Pi[\mathcal T(\psc)\mathcal D](\cdot)\psi\rangle\big)\\
&=\theta\big(\langle\mathcal D(-\psc)\phi, \tfrac12 \F_\sigma[\mathcal D(\cdot)\mathcal D(-\psc)\psi]\rangle\big)
= \langle\phi,\mathcal D(\psc)  \tfrac12\theta\big[ \F_\sigma\big(  (\mathcal D(\cdot)\mathcal D(-\psc)\psi)  \big)  \big]\rangle\\
&=\langle\phi,\mathcal D(\psc)\Pi_\theta\mathcal D(-\psc)\psi\rangle=\langle\phi,\mathcal D(\psc)\Pi_\theta\mathcal D^\dagger(\psc)\psi\rangle.
\end{align*}
In the second-to-last step we used the general definition of $\Pi_\theta$ from Remark \ref{rem_general_def_parity}.
Now if $\Pi_\theta$ is bounded then the l.h.s.~of \eqref{eq:wigner_grossmann_royer_connection}
extends to all square-integrable functions by density of $\mathcal S(\mathbb R)$ in $L^2(\mathbb R)$.\qed
\end{proof}

Thus we have the (formal) equality $\langle\phi,\Pi_\theta\psi\rangle=\pi\hbar[\theta\ast W_{\phi,\psi}](0,0)$
for all $\phi,\psi:\mathbb R\to\mathbb C$ where this expression is well-defined. Using this we are ready to
prove Lemma \ref{lemma_bounded_op_sufficient}, in particular the equivalence of (i,a), (i,b) and (ii,a), (ii,b), (ii,c)
for general $\theta\in\mathcal S'(\mathbb R^2)$.

\subsection{Proof of Lemma \ref{lemma_bounded_op_sufficient}}\label{proofs_4_1_a}

\begin{proof}[Lemma \ref{lemma_bounded_op_sufficient}]
``(i,a) $\Rightarrow$ (i,b)'': Because $\Pi_\theta$ is well defined, $\psi\mapsto[\theta\ast W_\psi](0,0)
=(\pi\hbar)^{-1}\langle\psi|\Pi_\theta \psi\rangle$ is well defined on $L^2(\mathbb R)$ as well.
``(i,b) $\Rightarrow$ (i,a)'': Assume that $\psi\mapsto[\theta\ast W_\psi](0,0)$ is well defined on 
$L^2(\mathbb R)$. Then $\langle\psi|\Pi_\theta \psi\rangle=\pi\hbar
[\theta\ast W_{\psi,\psi}](0,0)=\pi\hbar[\theta\ast W_\psi](0,0)$ exists for all
$\psi\in L^2(\mathbb R)$ and the same is true for $\langle\psi|\Pi_\theta \phi\rangle$
using the parallelogram
law 
$4\langle\psi|\Pi_\theta\phi\rangle= \langle\psi{+}\phi|\Pi_\theta (\psi{+}\phi)\rangle -  
\langle\psi{-}\phi|\Pi_\theta (\psi{-}\phi)\rangle+i\langle\psi{-}i\phi|\Pi_\theta (\psi{-}i\phi)\rangle -i
\langle\psi{+}i\phi|\Pi_\theta (\psi{+}i\phi)\rangle$. 
Hence $\tilde\phi\mapsto\int\tilde\phi(x)(\Pi_\theta\phi)(x)\,\mathrm{d}x$ is a
well-defined linear functional on $L^2(\mathbb R)$ meaning---because it is a
functional ``of integral pairing form''---it is automatically continuous as we prove now: 
If $\langle f,\cdot\,\rangle:L^2(\mathbb R)\to\mathbb C$,
$g\mapsto\int f(x)g(x)\,\mathrm{d}x$ is well defined for some
$f:\mathbb R\to\mathbb C$, then $(\Re(fg))_\pm$, $(\Im(fg))_\pm$ 
are integrable by definition of the Lebesgue integral,
where $f_{+}(x) := \max(f(x),0)$ and $f_{-} := - \min(f(x),0)$.
But these can be expanded into 
$(\Re(f))_\pm(\Re(g))_\pm$, $(\Re(f))_\pm(\Im(g))_\pm$, $(\Im(f))_\pm(\Re(g))_\pm$, $(\Im(f))_\pm(\Im(g))_\pm$ 
meaning the linear functionals $g\mapsto\int(\Re(f))_\pm g$ and $g\mapsto\int(\Im(f))_\pm g$ are also well defined on 
$L^2(\mathbb R)$. Now each of these is a positive functional on $L^2(\mathbb R)$ which is well known to be 
continuous (one can prove this similar to \cite[Ch.~2, Lemma 2.1]{Davies76}). Therefore $\langle f,\cdot\,\rangle$ is continuous as it is
the linear combination of four continuous functionals.

Then the Riesz representation theorem
(cf., e.g., \cite[Supplementary Material, Thm.~S.4]{ReedSimon1}) yields
$f\in L^2(\mathbb R)$ such that $f^*(x)=(\Pi_\theta\phi)(x)$ for almost all 
$x\in\mathbb R$; in particular $\Pi_\theta\phi\in L^2(\mathbb R)$. But
$\phi\in L^2(\mathbb R)$ was chosen arbitrarily meaning $\Pi_\theta$ is a
well-defined linear operator on $L^2(\mathbb R)$. The equivalence 
``(ii,a) $\Leftrightarrow$ (ii,b)'' is obvious and ``(ii,a) $\Leftrightarrow$ (ii,c)'' follows at
once from \eqref{eq:wigner_grossmann_royer_connection} together with
\begin{equation*}
\sup_{\|\psi\|,\|\phi\|=1}| [\theta \ast W_{\phi,\psi}](0,0)|=\sup_{\|\psi\|,\|\phi\|=1}| \langle\phi|\Pi_\theta\psi\rangle
|=\sup_{\|\phi\|=1}\|\Pi_\theta\phi\|_{L^2}=\|\Pi_\theta\|_{\indsupnorm}.
\end{equation*}

Now assume that $\theta$ is admissible. Because ``(ii,c) $\Rightarrow$ (i,a)'' is trivial, all that remains to show is
``(i,a) $\Rightarrow$ (ii,c)'':  Our idea is to show that $\theta$ being admissible implies that $\Pi_\theta$ can be written 
as the linear combination of two well-defined symmetric operators on $L^2(\mathbb R)$. This would conclude the 
proof because every symmetric operator is bounded by the Hellinger-Toeplitz theorem \cite[p.~84]{ReedSimon1},
hence $\Pi_\theta$ is bounded as well.
Set $K_{\theta^*}(\Omega):=K_\theta^*(-\Omega)$ and define $\Pi_{\theta^*}$ to be the parity operator
generated by $K_{\theta^*}$. First we have to see whether $\Pi_\theta$ being well defined on $L^2(\mathbb R)$ implies that the 
same holds for $\Pi_{\theta^*}$. Given $\psi,\phi\in L^2(\mathbb R)$ we formally compute
\begin{align*}
\langle \psi | \parity_{\theta^*} \phi \rangle&=(4\pi\hbar)^{-1}\int 
K_\theta^*(-\Omega)\langle\psi|\mathcal D(\Omega)\phi\rangle\,\mathrm{d}\Omega
= (4\pi\hbar)^{-1}\int K_\theta^*(\tilde\Omega)\langle\psi|
\mathcal D^\dagger(\tilde\Omega)\phi\rangle\,d\tilde\Omega\\
&= (4\pi\hbar)^{-1}\int K_\theta^*(\tilde\Omega)\big(\langle\phi|
\mathcal D(\tilde\Omega)\psi\rangle\big)^*\,d\tilde\Omega=(\langle\phi|\Pi_\theta\psi\rangle)^*.
\end{align*}
Thus $\langle \psi | \parity_{\theta^*} \phi \rangle$ exists for all $\psi,\phi\in L^2(\mathbb R)$ so by the same argument we used 
above, $\Pi_{\theta^*}$ is well defined on $L^2(\mathbb R)$. This yields the decomposition
$
\Pi_\theta=(\Pi_\theta{+}\Pi_{\theta^*})/2+i\cdot (\Pi_\theta{-}\Pi_{\theta^*})/(2i)
$
meaning all that is left to show is that $\Pi_\theta{+}\Pi_{\theta^*}$, $i(\Pi_\theta{-}\Pi_{\theta^*})$
are symmetric operators: indeed given $\psi,\phi\in L^2(\mathbb R)$ one computes
$$
\langle \psi|(\Pi_\theta{+}\Pi_{\theta^*})\phi\rangle = \langle\psi|\Pi_\theta\phi\rangle+\big(\langle\phi|\Pi_{\theta}\psi\rangle\big)^*
=\big( \langle\phi|\Pi_{\theta^*}\psi\rangle+\langle\phi|\Pi_{\theta}\psi\rangle\big)^*= \langle (\Pi_\theta{+}\Pi_{\theta^*})\psi|\phi\rangle
$$
and analogously for $i(\Pi_\theta{-}\Pi_{\theta^*})$.
As stated above $\Pi_\theta{+}\Pi_{\theta^*}$, $i(\Pi_\theta{-}\Pi_{\theta^*})\in\mathcal B(L^2(\mathbb R))$
by the Hellinger-Toeplitz theorem so $\Pi_\theta\in\mathcal B(L^2(\mathbb R))$ as well.\qed
\end{proof}

\subsection{Proof of Theorem \ref{convolutionproposition}}\label{proofs_4_1_b}
Moreover Lemma \ref{lemma_wign_gros_roy} enables a simple proof of Theorem \ref{convolutionproposition}:

\begin{proof}[Theorem \ref{convolutionproposition}]
Using the spectral decomposition 
$\rho= \sum_{n = 1}^\infty \, p_n | \psi_n \rangle \langle \psi_n |$ as well as Definition \ref{inifnitedimdefinition}, 
we compute for equation \eqref{convolutioneq} that
\begin{align*}
\FF_\rho(\psc,\theta) & =(\pi \hbar)^{-1} \, \tr\,[ \, \rho \, \mathcal{D}(\psc)  \parity_\theta  \mathcal{D}^\dagger(\psc) ]
=\sum_{n=1}^\infty p_n(\pi \hbar)^{-1}\langle\psi_n |	\mathcal{D}(\psc)  \parity_\theta  \mathcal{D}^\dagger(\psc)  |\psi_n \rangle\\
&=\sum_{n=1}^\infty p_n  [\theta\ast W_{\psi_n} ](\psc)=\Big[\theta\ast\sum_{n=1}^\infty p_n  W_{\psi_n} \Big](\psc)= [\theta \ast W_\rho](\psc).
\end{align*}	
Now for Equation~\eqref{generalizedfourier}: if $\theta$ is admissible,
i.e.~$\theta=(2\pi\hbar)^{-1}\langle K_\theta^*, \F_\sigma(\cdot)\rangle=(2\pi\hbar)^{-1}\langle (\F_\sigma K_\theta^\vee)^*|$, 
then Lemma \ref{lemma_wign_gros_roy} verifies the desired equality of quadratic forms as
\begin{align*}
& \mathcal D(\psc)\Pi_\theta\mathcal D^\dagger(\psc) 
=[\theta\ast\mathcal D(\cdot)\Pi\mathcal D^\dagger(\cdot)](\psc)\\
&\overset{\hphantom{\eqref{defgr}}}
 =(2\pi\hbar)^{-1}[\langle (\F_\sigma K_\theta^\vee)^*|\ast\mathcal D(\cdot)\Pi\mathcal D^\dagger(\cdot)](\psc)
= (2\pi\hbar)^{-1}[\F_\sigma K_\theta^\vee\ast\mathcal D(\cdot)\Pi\mathcal D^\dagger(\cdot)](\psc)\\
&\overset{\hphantom{\eqref{defgr}}}
=\F_\sigma[K_\theta^\vee\cdot \F_\sigma(\mathcal D(\cdot)\Pi\mathcal D^\dagger(\cdot))](\psc)\\
&
\overset{\eqref{defgr}}=\tfrac12 \F_\sigma[K_\theta^\vee\cdot \mathcal D^\dagger](\psc)
=\tfrac12 \F_\sigma[K_\theta\cdot \mathcal D]^\vee(\psc)=\tfrac12 \F_\sigma[K_\theta\cdot \mathcal D](-\psc). \tag*{\qed}
\end{align*}
\end{proof}
\section{Phase-Space Distributions for Arbitrary Convolution Kernels}
\label{app_phase_space_extension}
Given arbitrary $\theta\in\mathcal S'(\mathbb R^2)$
one can make sense of 
the phase-space distribution function by restricting the domain of $\rho\mapsto 
\FF_\rho(\psc,\theta)$ to quantum states $\rho$ which, e.g., have a finite 
representation in the number state basis. 
Indeed let $\rho\in\mathcal B^1(L^2(\mathbb R))$ be given such that
$\rho=\sum_{m,n=1}^N\langle m|\rho n\rangle|m\rangle\langle n|$ for some $N\in\mathbb N_0$. Then \eqref{genpsdef} becomes
\begin{equation}\label{eq:def_F_rho_finite}
	\FF_{\rho}(\psc,\theta)=\sum_{m,n=0}^N(\pi\hbar)^{-1}\langle m\rho|n\rangle\langle\mathcal D^\dagger(\psc) n|\Pi_\theta D^\dagger(\psc)m\rangle,
\end{equation}
which is a well-defined expression \textit{regardless} of the chosen $\theta\in\mathcal S'(\mathbb R^2)$, cf.~the paragraph
right before Lemma \ref{lemma_bounded_op_sufficient} together with the simple fact that $\mathcal D$ is an automorphism
on $\mathcal S(\mathbb R)$.
To see that \eqref{eq:def_F_rho_finite} generalizes Definition \ref{def_phase_space_function}
note that if $\Pi_\theta\in\mathcal B(L^2(\mathbb R))$, then
$
\lim_{N\to\infty}\FF_{\rho_N}(\psc,\theta)=\FF_\rho(\psc,\theta)
$
uniformly in $\psc\in\mathbb R^2$ for all states $\rho$, where
$\rho_N:=\sum_{m,n=1}^N\langle m|\rho n\rangle|m\rangle\langle n|$ is just the ``upper left $N\times N$ block'' of
$\rho$. One sees this using Prop.~16.6.6 from \cite{MeiseVogt} as
\begin{align*}
|\FF_{\rho_N}(\psc,\theta)-\FF_\rho(\psc,\theta)|\leq \|\rho-\rho_N\|_1\|\mathcal{D}(\psc)
\parity_\theta  \mathcal{D}^\dagger(\psc)\|_\indsupnorm\leq \|\rho{-}\rho_N\|_1\|\parity_\theta\|_\indsupnorm\to 0.
\end{align*}
Here we used that $\mathcal D(\psc)$ is unitary so it has operator norm one, together with the fact that $\|\rho{-}\rho_N\|_1\to 0$ as $N\to\infty$
which is a simple consequence of Prop.~2.1 in \cite{Widom76}. 
This motivates the general definition
\begin{align*}
\FF(\psc,\theta):D_{\FF}&\to(\mathbb R^2\to\mathbb C)\\
\rho&\mapsto \FF_\rho(\psc,\theta):=\sum_{m,n=0}^\infty(\pi\hbar)^{-1}\langle
m\rho|n\rangle\langle\mathcal D^\dagger(\psc) n|\Pi_\theta D^\dagger(\psc)m\rangle
\end{align*}
with domain
\begin{equation*}
D_{\FF}:=\Big\{\,\rho\in\mathcal B^1(L^2(\mathbb R))
	\,\text{ s.t. }\,
	\Big(\sum_{m,n=0}^N(\pi\hbar)^{-1}\langle
	m\rho|n\rangle\langle\mathcal D^\dagger(\psc) n|\Pi_\theta D^\dagger(\psc)m\rangle\Big)_{N\in\mathbb N}\text{ converges }
	\Big\}.
\end{equation*}
In particular Equation \eqref{eq:def_F_rho_finite} shows that for \textit{all} $\theta\in S'(\mathbb R^2)$
the domain $D_{\FF}$ is a dense in the full trace class.
However, unlike in the bounded case, it may happen that $\psc\mapsto \FF_\rho(\psc,\theta)$ is not a
function of slow growth so $\FF(\psc,\theta)$ may not map to the phase-space distributions.
\section{Proofs of the Properties from Section~\ref{property_phase_space} \label{proofsofproperties}}

\subsection{Proof of Property~\ref{bound}}
Recall that the Hilbert-Schmidt norm of an operator $A$  is
defined as $\|A\|^2_{\indhsnorm} := 
\tr( A^\dagger A )$.
One obtains
\begin{equation*}
\|\parity_{\theta}\|^2_{\indhsnorm} = 
\tr( \parity_{\theta}^\dagger \parity_{\theta} )
=
(4 \pi \hbar)^{-2} \iint
K_\theta^*(\psc)  K_\theta(\psc')  \,
\tr[ \mathcal{D}^\dagger(\psc) \mathcal{D}(\psc') ]
\, \mathrm{d}\psc \, \mathrm{d}\psc'
\end{equation*}
by substituting $\parity_{\theta}$ with its definition from \eqref{genparitydef}.
We formally replace the trace
$\tr[ \mathcal{D}^\dagger(\psc) \mathcal{D}(\psc') ]$  with $2 \pi \hbar \,   \delta(\psc {-}\psc')$
\cite{Cahill68} and  it follows that
\begin{equation*}
\|\parity_{\theta}\|^2_{\indhsnorm} = 
(8 \pi \hbar)^{-1}  \int K_\theta^*(\psc)  K_\theta(\psc)  \, \mathrm{d}\psc
= (8 \pi \hbar)^{-1} \| K_\theta(\psc)\|^2_{\indLnorm}.
\end{equation*}
The inequality 
$\|\parity_{\theta} \|_\indsupnorm \leq 	 \|\parity_{\theta}\|_{\indhsnorm} =  
\| K_\theta(\psc)\|_{\indLnorm} / \sqrt{ 8 \pi \hbar}$
\cite[Thm.~VI.22.(d)]{ReedSimon1}
concludes the proof.

\subsection{Proof of Property~\ref{l2property}}
Recall that the Wigner function is square integrable as
$\tr( \rho_1^\dagger \, \rho_2 ) = \int W_{\rho_1}^* \,  W_{\rho_2} \, \mathrm{d} \psc$ and   
$|\langle W_{\rho_1} | W_{\rho_2} \rangle | \leq 1$ hold. Similarly, one obtains
for elements $\FF_\rho(\psc,\theta) = \theta(\psc) \ast W_\rho(\psc)$ of the Cohen class the 
scalar products
\begin{align*}
\int \FF_{\rho_1}^*(\psc,\theta) \, \FF_{\rho_2}(\psc,\theta)  \, \mathrm{d} \psc
&=
\int [\theta(\psc) \ast W_{\rho_1}(\psc)]^* [  \theta(\psc) \ast W_{\rho_2}(\psc)] \, \mathrm{d} \psc\\
&=
\int \F_\sigma [\theta(\cdot) \ast W_{\rho_1}(\cdot)]^*(\psc) \, 
\F_\sigma [\theta(\cdot) \ast W_{\rho_2}(\cdot)](\psc) \, \mathrm{d} \psc
\end{align*}
using the Plancherel formula
$\int a(\psc) b^*(\psc)\, \mathrm{d} \psc = \int a_\sigma(\psc) b^*_\sigma(\psc)\, \mathrm{d} \psc $.
One can simplify the integrands to
\begin{equation*}
\F_\sigma [ \theta(\cdot) \ast W_{\rho_2}(\cdot) ](\psc) = 2\pi \hbar 
\F_\sigma [\theta(\cdot) ](\psc) \, \F_\sigma [W_{\rho_2}(\cdot) ](\psc)
\end{equation*}
by applying the convolution formula from \eqref{convolutiondef}.
Theorem \ref{convolutionproposition} implies
$K_\theta(\psc) = 2\pi \hbar [\F_{\sigma} \theta(\cdot)](-\psc)$ 
which yields the simplified integral
\begin{equation} \label{simplified_integral}
\int \FF_{\rho_1}^*(\psc,\theta)\,  \FF_{\rho_2}(\psc,\theta)  \, \mathrm{d} \psc
=
\int  |K_\theta(-\psc)|^2  \,  \F_\sigma [W_{\rho_1}(\cdot) ]^*(\psc)  \,
\F_\sigma [W_{\rho_2}(\cdot) ](\psc) \, \mathrm{d} \psc.
\end{equation}
By assumption, $K_\theta(\psc) \in L^\infty(\mathbb{R}^2)$, i.e, 
there exists a constant $C\in \mathbb{R}$ such that $|K_\theta(\psc)| \leq C$ holds 
almost everywhere.
Applying this bound to Eq.~\eqref{simplified_integral}
after setting $\rho_1 = \rho_2 =: \rho$ yields 
\begin{equation*}
\int  |K_\theta(-\psc)|^2  \, | \F_\sigma [W_{\rho}(\cdot) ](\psc) |^2 \, \mathrm{d} \psc
\leq
C^2 \int   \,  | \F_\sigma [W_{\rho}(\cdot) ](\psc) |^2  \, \mathrm{d} \psc
=
C^2 \int   \,  | W_{\rho} (\psc) \, |^2  \, \mathrm{d} \psc
\end{equation*}
with the help of the  Plancherel formula.
The above mentioned result for the Wigner function
implies the square integrability of $\FF_{\rho}(\psc,\theta)$,
which concludes the proof.

\subsection{Proof of Property~\ref{translationalcovariance}}
	As in \eqref{densityop}, one considers the density operators 
	$
	\rho= \sum_{n} \, p_n | \psi_n \rangle \langle \psi_n |
	$
	and
	$
	\rho'= \sum_{n} \, p_n | \phi_n \rangle \langle \phi_n |
	$.
	The orthonormality of $| \phi_n \rangle = \mathcal{D}(\psc')| \psi_n \rangle$
	is used to evaluate the trace and this yields
	\begin{align*} 
	\tr\,[ \, \rho' \, \mathcal{D}(\psc)  \parity_\theta  \mathcal{D}^\dagger(\psc) ]
	&=
	\sum_n p_n \langle \phi_n | \mathcal{D}(\psc)  \parity_\theta  \mathcal{D}^\dagger(\psc)\,  \phi_n \rangle \\
	&=
	\sum_n p_n \langle \psi_n | \mathcal{D}^\dagger(\psc') \mathcal{D}(\psc)  \parity_\theta 
	\mathcal{D}^\dagger(\psc) \mathcal{D}(\psc')\, \psi_n \rangle.
	\end{align*}
	Computing the addition of products $\mathcal{D}(\psc) \mathcal{D}(\psc')$
	of displacement operators \cite[Eq.~(1.10)]{thewignertransform} concludes the proof
	by using $\mathcal{D}^\dagger(\psc') = \mathcal{D}(-\psc')$ and
	$
	\tr\,[ \, \rho' \, \mathcal{D}(\psc)  \parity_\theta  \mathcal{D}^\dagger(\psc) ]
	=
	\tr\,[ \, \rho \, \mathcal{D}(\psc {-} \psc')  \parity_\theta  \mathcal{D}^\dagger(\psc {-} \psc') ]
	$.

\subsection{Proof of Property~\ref{rotationalcovariance}}

First, we prove that the displacement operator is covariant under rotations,
i.e.~$U^\dagger_\phi  \mathcal{D}(\psc) U_\phi  = \mathcal{D}( \psc^{-\phi})$. This is conveniently shown 
in the coherent-state representation as detailed in Eq.~\eqref{D_alpha}. Note that
\begin{equation*}
U^\dagger_\phi  \mathcal{D}(\psc) U_\phi | 0 \rangle
=
 e^{-|\alpha|^2/2}   \, \sum_{n=0}^{\infty} \tfrac{\alpha^n}{\sqrt{n!}} U^\dagger_\phi   | n \rangle 
= e^{-|\alpha|^2/2}   \, \sum_{n=0}^{\infty} \tfrac{[\exp{(i \phi  )}  \alpha ]^n}{\sqrt{n!}}    | n \rangle 
=  \mathcal{D}( \psc^{-\phi}) | 0 \rangle,
\end{equation*}
where the eigenvalue equation $U^\dagger_\phi   | n \rangle  = \exp{( i n \phi  )}  | n \rangle$
was used together with its special case $U_\phi | 0 \rangle =  | 0 \rangle$. It follows
from \eqref{genparitydef} that
\begin{equation*} 
U_\phi  \parity_\theta U^\dagger_\phi
=
 (4 \pi \hbar)^{-1}\int K_\theta(\psc) U_\phi  \mathcal{D}(\psc) U^\dagger_\phi \, \mathrm{d}\psc
= (4 \pi \hbar)^{-1}\int K_\theta(\psc) \mathcal{D}(\psc^{\phi}) \, \mathrm{d}\psc
= \parity_\theta
\end{equation*}
and the last equation is a consequence of 
the invariance $K_\theta(\psc) = K_\theta(\psc^{\phi}) $. Now, 
considering the density operators $\rho$ and $\rho^{\phi} = U_\phi \rho U^\dagger_\phi$,
the traces can be evaluated as
\begin{gather*} 
\tr [   \rho^{\phi}   \mathcal{D}(\psc)  \parity_\theta  \mathcal{D}^\dagger(\psc) ]
=
\tr [   \rho^{\phi}   \mathcal{D}(\psc)  U_\phi  \parity_\theta U^{\dagger}_\phi \mathcal{D}^\dagger(\psc) ]
= \tr [   U_\phi \rho U^\dagger_\phi   \mathcal{D}(\psc)  U_\phi  \parity_\theta U^{\dagger}_\phi \mathcal{D}^\dagger(\psc) ]\\
= \tr [   \rho U^\dagger_\phi   \mathcal{D}(\psc)  U_\phi  \parity_\theta\, (U^{\dagger}_\phi \mathcal{D}(\psc) U_\phi)^\dagger ]
=
\tr [   \rho   \mathcal{D}(\psc^{-\phi})  \parity_\theta  \mathcal{D}^\dagger(\psc^{-\phi}) ],
\end{gather*}
which verifies that the displacement operator is covariant under rotations.
The diagonality of $\parity_\theta$ in the Fock basis can be shown as follows:
If $K_\theta(\psc)$ is invariant under rotations it must be a function of
the polar radius in the phase space, i.e.~$K_\theta(\psc) = K_\theta(|\alpha|^2)$.
The matrix elements can be calculated via \eqref{displacementopmatirx} as
\begin{equation*}
\langle n | \parity_\theta\, m \rangle =
(4 \pi \hbar)^{-1} \int K_\theta(|\alpha|^2) \,[\mathcal{D}(\alpha)]_{nm}  \, \mathrm{d}\psc 
\propto  
\int \alpha^{m-n} f(|\alpha|^2)
 \, \mathrm{d}\psc \propto \delta_{nm}
\end{equation*}
with $f(|\alpha|^2) = K_\theta(|\alpha|^2) e^{-|\alpha|^2/2} L_n^{(m-n)}(|\alpha|^2)$,
so the integral vanishes unless $n=m$.
	
\subsection{Proof of Property~\ref{reality}}

The expectation value
$
\langle \psi_n | \mathcal{D}(\psc)  \parity_\theta  \mathcal{D}^\dagger(\psc)\, \psi_n \rangle
=
\langle \phi_n |   \parity_\theta\, \phi_n \rangle
$
 is real if $\parity_\theta$  is self-adjoint, where
 the orthonormal bases
  $ \lbrace|  \psi_n \rangle\rbrace_{n=0,1,\ldots}$ and $ \lbrace| \phi_n \rangle\rbrace_{n=0,1,\ldots}$
of the considered Hilbert space have been applied. Assuming $K^*_\theta(-\psc)=K_\theta(\psc)$, this translates to 
\begin{equation*}
\parity^\dagger_\theta 
= \tfrac{1}{4 \pi \hbar}\int K^*_\theta(\psc) \mathcal{D}^\dagger(\psc) \, \mathrm{d}\psc
= \tfrac{1}{4 \pi \hbar}\int K^*_\theta(\psc) \mathcal{D}(-\psc) \, \mathrm{d}\psc 
= \tfrac{1}{4 \pi \hbar}\int K_\theta(\psc) \mathcal{D}(\psc) \, \mathrm{d}\psc 
= \parity_\theta.
\end{equation*}

\subsection{Proof of Property~\ref{cohencl}}
The phase-space integral
\begin{align*}
\int \FF_\rho(\psc,\theta) \,\mathrm{d} \psc &= 
(\pi\hbar)^{-1} \tr\,[ \, \rho \,
\int 
\mathcal{D}(\psc)  \parity_\theta  \mathcal{D}^\dagger(\psc) 
\,\mathrm{d} \psc 
]\\
&=
2 \, \tr\{ \, \rho \,
\F_\sigma[
\mathcal{D}(\cdot)  \parity_\theta  \mathcal{D}^\dagger(\cdot) 
](\psc)|_{\psc=0}
\}=
\tr\,[ \, \rho \, \mathcal{D}(0) K_\theta(0)  ]
=K_\theta(0) \tr(\rho) 
\end{align*}
is mapped to the trace of $\rho$ 
if $K_\theta(0) = 1$. The second equality applies the symplectic Fourier transform
of Eq.~\eqref{generalizedfourier} at the point $\psc = 0$.
Formally the trace of $\parity_{\theta}$ is given by
\begin{equation*}
\tr [ \parity_{\theta} ] =  (4 \pi \hbar)^{-1}\int K_\theta(\psc) \tr [ \mathcal{D}(\psc)  ] \, \mathrm{d}\psc
= (2)^{-1}\int K_\theta(\psc) \delta(\psc) \, \mathrm{d}\psc,
\end{equation*}
where we used
$\tr[ \mathcal{D}^\dagger(\psc) \mathcal{D}(\psc') ] = 2 \pi \hbar \,   \delta(\psc {-}\psc')$ \cite{Cahill68}.
Alternatively, this also follows from \eqref{quantizationofdelta} by formally computing the trace
\begin{equation*}
\tr[ \parity_{\theta} ]
=
(\pi \hbar)^{-1} \tr[  \textup{Op}_{\textup{Weyl}}( \theta )  ]
=
(\pi \hbar)^{-2} \int \theta(\psc) \, \tr[  \mathcal{D}(\psc)  \parity  \mathcal{D}^\dagger(\psc) ] \, \mathrm{d} \psc,
\end{equation*}
where the trace of the Grossmann-Royer operator from \eqref{defgr}
evaluates to 
$
\tr[  \mathcal{D}(\psc)  \parity  \mathcal{D}^\dagger(\psc) ] = \tr[  \parity  ] = 1/2
$, refer to (6.38) and the following text in \cite{Cahill68}.
Substituting its definition from \eqref{genparitydef}, the trace of $\parity_{\theta}$ is computed as
$
\tr[ \parity_{\theta} ]=
(2 \pi^2 \hbar^2)^{-1}
 \int  \theta(\psc) \,\mathrm{d}\psc
=(\pi \hbar)^{-1} \F_\sigma[\theta(\cdot)]|_{\psc = 0}
= K(0)/2
$.

\section{Proof of \eqref{sparametrizedparityform} \label{proofofsparametrized}}
Due to Property~\ref{rotationalcovariance}, the parity operator is
diagonal in the number-state representation
$\langle m | \parity_s  | n \rangle \propto \delta_{nm}$.
Its diagonal elements can be calculated
\begin{equation*} 
[ \parity_s ]_{nn}=
(4 \pi\hbar)^{-1} \int e^{s |\alpha|^2/2} \, [\mathcal{D}(\psc)]_{nn} \, \mathrm{d}\psc = 
(4 \pi\hbar)^{-1} \int e^{s |\alpha|^2/2} \,  e^{-|\alpha|^2/2} L_n(|\alpha|^2) \, \mathrm{d}\psc
\end{equation*}
where \eqref{displacementopmatirx} was used for $ [\mathcal{D}(\psc)]_{nn}$.
One applies the polar parametrization of the complex plane
via $\psc=\alpha=r \exp{(i \phi)}$ so that
$\mathrm{d}\psc = 2 \hbar \, \mathrm{d}\Re(\alpha) \,\mathrm{d}\Im(\alpha) = 2 \hbar \, r \, \mathrm{d}r \, \mathrm{d}\phi$. 
Then
\begin{align*} 
[ \parity_s ]_{nn} &=
(2 \pi)^{-1} \int_{0}^{2 \pi} \,\mathrm{d}\phi \int_{0}^{\infty} 
e^{s r^2/2} \,  e^{-r^2/2} L_n(r^2) \, r \, \mathrm{d}r \\
&= \tfrac{1}{2} \int_{0}^{\infty} e^{ y (s-1)/2 }  L_n(y) \, \mathrm{d}y 
= \tfrac{1}{2} \int_{0}^{\infty} e^{ -y } e^{ y (s+1)/2 }  L_n(y) \, \mathrm{d}y ,
\end{align*}
where the second equality is due to $ r \, \mathrm{d}r =  \mathrm{d}y/2$ with $y=r^2$
and the integral with respect to $\phi$ results in the multiplication by $2\pi$.
The Laguerre polynomial decomposition of the
exponential function 
$e^{-\gamma x}= \sum_{m=0}^\infty [\gamma^m/(1{+}\gamma)^{m+1}] L_m(x)$
with $\gamma = -(s{+}1)/2$ \cite[p.~90]{lebedev1972special}
and the orthogonality relation
$\int_{0}^{\infty} e^{-x} L_n(x) L_m(x) \, \mathrm{d}x = \delta_{nm}$
finally yields
$
[ \parity_s ]_{nn}= \gamma^n/[2 (1{+}\gamma)^{n+1}]
=  (-1)^n {(s{+}1)^n}/{(1{-}s)^{n+1}}
$,
which concludes the proof.

\section{Spectral Decomposition of the Squeezing Operator \label{appendixspectral}}

The eigenvectors from \eqref{eigenvectors} are orthogonal and normalized in terms of the delta function $\delta$
as detailed by
\begin{equation*}
\langle \psi^{E_1}_\pm | \psi^{E_2}_\pm \rangle = \int  [\psi^{E_1}_{\pm}(x)]^*  \,  \psi^{E_2}_{\pm}(x)  \, \mathrm{d} x
=
\delta(E_1{-}E_2).
\end{equation*}
The integral can be calculated using a change of variables $\mathrm{d} x = e^{v} \,\mathrm{d}v $
with $v=\ln(|x|)$. One obtains a complete basis 
\begin{equation*}
\int  [\psi^{E}_{\pm}(x)]^* \,  \psi^{E}_{\pm}(x')  \, \mathrm{d} E
=
\delta(x{-}x'),
\end{equation*}
by applying an integral of two different Fourier components indexed by $x$ and $x'$,
refer to \cite{chruscinski2003,chruscinski2004} for more details.
The eigenfunctions $\psi_\pm^E (x)$ are not square integrable, but they
can be decomposed into the number-state basis with finite coefficients.
The coefficients shrink to zero, but are not square summable.
The resulting integrals $\langle n | \psi_\pm^E  \rangle$ can be
specified in terms of a finite sum. In particular,
$\psi_+^0 = |x|^{- 1/2} /({2\sqrt{\pi}}) $ has the largest eigenvalue.
Its number-state representation is given by
\begin{equation*}
\langle n | \psi_+^0  \rangle=
\tfrac{1}{2\sqrt{\pi}}
\sum_{k=0}^{n/2}
\frac{\pi  2^{-3 k+n+\frac{1}{4}} \sqrt{n!}}{ k! (n{-}2 k)! \, \Gamma (k-\frac{n}{2}+\frac{3}{4})}
\; \text{ if } \;  n \bmod{4} =0,
\end{equation*}
where every fourth entry is non-zero and the entries decrease to zero for large $n$.

\section{Matrix Representation of the Born-Jordan Parity Operator \label{BJcomputation}}

The matrix elements of the parity operator can be computed via Theorem \ref{BJdistributionTheorem} as
\begin{equation*}
[\parity_{\textup{BJ}}]_{mn} 
=
(4 \pi\hbar)^{-1} \int K_{\textup{BJ}}(\psc) \, [\mathcal{D}(\psc)]_{mn} \, \mathrm{d}\psc
=
(4 \pi\hbar)^{-1} \int  \mathrm{sinc}( \tfrac{px}{2\hbar} 	)  \, [\mathcal{D}(\alpha)]_{mn} \, \mathrm{d}x  \, \mathrm{d}p.
\end{equation*}
It was discussed in Section~\ref{translations}
	 that one can substitute $\alpha=(\lambda x + i \lambda^{-1} p )/\sqrt{2 \hbar}$, 
	 which results in the integral
	 \begin{equation*}
	 [\parity_{\textup{BJ}}]_{mn} 
	 =
	 (4 \pi\hbar)^{-1} \int  \mathrm{sinc}( \tfrac{px}{2\hbar} 	)  \,
		  \{\mathcal{D}[(\lambda x + i \lambda^{-1} p )/\sqrt{2 \hbar}]\}_{mn} \, \mathrm{d}x  \, \mathrm{d}p.
	 \end{equation*}
Let us now apply a change of variables
$x\mapsto\lambda^{-1}\sqrt{\hbar}x$ and $p\mapsto\lambda\sqrt{\hbar}p$
which yields $\mathrm{d}x  \, \mathrm{d}p \mapsto \hbar \mathrm{d}x  \, \mathrm{d}p$
and the integral
 \begin{equation*}
[\parity_{\textup{BJ}}]_{mn} 
=
(4 \pi)^{-1} \int  \mathrm{sinc}( \tfrac{px}{2} 	)  \,
\{\mathcal{D}[(x {+} i  p )/\sqrt{2}]\}_{mn} \, \mathrm{d}x  \, \mathrm{d}p.
\end{equation*}
We now substitute the explicit form of 
$[\mathcal{D}(\alpha)]_{mn}$ with $\{\mathcal{D}[( x {+} i  p )/\sqrt{2}]\}_{mn}$
from \eqref{displacementopmatirx} and obtain
\begin{equation}
[ \parity_{\textup{BJ}}]_{mn}
=
(4\pi)^{-1}\sum_{k=0}^n c^k_{mn} \int \mathrm{sinc}( \tfrac{px}{2}) \,   
[\tfrac{x + ip}{\sqrt{2}}  ]^{m-n}    e^{-(x^2 {+} p^2 )/4} (x^2 {+} p^2 )^k
\, \mathrm{d}x \, \mathrm{d}p, \label{eq_int}
\end{equation}
where the Laguerre polynomials are expanded using the new coefficients
\begin{equation*}
c^k_{mn}:=\sqrt{\tfrac{n!}{m!}} (-1)^k 2^{-k}  \binom{m}{n{-}k} / k!.
\end{equation*}
One applies the expansion
\begin{equation*}
 [\tfrac{x + ip}{\sqrt{2}}  ]^{m-n}  =2^{-(m-n)/2} \sum_{\ell=0}^{m-n}
\binom{m{-}n}{\ell}
x^{m-n-\ell}  \, (i p)^{\ell}.
\end{equation*}
and the integral in \eqref{eq_int} vanishes for odd powers of $x$ and $p$ due to symmetry of the integrand. 
Therefore, all non-vanishing matrix elements have even $m-n$ values and
the summations can be restricted to
$\ell \in \{0,2,4, \dots, m{-}n\}$.
The integral is also invariant under a permutation of  $x$ and $p$ 
and certain terms in the sum cancel each other out:
every term $x^{m-n-\ell}  \, (ip)^{\ell}$ in the sum 
has a counterpart $(ip)^{m-n-\ell}  \, x^{\ell}$ which results in the
same integral and these two terms therefore cancel each other out 
after the integration 
if the condition $(i)^{\ell} = - (i)^{m-n-\ell}$ holds
(which occurs unless $m-n$ is a multiple of four).
This elementary argument shows that only matrix elements $[ \parity_{\textup{BJ}} ]_{mn}$
with $m-n$ multiples of four are non-zero.
Recall that we have been using an indexing scheme with $m \geq n$ on account of the Laguerre polynomials in
\eqref{displacementopmatirx}, but matrix elements with $m < n$ are trivially obtained
as $[ \parity_{\textup{BJ}} ]_{mn}=[ \parity_{\textup{BJ}} ]_{nm}$.
Introducing the coefficient [with $i^\ell=(-1)^{\ell/2}$]
\begin{equation*}
w^k_{mn\ell}:=(-1)^{\ell/2} 2^{-(m-n)/2} \binom{m{-}n}{\ell} c^k_{mn}
\end{equation*}
and denoting $a=(m-n-\ell)/2$ and $b=\ell/2$, one obtains
\begin{equation}
[ \parity_{\textup{BJ}} ]_{mn} =
(4\pi)^{-1} \sum_{k=0}^n \sum_{\substack{\ell=0 \\ \textrm{$\ell$ even}}}^{m-n}
w^k_{mn\ell} \int  \mathrm{sinc}( \tfrac{px}{2})  \, x^{2a}  \, p^{2b} \,  (x^2 {+} p^2 )^k \,  e^{-(x^2 + p^2 )/4}
\, \mathrm{d}x \, \mathrm{d}p. \label{specialForm}
\end{equation}
The integral in \eqref{specialForm} is simplified using 
new variables $\lambda,\mu\in[1-\varepsilon,1+\varepsilon]$ for some $\varepsilon\in(0,1/2)$ as
\begin{equation*}
(\partial^k_{\mu} [\partial^{a}_{\lambda}\partial^{b}_{\mu}
e^{-(\lambda x^2 + \mu p^2 )/4}]|_{\lambda=\mu})]|_{\mu=1}=(-1)^{a+b+k} 4^{-(a+b+k)} x^{2a}  \, p^{2b} \,  (x^2 {+} p^2 )^k \,  e^{-(x^2 + p^2 )/4}
\end{equation*}
for all $a,b,k\in\mathbb N_0:=\{0,1,\ldots\}$ and $x,p\in{\mathbb{R}}$. Considering the mapping
\begin{equation*}
g:\mathbb R^2\times [1-\varepsilon,1+\varepsilon]^2\to\mathbb R,\,
(x,p,\lambda,\mu)\mapsto \mathrm{sinc}( \tfrac{px}{2}) \, e^{-(\lambda x^2 + \mu p^2 )/4},
\end{equation*}
the corresponding partial derivatives can be bounded by 
\begin{equation*}
 |  \mathrm{sinc}( \tfrac{px}{2}) \, x^{2a}  \, p^{2b} \,  (x^2 {+} p^2 )^k \,  e^{-(\lambda x^2 + \mu p^2 )/4} |\leq
	x^{2a}  \, p^{2b} \,  (x^2 {+} p^2 )^k \,  e^{-(x^2 + p^2 )/8}
\end{equation*}
where the upper bound is independent of $\lambda,\mu$ and integrable as $e^{-( x^2 + p^2 )/8}\in\mathcal S(\mathbb R^2)$.
We now may interchange the partial derivatives
by a version of Lebesgue's dominated convergence theorem 
\cite[Thm.~2.27.b]{FollandAnalysis}. 
The integral in \eqref{specialForm} is then given by
\begin{align*}
& (4 \pi)^{-1} \int \mathrm{sinc}( \tfrac{px}{2})  \, x^{2a}  \, p^{2b} \,  (x^2 {+} p^2 )^k \,  e^{-(x^2 {+} p^2 )/4}
\, \mathrm{d}x \, \mathrm{d}p\\
& =
(-1)^{a+b+k} 4^{a+b+k}[  (\partial^k_{\mu} [\partial^{a}_{\lambda}\partial^{b}_{\mu} f(\lambda,\mu)]|_{\lambda=\mu})]|_{\mu=1},
\end{align*}
where
\begin{equation*}
f(\lambda,\mu) = (4 \pi)^{-1} \int \mathrm{sinc}( \tfrac{px}{2}) \, e^{-(\lambda x^2 + \mu p^2 )/4} \, \mathrm{d}x \, \mathrm{d}p = 
\mathrm{arcsinh}[(\lambda \mu)^{-1/2}].
\end{equation*}
Note that $\lambda$ now denotes the variable of the function $f(\lambda,\mu)$ and should not
be confused with the scaling parameter $\alpha=(\lambda x + i \lambda^{-1} p )/\sqrt{2 \hbar}$ 
from Section~\ref{translations}, which has also been used in the beginning of this section.
This finally results in
\begin{equation*}
[ \parity_{\textup{BJ}} ]_{mn}
=
\sum_{k=0}^n \sum_{\substack{\ell=0 \\ \textrm{$\ell$ even} }}^{m-n} w^k_{mn\ell}
(-1)^{k+(m-n)/2} 4^{k+(m-n)/2} \nonumber  
(\partial^k_{\mu} \{[\partial^{a}_{\lambda}\partial^{b}_{\mu} 
f(\lambda,\mu)]|_{\lambda=\mu}\})|_{\mu=1}.
\end{equation*}

\section{Calculating Derivatives for the Sum in Theorem \ref{BJParityMatrix}
 \label{derivativescalculation}}

The derivatives $\mathrm{\Phi}_{ab}^k =
	[  \partial^k_{\mu} [\partial^{a}_{\lambda}\partial^{b}_{\mu} f(\lambda,\mu)]|_{\lambda=\mu} ]|_{\mu=1}$ of the 
function (cf.\ \eqref{derivatives})
$$
f: (0,\infty)\times (0,\infty)\to\mathbb R,
(\lambda,\mu) \mapsto  \mathrm{arcsinh}[1/\sqrt{\lambda \mu}]
$$
can be computed recursively. Note that, obviously, $f$ is smooth.
The inner derivatives of $\mathrm{\Phi}_{ab}^k$ gives rise to the following lemma.
\begin{lemma}\label{app_e_lemma_1}
Let any $a,b\in\mathbb N_0:=\{0,1,\ldots\}$ with $a+b\geq 1$ (else we are not taking any derivative). Then
\begin{equation}\label{eq:app_e_1}
\partial_\lambda^a\partial_\mu^b f(\lambda,\mu)=
\frac{\sum_{j=0}^{a+b-1}c_j^{a{}b}\lambda^{j-b}\mu^{j-a}}{(-2)^{a+b}(\sqrt{\lambda\mu{+}1})^{2(a+b)-1}}
\end{equation}
where the coefficients $c_j^{a{}b}$ are defined recursively by
\begin{equation*}
c_j^{a{}b}=
\begin{cases}
c_j^{a{}b} = 0 \; \text{ if } \;  j<0 \; \text{ or } \;   j \geq a+b \\
c_0^{1\,0} = 1 \\
c_j^{a+1,b} =c_{j-1}^{a{}b} (4a + 2b + 1 -2j) - 2 \, c_j^{a{}b} (j{-}a)
\end{cases}
\end{equation*}
and have the symmetry $c_j^{a{}b} = c_j^{b{}a}$.
\end{lemma}
\begin{proof}
Note that the symmetry of the $c_j^{a{}b}$ holds due to Schwarz's theorem \cite[pp.~235-236]{rudin1976principles} as $f$ is smooth. 
Then this statement is readily verified via induction over $n=a+b$. First, $n=1$ corresponds to $a=1,b=0$ so
$
\partial_\mu \mathrm{arcsinh}[(\lambda \mu)^{-1/2}]=1/(-2\mu\sqrt{\lambda\mu{+}1})
$
which reproduces \eqref{eq:app_e_1}. For $n\mapsto n+1$ it is enough to 
consider $(a,b)\mapsto (a{+}1,b)$ due to the stated symmetry. The key result here is that
\begin{align*}
\partial_\mu \frac{\mu^{j-a}}{(\sqrt{\lambda\mu{+}1})^{2(a+b)-1}}
=\frac{\mu^{j-a-1}}{2(\sqrt{\lambda\mu{+}1})^{2(a+b)+1}}[\lambda\mu (2j-4a-2b+1)+2(j{-}a)]
\end{align*}
which is readily verified.  Straightforward calculations conclude the proof.\qed
\end{proof}

For $a+b\geq 1$, the above result immediately yields
$$
\mathrm{\Phi}_{ab}^k =
	 \partial^k_{\mu} [\,\frac{\sum_{j=0}^{a+b-1}c_j^{a{}b}\mu^{2j-b-a}}{(-2)^{a+b}(\sqrt{\mu^2{+}1})^{2(a+b)-1}}\,]|_{\mu=1}.
$$
Now the $c_j^{a{}b}$ are used to initialize the
recursion of the coefficients $\xi_j^{a{}b{}k}$ for $a+b\geq 1$,
the sum of which determines the resulting derivatives as we will see now.
\begin{lemma}\label{app_e_lemma_2}
Let any $a,b,k\in\mathbb N_0$ with $a+b+k\geq 1$. Then
\begin{equation}\label{eq:app_e_2}
\partial^k_{\mu} [\partial^{a}_{\lambda}\partial^{b}_{\mu} 
f(\lambda,\mu)]|_{\lambda=\mu} ]=\frac{\sum_{j=0}^{a+b+k-1}
\xi_j^{a{}b{}k}\mu^{2j}}{ (-2)^{a+b}\mu^{a+b+k}(\sqrt{\mu^2{+}1})^{2(a+b+k)-1} }
\end{equation}
where the coefficients $\xi_j^{a{}b{}k}$ 
have the symmetry $\xi_j^{a{}b{}k}=\xi_j^{b{}a{}k}$ and
are defined by
\begin{equation}\label{eq:88b}
\xi_j^{a{}b{}k}=
\begin{cases}
\xi_j^{a{}b{}k} = 0 \; \text{ if } \;  j<0 \; \text{ or } \;   j \geq a+b+k \\
\xi_0^{0\,0\,1} = -1\\
\xi_j^{a{}b\,0} = c_j^{a{}b}\;\text{ if }\; a+b\geq 1 \\
\xi_j^{a{}b, k+1} =\xi_{j-1}^{a{}b{}k} (2j - 1 - 3a - 3b - 3k) + \xi_j^{a{}b{}k} (2j - a - b - k).
\end{cases}
\end{equation}
\end{lemma}
\begin{proof}
The key result here is
\begin{equation}\label{eq:app_e_3}
\partial_\mu \frac{\mu^{2j-\beta}}{(\sqrt{1{+}\mu^2})^{2\beta-1}}
=\frac{\mu^{2j-\beta-1}}{(\sqrt{1{+}\mu^2})^{2\beta+1}}[\mu^2(2j-3\beta+1)+(2j{-}\beta)]
\end{equation}
for any $\beta,j\in\mathbb N$ which can be easily seen. We have to distinguish 
the cases $a+b=0$ and $a+b\geq 1$. First, let $a+b=0$ so $a=0,b=0$ and the expression in question boils down to
$$
\partial_\mu^k f(\mu,\mu)=\partial_\mu^k \mathrm{arcsinh}[\mu^{-1}]= 
\frac{\sum_{j=0}^{k-1}\xi_j^{0{}0{}k}\mu^{2j-k}}{(\sqrt{1{+}\mu^2})^{2k-1}} 
$$
as can be shown via induction over $k\in\mathbb N$. Here, setting $\beta=k$ in \eqref{eq:app_e_3} yields
$$
\xi_j^{0{}0,k+1}=\xi_{j-1}^{0{}0{}k} (2j - 1 - 3k) + \xi_j^{0{}0{}k} (2j {-} k)
$$
which recovers the recursion formula of $\xi_j^{a{}b{}k}$ for $a=0$ and $b=0$. 
Now assume $a+b\geq 1$ such that we can carry out the proof via induction over 
$k\in\mathbb N_0$ (where $k=0$ is obvious as it is simply Lemma \ref{app_e_lemma_1}). 
Using \eqref{eq:app_e_3} in the inductive step for $\beta=a+b+k$ recovers the 
recursion formula of the $\xi_j^{a{}b{}k}$ by straightforward computations.\qed
\end{proof}

Finally, evaluating \eqref{eq:app_e_2} at $\mu=1$ for any 
$a,b,k\in\mathbb N_0$ with $a+b+k\geq 1$ readily implies 
Eq.~\eqref{xi_expansion}.

\section{Proof of Proposition \ref{prop_diag_ele}}\label{proof_prop_diag_ele}
The proof which is given below was informed by a discussion 
on MathOverflow\cite{MO295019} and its idea was provided GH and M.~Alekseyev.
We consider the generating function of the entries $[ \parity_{\textup{BJ}}]_{nn}$.
\begin{lemma}\label{Mnn_generation_function}
For all $|t|<1$ one has 
\begin{align*}
\sum_{n=0}^\infty [ \parity_{\textup{BJ}}]_{nn}t^n=\frac{1}{1-t}\operatorname{arcsinh}\Big(\frac{1{-}t}{1{+}t}\Big)
\end{align*}
where
\begin{align*}
[ \parity_{\textup{BJ}}]_{nn}=\sum_{k=0}^n\binom{n}{k} \frac{2^kc_k}{k!}
\;\text{ and }\; c_n=\frac{d^n}{\mathrm{d}x^n}\operatorname{arcsinh}\Big(\frac{1}{x}\Big)\Big|_{x=1}
\end{align*} 
for all $n\in\mathbb N_0$.
\end{lemma}
\begin{proof}
Obviously,
$
\operatorname{arcsinh}({1}/{w})=\sum_{n=0}^\infty({c_k}/{k!})(w{-}1)^k
$
for all $|w{-}1|<1$, so changing $w$ to $1+2w$ yields
\begin{align}\label{eq_arcsinh_1}
\operatorname{arcsinh}\Big(\frac{1}{1{+}2w}\Big)=\sum_{n=0}^\infty\frac{2^kc_k}{k!}w^k
\end{align}
for all $|w|<1/2$. By the generalized Leibniz rule,
\begin{gather}
[w^n](1{+}w)^n \operatorname{arcsinh}\Big(\frac{1}{1{+}2w}\Big)
= \frac{1}{n!}\frac{d^n}{dw^n}(1{+}w)^n\operatorname{arcsinh}\Big(\frac{1}{1{+}2w}\Big)\Big|_{w=0} \notag\\
= \frac{1}{n!}\sum_{k=0}^n
\binom{n}{k}\underbrace{\frac{d^k}{dw^k}
\operatorname{arcsinh}\Big(\frac{1}{1{+}2w}\Big)\Big|_{w=0}}_{=2^kc_k\text{ by }
\eqref{eq_arcsinh_1}}\underbrace{\frac{d^{n-k}}{dw^{n-k}}(1{+}w)^n\Big|_{w=0}}_{=n!/k!}
=\sum_{k=0}^n\binom{n}{k} \frac{2^kc_k}{k!}= [ \parity_{\textup{BJ}}]_{nn} \label{eq_arcsinh_2}
\end{gather}
for all $n\in\mathbb N_0$.
Here, $[t^n]g(t)={g^{(n)}(0)}/{n!}$ denotes the $n$th coefficient in the Taylor series of $g(t)$ around $0$. 
Now we apply the Lagrange–Bürmann formula 
\cite[3.6.6]{abramowitz1965handbook} to $\phi(w)=1+w$ [so $w/\phi(w)=t$ for $|t|<1$ has the unique solution 
$w={t}/(1{-}t)$] and $H(w)=(1{+}w)\operatorname{arcsinh}[{1}/(1{+}2w)]$ which concludes the proof via
\begin{align*}
[t^n]\frac{1}{1{-}t}\operatorname{arcsinh}\Big(\frac{1{-}t}{1{+}t}\Big)=
[t^n]H\Big(\frac{t}{1{-}t}\Big)&=[w^n]H(w)\phi(w)^{n-1}[\phi(w)-w\phi'(w)]\\
&=[w^n](1{+}w)^n \operatorname{arcsinh}\Big(\frac{1}{1{+}2w}\Big)\overset{\eqref{eq_arcsinh_2}}=[ \parity_{\textup{BJ}}]_{nn}.
\tag*{\qed}
\end{align*}
\end{proof}
\begin{lemma}\label{inf_sum_conv}
The following sum converges:
\begin{align}\label{eq:inf_sum_conv_1}
\sum_{k=0}^\infty\frac{(-1)^k}{k{+}1}\sum_{m=0}^{\lfloor \frac{k}{2}\rfloor}\binom{2m}{m}\Big(\frac{-1}{4}\Big)^m
\end{align}
\end{lemma}
\begin{proof}
For arbitrary $k\in\{0,1,2,\ldots\}$ we define
\begin{align*}
b_k:=(-1)^k\sum_{m=0}^{\lfloor \frac{k}{2}\rfloor}
\binom{2m}{m}\Big(\frac{-1}{4}\Big)^m.
\end{align*}
Due to the summation limit $\lfloor \frac{k}{2}\rfloor$, one has $b_{2k}=-b_{2k+1}$ for all $k\in\{0,1,2,\ldots\}$ and thus
$$
\Big(\sum_{k=0}^n b_k\Big)_{n=0,1,2,\ldots}=(b_0,0,b_2,0,b_4,0,\ldots).
$$
Therefore $(\sum_{k=0}^n b_k)_{n=0,1,2,\ldots}$ consists of the null sequence and $(b_{2n})_{n=0,1,2,\ldots}$, so it is bounded due to
$$
\lim_{n\to\infty} b_{2n}=\sum_{m=0}^\infty \binom{2m}{m}(-1/4)^m={1}/{\sqrt{2}}.
$$
In total, \eqref{eq:inf_sum_conv_1} then converges due to Dirichlet's test \cite[p.~328]{hardy2015course}.\qed
\end{proof}

With these intermediate results we can finally prove the proposition in question.

\begin{proof}
Again using the generalized Leibniz rule, Lemma \ref{Mnn_generation_function}  yields that
\begin{gather*}
[ \parity_{\textup{BJ}}]_{nn}=[t^n]\frac{1}{1{-}t}\operatorname{arcsinh}\Big(\frac{1{-}t}{1{+}t}\Big)
=\frac{1}{n!}\frac{d^n}{dt^n}\frac{1}{1{-}t}\operatorname{arcsinh}\Big(\frac{1{-}t}{1{+}t}\Big)\Big|_{t=0}\\
=\frac{1}{n!}\sum_{k=0}^n\binom{n}{k}\frac{d^k}{dt^k}
\operatorname{arcsinh}\Big(\frac{1{-}t}{1{+}t}\Big)\Big|_{t=0}\underbrace{\frac{d^{n-k}}{dt^{n-k}}\frac{1}{1{-}t}\Big|_{t=0}}_{=(n-k)!}
= \operatorname{arcsinh}(1)+\sum_{k=1}^n \frac{1}{k!}\frac{d^k}{dt^k}\operatorname{arcsinh}\Big(\frac{1{-}t}{1{+}t}\Big)\Big|_{t=0}
\end{gather*}
holds for any $n\in {\mathbb{N}}_0$. It follows that
\begin{gather*}
\frac{d}{dt}\operatorname{arcsinh}
\Big(\frac{1{-}t}{1{+}t}\Big)=\frac{-\sqrt{2}}{(1{+}t)\sqrt{1{+}t^2}}
=-\sqrt{2}\Big[\sum_{m=0}^\infty (-t)^m\Big]\Big[\sum_{m=0}^\infty 
\underbrace{\binom{-1/2}{m}}_{=\big(\frac{-1}{4}\big)^m \big(\substack{2m\\m}\big)}t^{2m}\Big]\\
= -\sqrt{2} \sum_{m=0}^\infty \Big[\sum_{n=0}^m (-t)^{m-n}\binom{2n}{n}\Big(\frac{-1}{4}\Big)^nt^{2n} \Big]
= -\sqrt{2}\sum_{m=0}^\infty (-1)^mt^m\sum_{n=0}^m\binom{2n}{n}\Big(\frac{t}{4}\Big)^n
\end{gather*}
for any $|t|<1$ by taking the Cauchy product. Thus the $k$th derivative of 
$\operatorname{arcsinh}[(1{-}t)/(1{+}t)]$ at $t=0$ only consists of the coefficients with exponent $n+m=k-1$ of $t$. Explicitly,
\begin{align*}
\frac{d^k}{dt^k}\operatorname{arcsinh}\Big(\frac{1{-}t}{1{+}t}\Big)
\Big|_{t=0}&=-\sqrt{2}\underbrace{\frac{d^{k-1}}{dt^{k-1}}t^{k-1}\Big|_{t=0}}_{=(k-1)!}
\sum_{\substack{m=0\\0\leq k-m-1\leq m}}^\infty \hspace*{-15pt}(-1)^m
\begin{pmatrix}2k-2m-2\\{k-m-1}\end{pmatrix}\Big(\frac{1}{4}\Big)^{k-m-1}
\end{align*}
for all $k\in\mathbb N$ as $n\in\lbrace 0,\ldots,m\rbrace$. The condition $0\leq k-m-1\leq m$ 
translates to $m\leq k-1\leq 2m$, so $(k{-}1)/2\leq m\leq k-1$ and thus
\begin{align*}
\frac{d^k}{dt^k}\operatorname{arcsinh}\Big(\frac{1{-}t}{1{+}t}\Big)\Big|_{t=0}&=-\sqrt{2}(k{-}1)!
\sum_{m=\lceil \frac12(k{-}1)\rceil}^{k-1} (-1)^m\binom{2{k-2m-2}}{k-m-1}\Big(\frac{1}{4}\Big)^{k-m-1}\\
&=\sqrt{2}(-1)^k(k{-}1)!\sum_{m=0}^{\lfloor \frac12(k-1)\rfloor} (-1)^m\binom{2m}{m}\Big(\frac{1}{4}\Big)^{m},
\end{align*}
where the second equality follows by substituting $m$ with $k-1-m$.
One then obtains
\begin{align*}
[ \parity_{\textup{BJ}}]_{nn}&= \operatorname{arcsinh}(1)+
\sum_{k=1}^n \frac{1}{k!}\frac{d^k}{dt^k}\operatorname{arcsinh}\Big(\frac{1{-}t}{1{+}t}\Big)\Big|_{t=0}\\
&=\operatorname{arcsinh}(1)+\sqrt{2}\sum_{k=1}^n \frac{(-1)^k}{k} 
\sum_{m=0}^{\lfloor \frac12(k-1)\rfloor}\binom{2m}{m}\Big(\frac{-1}{4}\Big)^{m}.
\end{align*}
To get \eqref{eq:M_nn_finite_sum_1} we shift $k$ to $k+1$. 
Due to \eqref{eq:M_nn_finite_sum_1} and Lemma \ref{inf_sum_conv}, 
the limit $\lim_{n\to\infty}[ \parity_{\textup{BJ}}]_{nn}$ exists. Now consider $\operatorname{arcsinh}[(1{-}t)/(1{+}t)]$ 
and its Taylor series $\sum_{k=0}^\infty a_kt^k$ around $t_0=0$ for any $|t|<1$. By Lemma \ref{Mnn_generation_function},
\begin{align*}
\sum_{k=0}^\infty a_kt^k=\operatorname{arcsinh}\Big(\frac{1{-}t}{1{+}t}\Big)
&=(1{-}t)\frac{1}{1{-}t}\operatorname{arcsinh}\Big(\frac{1{-}t}{1{+}t}\Big)=\sum_{n=0}^\infty [ \parity_{\textup{BJ}}]_{nn}(1{-}t)t^n\\
&=[ \parity_{\textup{BJ}}]_{00}+\sum_{n=1}^\infty ([ \parity_{\textup{BJ}}]_{nn}-[ \parity_{\textup{BJ}}]_{(n-1)(n-1)}) t^n,
\end{align*}
thus $\sum_{k=0}^n a_k=[ \parity_{\textup{BJ}}]_{nn}$ for 
any $n\in\mathbb N_0$. By Lemma \ref{inf_sum_conv}, 
$\sum_{k=0}^\infty a_k=\lim_{n\to\infty}[ \parity_{\textup{BJ}}]_{nn}$ exists so 
Abel's theorem \cite[Th.~8.2]{hardy2015course} yields
$
\lim_{n\to\infty}[ \parity_{\textup{BJ}}]_{nn}=\sum_{k=0}^\infty a_k
=\lim_{t\to 1^-}\operatorname{arcsinh}[(1{-}t)/(1{+}t)]=\operatorname{arcsinh}(0)=0
$
as claimed.\qed
\end{proof}

\section{Direct Recursive Calculation of the Matrix Elements \label{app:directrecursive}}
The non-zero matrix elements  are defined by 
a set of rational numbers 
\begin{equation}
\label{BJoprtorationalnumbers}
M_{k\ell} := [ \parity_{\textup{BJ}} ]_{k+4\ell,k} / (\Gamma_{k\ell}) - \delta_{\ell0}\mathrm{arcsh}(1)/\sqrt{2},
\end{equation}
where the indexing $k,\ell \in \{0,1,2, \dots\}$ is now relative to the diagonal
(where $\ell=0$) and $\Gamma_{k\ell} := \gamma_{k+4\ell,k}= 2^{-2\ell + {1}/{2}} \sqrt{{k!}/{(k{+}4\ell)!}}$.
Here, $\delta_{nm}$ is the Kronecker delta and 
note the symmetry $[ \parity_{\textup{BJ}} ]_{k, k+4\ell}  = [ \parity_{\textup{BJ}} ]_{k+4\ell,k} $.
For example, the values $M_{k0}$
define the diagonal of the Born-Jordan parity operator  $[ \parity_{\textup{BJ}} ]_{kk}$
up to the constants $\Gamma_{k0} =  \sqrt{2}$ and $\mathrm{arcsh}(1)/\sqrt{2}$, 
compare to Fig.~\ref{BJmatrixPlot}.
These rational numbers appear to satisfy the following recursive relations
\begin{equation*}
M_{k{+}4,\ell} = 
\tfrac{1}{k+4} M_{k{+}3,\ell}
+
\tfrac{4 \ell+2 k+5}{(k+3) (k+4)}  M_{k{+}2,\ell} 
 +
\tfrac{4 \ell+k+2}{(k+3) (k+4)}  M_{k{+}1,\ell} 
+
\tfrac{(4 \ell+k+1) (4 \ell+k+2)}{(k+3) (k+4)} M_{k\ell},
\end{equation*}
i.e.~each element in a column is determined by the previous four values.
Calculating a column requires, however, the first four elements
$M_{0\ell}, M_{1\ell}, M_{2\ell}, M_{3\ell}$ as initial conditions.
Surprisingly, the first four rows appear to satisfy the following recursive
relations
\begin{align*}
M_{0,\ell+2} &= 4 [ (27+56 \ell+32 \ell^2) M_{0,\ell+1} - 16 \ell (1 {+} 4 \ell) (2 {+} 4 \ell) (3 {+} 4 \ell) M_{0\ell} ] \\
M_{1,\ell+2} &= 4 [ (39 + 72 \ell + 32 \ell^2) M_{1,\ell+1} - 16 \ell (2 {+} 4 \ell) (3 {+} 4 \ell) (5 {+} 4 \ell) M_{1\ell} ] \\
M_{2,\ell+2} &= 4 [ (55 + 88 \ell + 32 \ell^2)  M_{2,\ell+1} - 16 \ell (3 {+} 4 \ell) (5 {+} 4 \ell) (6 {+} 4 \ell)  M_{2\ell} ] \\
M_{3,\ell+2} &= 4 [ (75 + 104 \ell + 32 \ell^2) M_{3,\ell+1} - 16 \ell (5 {+} 4 \ell) (6 {+} 4 \ell) (7 {+} 4 \ell)  M_{3\ell} ].
\end{align*}
Ultimately, eight initial values
$M_{0,0}=0$, $M_{0,1}=4$, $M_{1,0}=-1$, $M_{1,1}=-8$, $M_{2,0}=-1/2$, $M_{2,1}=6$, $M_{3,0}=-2/3$, and $M_{3,1}=-4$
appear to determine the Born-Jordan parity operator via the above recursion relations for the elements $M_{k\ell}$.


\end{document}